\pdfoutput=1

\documentclass[english,11pt]{article}
\input{macros}

\newcommand{\FF}{\mathbb{F}}
\newcommand{\de}{\mathrm{d}}
\newcommand{\Sb}{\mathbf{S}}
\newcommand{\Ecr}{\mathcal{E}}
\newcommand{\Bcr}{\mathcal{B}}
\newcommand{\Ccr}{\mathcal{C}}
\newcommand{\Pcr}{\mathcal{P}}
\newcommand{\Ocr}{\mathcal{O}}
\newcommand{\Fcr}{\mathcal{F}}
\usepackage[cjk]{kotex}
\newcommand{\BlackBox}{\rule{1.5ex}{1.5ex}}
\usepackage{etoc}
\usepackage{authblk}
\etocdepthtag.toc{mtchapter}
\etocsettagdepth{mtchapter}{subsection}
\etocsettagdepth{mtappendix}{none}
\setlist[enumerate,1]{label=\arabic*., leftmargin=2em, labelsep=0.5em}

\begin{document}

\title{\Large Estimating Continuous Treatment Effects with Two-Stage Kernel Ridge Regression}

\author[1]{Seok-Jin Kim}
\author[2]{Kaizheng Wang}
\affil[1]{Department of IEOR, Columbia University}
\affil[2]{Department of IEOR and Data Science Institute, Columbia University}

\date{This version: \today}

\maketitle

\begin{abstract}
We study the problem of estimating the effect function for a continuous treatment, which maps each treatment value to a population-averaged outcome. A central challenge in this setting is confounding: treatment assignment often depends on covariates, creating selection bias that makes direct regression of the response on treatment unreliable. To address this issue, we propose a two-stage kernel ridge regression method. In the first stage, we learn a model for the response as a function of both treatment and covariates; in the second stage, we use this model to construct pseudo-outcomes that correct for distribution shift, and then fit a second model to estimate the treatment effect. Although the response varies with both treatment and covariates, the induced effect function obtained by averaging over covariates is typically much simpler, and our estimator adapts to this structure. Our optimal learning bounds are achieved without estimating the conditional treatment density, thereby bypassing a major bottleneck in existing methods. Furthermore, we introduce a fully data-driven model selection procedure that achieves provable adaptivity to both the unknown degree of overlap and the spectral decay of the underlying kernel.
\end{abstract}

\section{Introduction}

Estimating the causal effects of continuous treatments---such as drug dosages, training intensity, or varying policy levels---is a fundamental challenge in causal inference \citep{ishak2015simulation,kennedy2017non,colangelo2026double,zhang2025doubly}. The primary estimand of interest is the Treatment Effect Function (TEF), which maps each treatment value to the population-averaged outcome. A central difficulty in observational settings is confounding (or selection bias): treatment assignment often depends on complex, high-dimensional covariates. Consequently, direct regression of the outcome on the treatment yields biased estimates of the TEF. Because the TEF depends only on the treatment, a central goal is to derive learning bounds that adapt to its typically lower complexity rather than to that of the nuisance component---the mean outcome conditional on covariates and treatment.

State-of-the-art nonparametric approaches typically address this via two-stage ``debiasing'' procedures, which estimate nuisance components---most notably the conditional outcome mean and the treatment density, or generalized propensity score---to construct pseudo-outcomes \citep{kennedy2017non,colangelo2026double,bonvini2022fast}. 
A key bottleneck is the conditional treatment density: existing guarantees typically require estimating it, together with the outcome nuisance, in strong norms such as \(L^4\) or \(L^\infty\).
Providing such convergence is especially challenging when covariates are high-dimensional, discrete, or sparse. 
Moreover, estimating the conditional treatment density is often challenging and highly sensitive to weak overlap and model misspecification.
These challenges have motivated methods that construct weights using covariate-balance or independence criteria, reducing reliance on a fitted conditional-density model \citep{tubbicke2022entropy,HGC24,fong2018covariate}. However, these works do not establish minimax-optimal guarantees under weak overlap. 

In this work, we propose a kernel-based framework that overcomes these limitations. By leveraging the flexibility of Reproducing Kernel Hilbert Spaces (RKHS) \citep{scholkopf2002learning}, our method accommodates rich nonlinear interactions while naturally capturing the smoothness structures inherent in causal problems. We establish convergence rates that adapt to the intrinsic complexity of the TEF. Unlike prior debiasing methods, our approach avoids the instability associated with explicit conditional density estimation in high dimensions.
Furthermore, this adaptivity requires only that the nuisance outcome model---the mean response conditional on covariates and treatment---be correctly specified, with no explicit convergence-rate assumptions on nuisance estimation.
These features are especially valuable when the covariate distribution is poorly behaved, when conditional-density estimation is difficult, and when the problem lies in a weak-overlap regime.

We consider a variety of performance measures that may be chosen according to the user's preferences, and obtain minimax-optimal results in all cases. We first study the continuous mean integrated squared error (MISE), defined with respect to continuous weights, and also consider pointwise evaluation for estimating the treatment effect at a fixed treatment level. In addition, we consider a treatment selection rule that maximizes the treatment effect and derive a regret bound for the resulting decision rule.

Our main contributions can be summarized as follows:

\begin{itemize}
\item \textbf{Methodology:} We propose a two-stage kernel ridge regression (KRR) estimator. The first stage learns the conditional outcome surface, and the second stage regresses the estimated pseudo-outcomes on the treatment variable. To ensure practical applicability, we develop a fully data-driven model selection procedure that eliminates the need for manual tuning of regularizers. 

\item \textbf{Theory:} We establish minimax-optimal error bounds for our estimator across multiple evaluation targets: continuous MISE, point evaluation at a fixed treatment value, and regret for treatment optimization.
Our learning bound adapts to the complexity of the TEF.
The procedure does not estimate the conditional treatment density and imposes no learning-rate requirement on nuisance estimation beyond correct specification of the function class.
For continuous MISE, we further prove adaptivity to unknown structural parameters, namely the degree of overlap and the spectral decay of the kernel. 
\end{itemize}

\paragraph{Related Work.}

Early studies on continuous treatment effects focused primarily on parametric models \citep{robins2000marginal, laan2003unified, diaz2013targeted,imbens2000role,hirano2004propensity}, while recent advances have shifted toward flexible nonparametric estimation \citep{kennedy2017non, semenova2021debiased,bonvini2022fast,colangelo2026double,takatsu2025debiased}. The dominant paradigm in this literature combines orthogonalization (debiasing) with Hölder function classes. 
However, these methods rely on estimating the conditional treatment density and typically require convergence guarantees for the nuisance estimators in norms stronger than \(L^2\), such as \(L^4\) or \(L^\infty\), which become increasingly difficult to obtain as the covariate dimension grows.

RKHS-based methods offer a powerful alternative for structural adaptation in causal inference \citep{nie2021quasi,mou2023kernel,singh2024kernel,hirshberg2021augmented,kim2025transfer}. \citet{mou2023kernel} study kernel methods for policy evaluation in contextual bandits, but their
setting is essentially confounding-free rather than observational.  Most relevant to our work, \citet{singh2024kernel} analyze continuous treatment effect estimators using tensor-product kernels under source conditions on the response function. Their guarantees are formulated via assumptions on the joint tensor-product space, and it remains unclear whether those rates adapt to the spectrum of the treatment effect space alone. In contrast, we consider general kernels and dispense with the source condition requirement, thereby offering broader applicability.

Our upper bound is in fact faster than the standard double machine learning (DML) rate, which helps explain why conditional-density estimation can be avoided in our framework. There is a line of work that likewise avoids explicit conditional treatment-density estimation by achieving faster rates than standard DML, including \citet{bonvini2022fast} and \citet{HGC24}. These results, however, are restricted to Sobolev- or H\"older-type function classes, require strong assumptions on the nuisance structure, or rely on computationally expensive procedures. We return to a more detailed comparison with these works after presenting our main results.

A smaller decision-oriented literature connects continuous-treatment response curves to treatment choice. Closest in spirit among early work, \citet{flores2007dose} studies estimation of dose-response functions and optimal doses. A related but distinct line studies continuous-action policy learning, where the target is typically a covariate-dependent rule maximizing a welfare criterion, and regret is evaluated relative to a policy class \citep{kallus2018policy,demirer2019semiparametric,ai2024datadriven}. We give a more detailed comparison after presenting our regret result in Section~\ref{subsection: regret dose optimization}.

Finally, the concept of two-stage regression has appeared in semi-supervised learning with auxiliary covariates \citep{liu2023augmented,XWa24}. However, these works address different problems and do not analyze TEF estimation or the adaptation phenomena that arise in continuous treatment settings.

An earlier version of this work \citep{kim2026estimating} appeared in the Forty-Third International Conference on Machine Learning (ICML 2026). The current version adds minimax-optimal guarantees for point evaluation under local overlap (Section~\ref{subsection: point evaluation error}) and treatment-selection regret (Section~\ref{subsection: regret dose optimization}), together with matching lower bounds (Theorem~\ref{theorem: lower bound point}). These results extend the continuous-MISE theory of the conference version and do not require the additional source condition imposed there for its \(L^\infty\) bound.

\paragraph{Outline.}
Section~\ref{section:problem-setup} introduces the estimands, RKHS framework, and assumptions. Section~\ref{section: methodology} presents the two-stage estimator and model selection procedure, and Section~\ref{section:general theory} establishes their general error bounds. Section~\ref{section: theory} specializes these results to continuous MISE and proves adaptivity, while Section~\ref{section:point evaluation} studies point evaluation and treatment-selection regret; each section also presents the corresponding lower bounds. Comparisons with related work accompany the relevant results. Section~\ref{section:experiments} reports experiments, and Section~\ref{section:conclusion} concludes. Proofs are provided in the supplementary materials.

\paragraph{Notation.}
The constants $c_1, c_2, C_1, C_2, \dots$ may differ from line to line. We use the symbol $[n]$ as a shorthand for $\{ 1, 2, \dots, n \}$. For nonnegative sequences $\{ a_n \}_{n=1}^{\infty}$ and $\{ b_n \}_{n=1}^{\infty}$, we write $a_n \lesssim b_n$ or $a_n = \cO(b_n)$ if there exists a positive constant $C$ such that $a_n \leq C b_n$ for all $n$. We use $\widetilde{\cO}$ to hide polylogarithmic factors. In addition, we write $a_n \asymp b_n$ if $a_n \lesssim b_n$ and $b_n \lesssim a_n$. For a bounded linear operator $\bA$, we use $\| \bA \|_{\mathrm{op}}$ to refer to its operator norm. For any $\bu$ and $\bv$ in a Hilbert space $\HH$, their inner and outer products are denoted by $\langle \bu, \bv \rangle_{\HH}$ and $\bu \otimes \bv$, respectively. We use $\cN(\bmu, \bSigma)$ to represent the Gaussian distribution with mean $\bmu$ and covariance matrix $\bSigma$, and $\cU (S)$ for the uniform distribution on a set \(S\).

\section{Problem Setup}\label{section:problem-setup}

\subsection{Formulation and Objective}

Let \(\cX\) denote the covariate space and \(\cA\) the domain of continuous treatments (e.g., drug dosage). Let \(X \in \cX\) and \(A \in \cA\) denote generic random covariates and treatment, respectively, and write \(\cP_{X,A}\) for the joint law of \((X,A)\), with marginal \(\cP_X\). We set \(\cZ := \cX \times \cA\).
Under the potential outcomes framework \citep{robins2000marginal}, a generic unit is endowed with the collection \(\{Y(a)\}_{a \in \cA}\), where \(Y(a)\) is the potential outcome under treatment level \(a\). The observed outcome is \(Y := Y(A)\).
We observe an independent and identically distributed (i.i.d.) sample \(\cD = \{(x_i, a_i, y_i)\}_{i=1}^n\) drawn from the law of \((X, A, Y)\), and write \(z_i := (x_i, a_i)\).

We define the conditional mean response by
\[
f^\star(x, a) := \EE[Y \mid X=x, A=a].
\]
Our target estimand is the \emph{Treatment Effect Function (TEF)},
\begin{equation*}
h^\star(a) := \EE [Y(a)]
\end{equation*}
which maps each treatment value to the population-averaged outcome. We refer to \(f^\star\) as the \emph{nuisance component}: it is needed to identify and estimate \(h^\star\), but is not itself our target estimand.
Consistency and no unmeasured confounding (Assumption~\ref{ass:id} in Section~\ref{subsection: assumptions}) yield the identification formula
\begin{equation*}
h^\star(a) = \EE_{X \sim \cP_X}[f^\star(X,a)].
\end{equation*}

A fundamental challenge in this setting is confounding (or selection bias). In observational data, treatment assignment typically depends on covariates; consequently, the conditional distribution of covariates given treatment, denoted by $\cP_{X|A=a}$, varies with $a$ and generally differs from the marginal distribution $\cP_X$.
Consider the function obtained by directly regressing the observed outcome \(Y\) on the treatment \(A\):
\begin{equation*}
g(a) = \EE [Y \mid A=a] = \EE_{x \sim \cP_{X|A=a}} [f^\star(x, a)].
\end{equation*}
Because the integration measure $\cP_{X|A=a}$ changes with $a$, the function $g(a)$ conflates the true treatment effect with the effect of the shifting covariate distribution. As a result, $g(a)$ generally differs from $h^\star(a)$, and recovering $h^\star$ necessitates adjusting for this distribution shift.

Our objective is to construct an estimator $\hat{h}$ whose error is small under a user-specified reference distribution $\cP_{\operatorname{ref}}$ on $\cA$:
\begin{equation}\label{eq:mise}
\mathcal{E}(\hat{h}) := \int_{\mathcal{A}} \left( \hat{h}(a) - h^\star(a) \right)^2 \mathrm{d}\cP_{\operatorname{ref}}(a).
\end{equation}
This is the squared \(L^2(\cP_{\operatorname{ref}})\) error of the estimated TEF.
The reference distribution reflects the region, or resolution, at which the treatment effect curve is to be evaluated.
When \(\cP_{\operatorname{ref}}\) has a continuous density on \(\cA\), \eqref{eq:mise} becomes the usual mean integrated squared error (MISE) for estimating the full dose-response curve; we refer to this regime as \emph{continuous MISE}.
Integrated \(L^2\)-error criteria are standard in nonparametric function estimation \citep{tsybakov2009nonparametric} because they capture the global quality of the estimated function.
They are also widely used in continuous-treatment settings: \citet{schwab2020learning} use an integrated-error criterion for dose-response estimation, and \citet{kennedy2017non} evaluate the estimated curve using integrated error in their experiments.
\(L^2\)-type criteria are also standard in binary-treatment settings; see \citet{nie2021quasi} and \citet{foster2023orthogonal}.

The same formulation also covers more localized targets.
If \(\cP_{\operatorname{ref}}=\delta(a_0)\) is a Dirac measure at a fixed treatment value \(a_0\), then \(\mathcal{E}(\hat h)=(\hat h(a_0)-h^\star(a_0))^2\), so the risk becomes point evaluation error.
This point-mass reference distribution formalizes the goal of accurately estimating the causal effect at a prespecified treatment level, such as a clinically relevant dose or a policy-relevant intervention intensity.
Pointwise dose-response estimation and inference are common in continuous-treatment work \citep{kennedy2017non,colangelo2026double}, and \(\delta(a_0)\) is the natural reference measure for this target.

Beyond estimating the curve itself, we also study the quality of decisions induced by \(\hat h\).
Let \(\hat a \in \cA\) denote a data-driven estimate of an optimal treatment value.
We evaluate this treatment recommendation through its simple regret,
\begin{equation*}
\operatorname{Reg}(\hat a)
:=
\sup_{a\in\cA} h^\star(a)-h^\star(\hat a),
\end{equation*}
which measures the loss from deploying the estimated recommendation rather than an optimal treatment value under the true TEF.

\subsection{Reproducing Kernel Hilbert Space (RKHS) Framework}

We assume that both the conditional response function $f^\star$ and the treatment effect function $h^\star$ reside in RKHSs, denoted by $\mathcal{F}$ and $\mathcal{H}$, defined on domains $\mathcal{Z}$ and $\mathcal{A}$, respectively.
Let $K_{\mathcal{F}}$ and $K_{\mathcal{H}}$ be their associated positive semidefinite kernels.
By standard theory \citep{aronszajn1950theory}, there exist Hilbert spaces \(\FF\) and \(\HH\), along with feature maps $\psi: \cZ \to \FF$ and $\phi: \cA \to \HH$, satisfying the reproducing properties: \(\langle \psi(z), \psi(z')\rangle_{\FF} = K_\cF(z,z')\) and \(\langle \phi(a), \phi(a')\rangle_{\HH} = K_\cH(a,a')\).
Accordingly, the hypothesis spaces are defined as:
\begin{equation*}
\begin{aligned}
\mathcal{F} &= \{ f_\theta(\cdot) = \langle \psi(\cdot), \theta \rangle_{\FF} \mid \theta \in \FF \}, \\
\mathcal{H} &= \{ h_\eta(\cdot) = \langle \phi(\cdot), \eta \rangle_{\HH} \mid \eta \in \HH \}.
\end{aligned}
\end{equation*}
Here, \(\mathcal{F}\) and \(\mathcal{H}\) denote RKHSs of functions, whereas \(\FF\) and \(\HH\) denote the coefficient Hilbert spaces associated with the feature maps \(\psi\) and \(\phi\).

A central motivation for our framework is the observation that $h^\star$, being a marginal average of $f^\star$ over the covariate distribution $\cP_X$, typically exhibits lower structural complexity than the nuisance component $f^\star$ since it is a function of lower dimension.
Our method leverages this distinction through flexible modeling. To illustrate, consider the case where $\cZ \subset \RR^d$ and $\cA \subset \RR^p$ are bounded regular domains.

\begin{itemize}
\item \textit{Sobolev smoothness.} The standard Sobolev space $H^\beta(\cZ)$ (with $\beta > d / 2$) and the mixed Sobolev space $H^\beta_{\mathrm{mix}}(\cZ)$ (with $\beta >1/2$) are known to admit RKHS structures \citep{doumeche2025fast,zhang2023optimality,dick2007construction,kuhn2015approximation,suzuki2019adaptivity}.
If \(f^\star \in H^\beta(\cZ)\) or \(f^\star \in H^\beta_{\mathrm{mix}}(\cZ)\), then the target function \(h^\star\) inherits the same degree of \(\beta\)-smoothness on \(\cA\); that is, \(h^\star \in H^\beta(\cA)\) or \(h^\star \in H^\beta_{\mathrm{mix}}(\cA)\), respectively.
Since the treatment dimension $p$ is typically much smaller than the dimension $d$ of the joint covariate-treatment space $\cZ$ ($p \ll d$), estimating $h^\star$ directly in $\mathcal{H}$ is statistically easier than estimating $f^\star$ in $\mathcal{F}$.

\item \textit{High-dimensional covariates.} 
Our framework allows for modeling $\mathcal{F}$ using kernels designed for high-dimensional data, such as Neural Tangent Kernels (NTK) \citep{ghorbani2021linearized,bietti2021deep} or inner product kernels \citep{mei2022generalization,liang2020just,mei2021learning}, to capture complex nuisance interactions.
Simultaneously, one can model the simpler dose-response curve in $\mathcal{H}$ using a standard smooth kernel (e.g., Matérn), thereby decoupling the complexity of the nuisance from the target.
\end{itemize}

\subsection{Assumptions}\label{subsection: assumptions}

We posit the following standard assumptions to facilitate our theoretical analysis.

\begin{assumption}[Identification]\label{ass:id}
(i) Consistency: \(Y = Y(A)\) almost surely. (ii) No unmeasured confounding: \(\{Y(a)\}_{a \in \cA}\) is independent of \(A\) conditional on \(X\).
\end{assumption}

\begin{assumption}[Sub-Gaussian Noise]\leavevmode\phantomsection\label{ass:noise}
The noise \(\varepsilon := Y - f^\star(X, A)\) is \(\sigma\)-sub-Gaussian given \((X, A)\), that is, \(\EE[e^{t \varepsilon} \mid X, A] \le e^{t^2 \sigma^2 / 2}\) for all \(t \in \RR\).
We assume $\sigma$ is bounded by an absolute constant.
\end{assumption}

\begin{assumption}[Boundedness]\label{ass:bounded}
The kernels are uniformly bounded by a constant $\xi$:
\[
\sup_{z \in \cZ} K_\cF(z, z) \le \xi
\quad\text{and}\quad
\sup_{a \in \cA} K_\cH(a, a) \le \xi.
\]
Additionally, we assume $\|h^\star\|_\cH$ is bounded by a constant.
\end{assumption}

Assumption~\ref{ass:id} is standard in the causal inference literature \citep{bonvini2022fast,kennedy2017non,kunzel2019metalearners,curth2021nonparametric}.
Under Assumption~\ref{ass:id}, we have \(f^\star(x, a) = \EE[Y(a) \mid X=x]\), and therefore \(h^\star(a) = \EE_{x \sim \cP_X}[f^\star(x, a)]\).
Assumption~\ref{ass:noise} allows unbounded outcomes with sub-Gaussian tails. For example, \citet{bonvini2022fast,kennedy2017non,nie2021quasi} assume bounded outcomes, which imply sub-Gaussian noise and are therefore stronger than our assumption.
The boundedness of the kernels in Assumption~\ref{ass:bounded} is common in the analysis of kernel methods \citep{wainwright2019high,singh2024kernel}. 

\section{Methodology}\label{section: methodology}

This section introduces the two-stage kernel ridge regression (KRR) estimator and the associated data-driven model selection procedure. Section~\ref{subsection: estimation algorithm} presents the estimation algorithm, and Section~\ref{subsection: model selection algorithm} describes how we select the second-stage regularizer.

\subsection{Estimation Procedure}\label{subsection: estimation algorithm}

We propose a two-stage KRR framework designed to decouple the estimation of the complex conditional response surface $f^\star$ from the smoothing of the simpler treatment effect function $h^\star$.

The procedure consists of three steps: (1) learning the full conditional mean $f^\star$ via KRR; (2) constructing ``pseudo-outcomes'' by empirically marginalizing the estimated nuisance function over the covariate distribution; and (3) regressing these pseudo-outcomes on treatment values to recover $h^\star$. A summary is provided in Algorithm~\ref{algorithm: main}.

\paragraph{Step 1: Nuisance Estimation.}
We first estimate the conditional mean function $f^\star: \cZ \to \RR$ using KRR on the joint space $\cZ = \cX \times \cA$.
Recall that $\cD = \{(x_i, a_i, y_i)\}_{i=1}^n$ denotes the observed data.
We solve the following optimization problem over the hypothesis space $\cF$:
\begin{equation}\label{equation: nuisance KRR}
\hat{f} := \argmin_{f \in \cF} \bigg\{ \frac{1}{n}\sum_{i=1}^n (y_i - f(x_i, a_i))^2 + \lambda_0\|f\|_\cF^2 \bigg\}.
\end{equation}
By the representer theorem \citep{scholkopf2002learning}, the solution has the finite-dimensional representation \(\hat f(z)=\sum_{i=1}^n\alpha_i K_\cF(z_i,z)\), where \((\alpha_1, \dots \alpha_n) =(\mathbf K_\cF+n\lambda_0\mathbf I_n)^{-1}\mathbf y\), \((\mathbf K_\cF)_{ij}=K_\cF(z_i,z_j)\), and \(\mathbf y=(y_1,\ldots,y_n)^\top\). Thus, fitting \(\hat f\) requires only the kernel matrix and a linear-system solve. We use a small \textit{nuisance regularizer} $\lambda_0 \asymp n^{-1} \log n$ because the first stage is intended primarily to control bias and recover the conditional surface accurately; the final bias-variance trade-off is handled in the second stage.

\paragraph{Step 2: Empirical Marginalization.}
To recover the target $h^\star(a) = \EE_{X \sim \cP_X}[f^\star(X, a)]$, we construct a synthetic dataset. We sample $n$ treatment values $\{a'_j\}_{j=1}^n$ i.i.d. from a user-specified sampling distribution $\cP_{\operatorname{samp}}$ (e.g., the uniform distribution on $\cA$), independently of \(\cD\).
This distribution is an algorithmic design choice that determines where the second-stage regression learns the curve, whereas \(\cP_{\operatorname{ref}}\) determines where its error is evaluated.
The number of query points is an algorithmic choice and need not equal the observational sample size; we use \(n\) for simplicity in both the algorithm and the analysis. For each $a'_j$, we compute the \emph{pseudo-outcome} $m_j$ by averaging the fitted nuisance function $\hat{f}$ over the empirical distribution of covariates:
\begin{equation}\label{equation: pseudo-outcome}
m_j := \frac{1}{n} \sum_{i=1}^n \hat{f}(x_i, a'_j).
\end{equation}
Here, $m_j$ serves as a noisy proxy for $h^\star(a'_j)$.
For continuous MISE, we set \(\cP_{\operatorname{samp}}=\cP_{\operatorname{ref}}\).
For point evaluation at \(a_0\), however, \(\cP_{\operatorname{ref}}=\delta_{a_0}\). Sampling from this point mass would place every query at \(a_0\) and repeat the same pseudo-outcome, leaving the second-stage regression no variation in treatment with which to exploit smoothness.
We therefore sample uniformly from a neighborhood of \(a_0\), allowing the regression to use nearby treatment values, as specified in Section~\ref{section:point evaluation}.
This choice must still be supported by the observed treatment coverage; Assumption~\ref{assumption: first overlap} formalizes this requirement.

\paragraph{Step 3: Pseudo-Outcome Regression.}
Finally, we treat the constructed set $\cD' = \{(a'_j, m_j)\}_{j=1}^n$ as our dataset for the target function. We perform a second KRR in the simpler space $\cH$ to smooth the pseudo-outcomes:
\begin{equation}\label{equation: second KRR}
\hat{h}_\lambda := \argmin_{h \in \cH} \bigg\{\frac{1}{n}\sum_{j=1}^n (m_j - h(a'_j))^2 + \lambda \|h\|_\cH^2 \bigg\}.
\end{equation}
The \textit{main regularizer} $\lambda$ controls the final smoothness of the TEF. 
We propose a data-driven procedure to select the optimal $\lambda$ in Section~\ref{subsection: model selection algorithm}.

\begin{algorithm}[h]
\caption{TEF Estimation via Two-Stage KRR}
\label{algorithm: main}
\begin{algorithmic}[1]
\REQUIRE Dataset $\cD = \{(x_i, a_i, y_i)\}_{i=1}^n$, sampling distribution $\cP_{\operatorname{samp}}$, nuisance regularizer $\lambda_0$, main regularizer \(\lambda\).
\STATE \textbf{Stage 1:} Compute $\hat{f}$ by solving the KRR problem in \eqref{equation: nuisance KRR} using the raw data \(\cD\) and regularizer \(\lambda_0\).
\STATE \textbf{Stage 2:} Sample query points $a'_1, \dots, a'_n \overset{\text{i.i.d.}}{\sim} \cP_{\operatorname{samp}}$.
\STATE Compute pseudo-outcomes: $m_j \leftarrow \frac{1}{n} \sum_{i=1}^n \hat{f}(x_i, a'_j)$ for $j=1, \dots, n$.
\STATE \textbf{Stage 3:} Compute $\hat{h}_\lambda$ by solving the KRR problem in \eqref{equation: second KRR} using the synthetic data $\{(a'_j, m_j)\}_{j=1}^n$ and regularizer \(\lambda\).

\item[\textbf{Output:}] Estimator $\hat{h}_\lambda$.
\end{algorithmic}
\end{algorithm}

\paragraph{Treatment Selection.}
The same estimator can also be used for dose optimization.
Given \(\hat h_\lambda\), a natural plug-in rule selects
\(\hat a_\lambda\in\argmax_{a\in\cA}\hat h_\lambda(a)\).
We later show that this simple greedy rule enjoys an optimal bound on the simple regret under a pointwise control developed for \(\hat h_\lambda\).

\subsection{Model Selection Procedure}\label{subsection: model selection algorithm}
While the nuisance regularizer $\lambda_0$ in Algorithm~\ref{algorithm: main} follows a concrete theoretical guideline ($\lambda_0 \asymp n^{-1}\log n$), the main regularizer $\lambda$ requires careful tuning.
Standard cross-validation is not directly applicable in this setting because we cannot directly observe the target label \(h^\star(a)\) or an unbiased noisy version of it. Instead, we only obtain noisy outcomes \(y_i\) that are confounded by covariates.
For continuous MISE, we address this challenge by introducing a fully data-driven model selection procedure based on \textit{proxy validation}.
The core idea is to construct a nearly unbiased (albeit noisy) ``proxy'' for the ground truth via data splitting, allowing us to estimate the generalization error of the candidate models.
Crucially, this procedure provably adapts to unknown structural parameters without requiring manual tuning.

For simplicity, we assume that the sample size \(n\) is even and randomly split \(\cD\) into two equally sized subsets \(\cD_{\text{train}}\) and \(\cD_{\text{val}}\) for training and validation, respectively. We use the same sampling distribution \(\cP_{\operatorname{samp}}\) and nuisance regularizer \(\lambda_0\) in every call to Algorithm~\ref{algorithm: main}; each call uses \(n/2\) independently sampled query treatments, matching its input sample size.
\begin{enumerate}
\item \textbf{Training ($\cD_{\text{train}}$):} On the training set, we run Algorithm~\ref{algorithm: main} with \(\cP_{\operatorname{samp}}\), \(\lambda_0\), and each main regularizer \(\lambda\in\Lambda\), generating a set of candidate estimators $\mathcal{M} = \{ \hat{h}_\lambda \mid \lambda \in \Lambda \}$.
\item \textbf{Proxy Construction ($\cD_{\text{val}}$):} On the validation set, we run Algorithm~\ref{algorithm: main} with \(\cP_{\operatorname{samp}}\), \(\lambda_0\), and a fixed main regularizer \(\tilde\lambda \asymp \log n/n\). Let \(\tilde h\) denote the resulting estimator.
\item \textbf{Selection:} We select the candidate $\hat{h}_\lambda \in \mathcal{M}$ that is closest to the proxy $\tilde{h}$ in terms of the $L^2(\cP_{\operatorname{ref}})$ distance. In practice, we approximate this distance by Monte Carlo integration, sampling query points from the reference distribution $\cP_{\operatorname{ref}}$.
\end{enumerate}

The formal procedure is detailed in Algorithm~\ref{algorithm: model selection}.

\begin{algorithm}[h]
\caption{Data-Driven Model Selection for TEF}
\label{algorithm: model selection}
\begin{algorithmic}[1]
\REQUIRE Dataset $\cD$ of even size \(n\), sampling distribution \(\cP_{\operatorname{samp}}\), reference distribution $\cP_{\operatorname{ref}}$, candidate grid $\Lambda$, nuisance regularizer $\lambda_0$, proxy regularizer \(\tilde\lambda\), number of Monte Carlo samples $M$.
\STATE Randomly split $\cD$ into $\cD_{\text{train}}$ and $\cD_{\text{val}}$ (sizes $n/2$).
\STATE \textbf{Train Candidates:} For each \(\lambda\in\Lambda\), run Algorithm~\ref{algorithm: main} with inputs \((\cD_{\text{train}},\cP_{\operatorname{samp}},\lambda_0,\lambda)\) to obtain \(\hat h_\lambda\).
\STATE \textbf{Train Proxy:} Run Algorithm~\ref{algorithm: main} with inputs \((\cD_{\text{val}},\cP_{\operatorname{samp}},\lambda_0,\tilde\lambda)\) and independent query treatments to obtain \(\tilde h\).
\STATE \textbf{Evaluate:} Independently sample test points $\tilde{a}_1, \dots, \tilde{a}_M \overset{\text{i.i.d.}}{\sim} \cP_{\operatorname{ref}}$.
\STATE Select $\hat{\lambda}$ by minimizing the distance to the proxy:
\begin{equation*}
\hat{\lambda} \in \argmin_{\lambda \in \Lambda} \sum_{k=1}^{M} \left( \hat{h}_\lambda(\tilde{a}_k) - \tilde{h}(\tilde{a}_k) \right)^2.
\end{equation*}
\item[\textbf{Output:}] Selected estimator $\hat{h}_{\hat{\lambda}}$.
\end{algorithmic}
\end{algorithm}

\section{General Theory}\label{section:general theory}

Our goal is to efficiently bound the MISE under the reference distribution.
The previous section introduced the estimation procedure in Algorithm~\ref{algorithm: main} and the regularizer selection procedure in Algorithm~\ref{algorithm: model selection}.
This section presents MISE guarantees for these algorithms.
Algorithm~\ref{algorithm: main} uses the sampling distribution \(\mathcal{P}_{\operatorname{samp}}\), whereas the MISE is evaluated under the reference distribution \(\mathcal{P}_{\operatorname{ref}}\).
Because the available information comes from the observational data distribution \(\cP_{X,A}\), we introduce two overlap measures that measure the efficiency with which information can be transferred from the observational distribution to the sampling and evaluation distributions.
We define these measures next.

\subsection{Overlap Measures in the General Theory}

We introduce two overlap measures for the general theory.

\paragraph{Treatment Coverage Ratio.}
The \emph{treatment coverage ratio} \(R>1\) measures how well the conditional treatment distribution covers the sampling distribution.
It controls the difficulty of constructing pseudo-outcomes at treatment values drawn from \(\cP_{\operatorname{samp}}\) using observational data generated from \(\cP_{A \mid X=x}\).
The formal condition is as follows.

\begin{assumption}[Treatment Coverage Ratio]\leavevmode\phantomsection\label{assumption: first overlap}
There exists \(R>1\) such that for all \(x \in \cX\),
\begin{equation*}
\frac{\de \cP_{\operatorname{samp}}}{\de \cP_{A \mid X = x}}\le  R.
\end{equation*}
\end{assumption}
The distribution \(\cP_{\operatorname{samp}}\) is chosen by the user, and a larger \(R\) means weaker coverage by the observational data. For example, if \(\cP_{\operatorname{samp}}\) is uniform on a region \(S\subseteq\cA\) of positive volume and \(p_{A\mid X=x}(a)\ge\kappa\) on \(S\) for every \(x\), then one may take any \(R>(\kappa|S|)^{-1}\). The assumption requires coverage only on the support of the sampling distribution, which permits the localized sampling used for point evaluation, as detailed in Section~\ref{subsection: point evaluation error}.

\paragraph{Sampling Coverage Function.}
The performance measure, MISE, is evaluated under the reference distribution \(\cP_{\operatorname{ref}}\), which need not coincide with \(\cP_{\operatorname{samp}}\).
We therefore introduce the \emph{sampling coverage function} \(B(\lambda)\), which measures how well the sampling distribution covers the reference distribution at regularization level \(\lambda\) through their second-moment operators.
First, define the second-moment operator under the sampling distribution as
\begin{equation*}
\bSigma_{\operatorname{samp}} :=\EE_{a \sim \cP_{\operatorname{samp}}}[\phi(a)\otimes \phi(a)].
\end{equation*}
Similarly, define the second-moment operator under the reference distribution by
\[
\bSigma_{\operatorname{ref}} :=  \EE_{a \sim \cP_{\operatorname{ref}}}[\phi(a)\otimes \phi(a)].
\]
For point evaluation at a fixed \(a_0\), this reduces to the rank-one operator \(\bSigma_{\operatorname{ref}} = \phi(a_0)\otimes \phi(a_0)\).
Finally, we compare \(\bSigma_{\operatorname{ref}}\) to the regularized operator \(\bSigma_{\operatorname{samp}}+\lambda \Ib\).
The formal condition is as follows.
\begin{assumption}[Sampling Coverage Function]\label{assumption: second overlap}
Suppose there exists a non-increasing function \(B:(0,\infty)\to [1,\infty)\) such that
\begin{equation*}
\bSigma_{\operatorname{ref}} \preceq B(\lambda) (\bSigma_{\operatorname{samp}} + \lambda \Ib)
\end{equation*}
for any \(\lambda >0\).
\end{assumption}
Equivalently, the assumption bounds \(\|h\|_{L^2(\cP_{\operatorname{ref}})}^2\) by \(B(\lambda)\{\|h\|_{L^2(\cP_{\operatorname{samp}})}^2+\lambda\|h\|_\cH^2\}\) for every \(h\in\cH\). The regularization term permits this comparison even when a density ratio between the two distributions does not exist. It is natural for \(B(\lambda)\) to decrease as \(\lambda\) increases, because the regularized operator on the right becomes larger.
If \(\cP_{\operatorname{samp}} = \cP_{\operatorname{ref}}\), then the condition holds with \(B(\cdot)\equiv 1\). More generally, a bounded density ratio \(\mathrm{d}\cP_{\operatorname{ref}}/\mathrm{d}\cP_{\operatorname{samp}}\le b\) gives \(B(\lambda)\equiv\max\{1,b\}\).
For the Sobolev class \(\cH=H^\beta(\cA)\) with \(\cA\subset\RR^p\), when the reference distribution is a point mass and the sampling distribution is uniform, one obtains
\(B(\lambda)\lesssim \lambda^{-\frac{p}{2\beta}}\) for \(\lambda\le 1\).
This point evaluation comparison is derived in the local form needed for our analysis in Lemma~\ref{lemma: Sobolev 1}.

\subsection{Upper Bound for Estimation Procedure (Algorithm~\ref{algorithm: main})}
We now state a general upper bound for the reference MISE for our estimator $\hat{h}_\lambda$ obtained by running Algorithm~\ref{algorithm: main}.
For any \(\lambda>0\), we define the \emph{effective dimension} of the reference distribution as
\begin{equation*}
\Gamma(\lambda) := \Tr \left((\bSigma_{\operatorname{ref}} + \lambda \Ib)^{-1} \bSigma_{\operatorname{ref}} \right).
\end{equation*}
This quantity depends only on the geometry of the target RKHS \(\cH\) under \(\cP_{\operatorname{ref}}\). For example, under \(\ell\)-polynomial spectral decay (formalized in Section~\ref{section: theory}), standard effective-dimension bounds in \citet{wainwright2019high} give
\begin{equation*}
\Gamma(\lambda) \lesssim \lambda^{-\frac{1}{2\ell}}.
\end{equation*}
The key point is that \(\Gamma(\lambda)\) is defined through the simpler target RKHS \(\cH\), rather than through the nuisance class \(\cF\).

We are now in a position to state the main upper bound theorem.
\begin{theorem}[Unified Upper Bound]\label{theorem: general theory}
Suppose Assumptions~\ref{ass:id}, \ref{ass:noise}, \ref{ass:bounded}, \ref{assumption: first overlap}, and~\ref{assumption: second overlap} hold and we run Algorithm~\ref{algorithm: main} with nuisance regularizer \(\lambda_0\asymp n^{-1}\log n\) and main regularizer \(\lambda \gtrsim n^{-1}\log n\).
With probability at least \(1-n^{-10}\), the reference MISE satisfies
\begin{equation*}
\cE_{\operatorname{ref}} (\hat h_\lambda) \lesssim B(\lambda) \cdot \bigg(\underbrace{\frac{R }{n} \cdot \Gamma\big(\lambda B(\lambda)\big)}_{\text{Variance}} + \underbrace{{\lambda} \|h^\star\|_\cH^2}_{\text{Bias}} + \underbrace{\frac{R }{n} \|f^\star\|^2_{\cF}}_{\text{Nuisance error}} \bigg).
\end{equation*}
\end{theorem}
The proof is deferred to Appendix~\ref{section:apdx proof general theory}.
Theorem~\ref{theorem: general theory} separates the risk into three interpretable components: a bias term incurred in the second KRR in Algorithm~\ref{algorithm: main}, a variance term governed by the effective dimension of the target space \(\cH\), and a nuisance term that depends on \(\|f^\star\|_\cF\).
This is the core decoupling phenomenon in our analysis: the two leading terms (variance and bias) are determined by the complexity of \(h^\star\), rather than by the complexity of the full response surface \(f^\star\). In Section~\ref{section: theory}, we show that the nuisance term is often negligible compared with these two terms.
The factor \(R\) quantifies the cost of querying pseudo-outcomes under \(\cP_{\operatorname{samp}}\), whereas \(B(\lambda)\) quantifies the additional price of evaluating error under \(\cP_{\operatorname{ref}}\).

\paragraph{Proof Idea.}
The key analytical step is to retain the randomness of \(\hat f\) as it propagates through the pseudo-outcomes, instead of bounding it solely by a deterministic nuisance-estimation error. First-stage undersmoothing controls the propagated bias, while the second KRR controls the propagated variance through the target kernel. This produces the bias, variance, and nuisance terms in Theorem~\ref{theorem: general theory}. The argument is tailored to empirical marginalization under confounding, which distinguishes it from two-stage regression with auxiliary covariates in semi-supervised learning \citep{liu2023augmented,XWa24}.

We next present an oracle inequality for the model selected by Algorithm~\ref{algorithm: model selection}.
After establishing the model-selection analogue, we specialize Theorem~\ref{theorem: general theory} to two regimes: continuous MISE and singular targets such as point evaluation. The appropriate sampling distribution
and overlap condition differ across these regimes, and we make this distinction explicit in
Sections~\ref{section: theory} and~\ref{section:point evaluation}.

\subsection{Adaptivity Guarantees for Model Selection Procedure (Algorithm~\ref{algorithm: model selection})}
We next analyze the model selected by Algorithm~\ref{algorithm: model selection}.
Recall that \(\tilde\lambda\) denotes the proxy regularizer used to construct the validation target; concrete choices are specified in the continuous-MISE specialization in Section~\ref{section: continuous model selection}.

\begin{theorem}[Oracle Inequality for Model Selection]\label{theorem: general theory model selection}
Let Assumptions~\ref{ass:id}, \ref{ass:noise}, \ref{ass:bounded}, \ref{assumption: first overlap}, and~\ref{assumption: second overlap} hold.
Suppose we run Algorithm~\ref{algorithm: model selection} with \(M=n\), nuisance regularizer \(\lambda_0\asymp n^{-1}\log n\), proxy regularizer \(\tilde\lambda\gtrsim n^{-1}\log n\), and a candidate grid \(\Lambda\) such that \(|\Lambda|\lesssim n\) and every \(\lambda\in\Lambda\) satisfies \(\lambda\gtrsim\lambda_0\).
Let \(\hat h_{\operatorname{ms}}:=\hat h_{\hat\lambda}\) denote the selected estimator.
With probability at least \(1-2n^{-10}\), we have
\begin{equation*}
\cE_{\operatorname{ref}}(\hat h_{\operatorname{ms}})
\lesssim
\min_{\lambda \in \Lambda} \cE_{\operatorname{ref}}(\hat h_{\lambda})
+B(\tilde\lambda)
\left\{
\tilde\lambda\|h^\star\|_\cH^2
+\frac{R}{n}\|f^\star\|_\cF^2
\right\}.
\end{equation*}
Here, \(\lesssim\) hides constants and logarithmic factors.
\end{theorem}
The proof is deferred to Appendix~\ref{section:apdx proof model selection}.
The first term is the oracle risk among the candidate regularization levels, while the remaining terms are the price paid for proxy-based model selection.
When the grid is chosen to be a sufficiently rich geometric grid, the first term recovers the minimax-optimal rate, and the second term becomes the cost of model selection.
This cost depends on the proxy regularizer \(\tilde\lambda\).
In the continuous-MISE specialization, we take \(\tilde\lambda\asymp \log n/n\), which makes the model-selection overhead negligible relative to the minimax rate, up to logarithmic factors.
Section~\ref{section: continuous model selection} makes this specialization precise.

We now apply these general results to estimation of the whole TEF under continuous MISE in Section~\ref{section: theory} and to point evaluation in Section~\ref{section:point evaluation}. The latter analysis also yields a regret guarantee for treatment selection. Each specialization specifies the sampling distribution and overlap conditions appropriate to its target.


\section{Results for Continuous MISE}\label{section: theory}
Recall that \emph{continuous MISE} evaluates the TEF under a reference distribution \(\cP_{\operatorname{ref}}\) with a continuous density, such as the uniform distribution on \(\cA\).
We take \(\cP_{\operatorname{samp}}=\cP_{\operatorname{ref}}\), so the sampling coverage function is \(B(\cdot)\equiv1\). The remaining overlap measure compares \(\cP_{\operatorname{ref}}\) with the observational treatment law \(\cP_{A\mid X=x}\).
We first define this relative overlap and the relevant spectral regimes, then present an adaptive error bound for the final estimator selected by Algorithm~\ref{algorithm: model selection} and a matching lower bound in the Sobolev setting.

\subsection{Relative Overlap and Spectral Decay}

Standard causal inference frameworks typically assume uniform overlap, requiring the conditional treatment density to be uniformly bounded away from zero. We introduce a more flexible condition, \emph{Relative Overlap}, which quantifies overlap specifically with respect to the user-specified reference distribution $\cP_{\operatorname{ref}}$.

\begin{definition}[Relative Overlap]\label{definition; relative overlap}
Let $\cP_{A|X=x}$ denote the conditional distribution of treatment given covariates $x$. We say the data distribution satisfies \emph{relative overlap of degree $\gamma \in (0, 1]$} with respect to a reference measure $\cP_{\operatorname{ref}}$ if the following bound holds for all $x \in \cX$:
\begin{equation*}
\frac{\mathrm{d} \cP_{\operatorname{ref}}}{\mathrm{d} \cP_{A|X=x}} \le \frac{1}{\gamma} \quad \text{almost everywhere on } \cA.
\end{equation*}
\end{definition}

Intuitively, $\gamma$ represents the minimum density of the observed treatments relative to the target distribution. A smaller $\gamma$ implies weaker overlap. Based on this, we define the \emph{effective sample size} as:
\[
n_{\operatorname{eff}} := \gamma n.
\]
This quantity reflects the effective number of samples available for learning the target function $h^\star$ after accounting for the distribution shift.

Under this overlap condition, we now characterize the complexity of the target RKHS $\cH$ via the spectral decay of its kernel integral operator in order to state our convergence rates.
Recall that we define the second-moment operator associated with kernel $K_{\cH}$ and reference distribution \(\cP_{\operatorname{ref}}\) as 
\begin{equation*}
\bSigma_{\operatorname{ref}} := \EE_{a \sim \cP_{\operatorname{ref}}}[\phi(a)\otimes \phi(a)].
\end{equation*}
The following definition summarizes the spectral decay regimes considered in our analysis.
\begin{definition}[Spectral Decay]\leavevmode\phantomsection\label{def:spectral}
Let $\rho_1 \ge \rho_2 \ge \dots$ denote the eigenvalues of $\bSigma_{\operatorname{ref}}$. We consider three commonly studied decay regimes:
\begin{enumerate}[label=(\alph*), leftmargin=*]
\item \textbf{Polynomial:} $\rho_j \lesssim j^{-2\ell}$ for some $\ell > 1/2$ (e.g., Sobolev spaces).
\item \textbf{Exponential:} $\rho_j \lesssim \exp(-c j)$ for some $c > 0$ (e.g., Gaussian kernel).
\item \textbf{Finite Rank:} $\rho_j = 0$ for all $j > D$ (e.g., Linear kernels).
\end{enumerate}
Recall that we defined the \emph{effective dimension} at scale $\lambda$ as
\begin{equation*}
\Gamma(\lambda) := \Tr [ (\bSigma_{\operatorname{ref}} + \lambda \Ib)^{-1} \bSigma_{\operatorname{ref}} ]
= \sum_{j=1}^\infty \frac{\rho_j}{\rho_j + \lambda}.
\end{equation*}
\end{definition}

\subsection{Adaptivity Guarantees}\label{section: continuous model selection}
The optimal main regularizer depends on the unknown overlap level and spectral decay. Algorithm~\ref{algorithm: model selection} avoids the need to know these quantities in advance.
The following guarantee follows from Theorem~\ref{theorem: general theory model selection} with \(\cP_{\operatorname{samp}}=\cP_{\operatorname{ref}}\), \(B\equiv1\), and \(R=1/\gamma\). The intermediate bound for Algorithm~\ref{algorithm: main} with a fixed regularizer is stated and proved in Appendix~\ref{section:apdx proof theorem MISE}.

\begin{corollary}[Continuous MISE: Adaptivity of Model Selection]\label{corollary: model selection}
Suppose we run Algorithm~\ref{algorithm: model selection} with $M =n$, $\lambda_0 \asymp \log n/n$, proxy regularizer \(\tilde\lambda\asymp \log n/n\), $\cP_{\operatorname{samp}} = \cP_{\operatorname{ref}}$ and a geometric grid \(
\Lambda \coloneqq\{ \frac{2^{k-1} \log n}{n} : k = 1,\dotsc,L \},
\) where $L = \lceil \log_2 n \rceil + 1$. Under Assumptions~\ref{ass:id}, \ref{ass:noise}, and~\ref{ass:bounded}, relative overlap of degree \(\gamma\), and the respective spectral conditions in Definition~\ref{def:spectral}, with probability at least $1 - 2n^{-10}$, the selected estimator $\hat{h}_{\hat{\lambda}}$ satisfies the following:
\begin{itemize}
\item \textbf{Case (a):} \quad \(\cE(\hat h_{\hat{\lambda}}) \lesssim \ (n_{\operatorname{eff}})^{-\frac{2\ell}{1+2\ell}} + \frac{1}{n_{\operatorname{eff}}} \|f^\star\|_\cF^2\);
\item \textbf{Case (b):}\quad \(\cE(\hat h_{\hat{\lambda}}) \lesssim  \frac{1}{n_{\operatorname{eff}}} + \frac{\|f^\star\|_\cF^2}{n_{\operatorname{eff}}}\);
\item \textbf{Case (c):} \quad  \(\cE(\hat h_{\hat{\lambda}}) \lesssim   \frac{D}{n_{\operatorname{eff}}} + \frac{\|f^\star\|_\cF^2}{n_{\operatorname{eff}}}.\)
\end{itemize}
The notation \(\lesssim\) hides constants and logarithmic factors.
\end{corollary}
The proof is provided in Appendix~\ref{section:apdx proof theorem model selection continuous}.
The selected estimator adapts to the unknown degree of overlap \(\gamma\) and the spectral decay of the target kernel. Its leading rate depends on the simpler target space \(\cH\), while the nuisance function enters only through \(\|f^\star\|_\cF^2/n_{\operatorname{eff}}\).
For bounded nuisance norms, this yields the following examples for the selected estimator.

\begin{example}[Sobolev Nuisance]
If $\cF = H^\beta(\cZ)$ and $\cH =  H^\beta(\cA)$ for \(\cZ \subset \RR^d, \cA \subset \RR\) with $\beta > d/2$, then Case (a) implies a rate of $\tilde{\cO}((\gamma n)^{-\frac{2\beta}{1+2\beta}})$. This is significantly faster than the rate for learning $f^\star$ directly, which would be $\cO(n^{-\frac{2\beta}{d+2\beta}})$.
\end{example}

\begin{example}[NTK Nuisance]
If $\cF$ corresponds to a Neural Tangent Kernel (modeling an overparameterized neural network) but $h^\star \in H^1(\cA)$ for some $\cA \subset \RR$ (a smooth dose-response), our method achieves the fast rate $\tilde{\cO}((\gamma n)^{-2/3})$ for the TEF, bypassing the slow learning rate associated with the high-dimensional NTK.
\end{example}

\paragraph{Growing Nuisance Norms.} 

Another useful feature of Corollary~\ref{corollary: model selection} is that it remains informative even when the nuisance norm \(\|f^\star\|_\cF\) is allowed to grow with \(n\). In many RKHS analyses, the theorem is stated only over a fixed-radius ball, with the radius absorbed into the constant. Our bound makes this dependence explicit. More generally, under polynomial eigendecay, the optimal rate continues to hold whenever \(\|f^\star\|_{\mathcal F}^2 = \cO\!\left((\gamma n)^{1/(1+2\ell)}\right)\).
For example, when \(\cH = H^1(\cA)\), we still achieve the optimal rate for estimating \(h^\star\) as long as \(\|f^\star\|_{\mathcal F}^2 \lesssim n_{\operatorname{eff}}^{1/3}\).
This clarifies the precise sense in which our result relaxes the usual fixed-radius RKHS-ball assumption.

\paragraph{Comparison with Related Work.}

Compared with double machine learning (DML) methods for continuous treatments \citep{kennedy2017non,bonvini2022fast,colangelo2026double}, our approach entails a clear tradeoff.
In addition to avoiding conditional-density estimation, our analysis imposes no explicit estimation-error requirement on the nuisance outcome model beyond correct specification of \(f^\star\in\cF\). As a result, the theory continues to apply even in irregular covariate settings where standard guarantees for outcome regression may be unavailable.

DML methods, by contrast, offer robustness to misspecification of either nuisance model, which is an advantage over our approach. For continuous MISE, \citet{bonvini2022fast} give an integrated \(L^2\)-risk bound under the marginal treatment distribution for a global doubly robust estimator. However, attaining its oracle rate requires quantitative control of the nuisance estimators. In particular, the second-order remainder is driven by products of outcome-regression and treatment-density (or density-ratio) estimation errors. This global-risk analysis builds on the orthogonal statistical learning framework of \citet{foster2023orthogonal}, which can require \(L^4\) control of nuisance-estimation errors rather than \(L^2\) control alone. When the fitted conditional density relative to this reference distribution is bounded below at scale \(\gamma\), bounding the squared remainder in terms of density-estimation errors can introduce a \(1/\gamma^2\) factor. This factor can itself become a serious burden for nuisance estimation under weak overlap.

Our model-selection procedure shares the same high-level held-out validation idea as \citet{wang2026pseudo}, but the technical setting is substantially different. Their analysis studies regularizer tuning for classical single-stage KRR under covariate shift, whereas our first contribution is an error analysis for a two-stage KRR estimator of the treatment effect function, followed by model selection built specifically for this two-stage structure.

\subsection{Minimax Lower Bound}
We next show that the dependence on both sample size and overlap is sharp in the Sobolev benchmark.

\begin{theorem}[Minimax Lower Bound for Continuous MISE]\label{theorem: lower bound MISE}
Let \(\cA = [0,1]^p\), and suppose \(\beta>p/2\).
Consider Sobolev classes \(\cH=H^\beta(\cA)\) and \(\cF=H^\beta(\cX\times\cA)\).
Let \(\gamma\in(0,1]\).
Let \(\cP(\gamma)\) denote the set of distributions satisfying Relative Overlap of degree \(\gamma\) with respect to \(\cU(\cA)\).
Then the minimax MISE with respect to the uniform reference distribution satisfies
\begin{equation}
\inf_{\hat{h}} \sup_{\cP \in \cP(\gamma)} \EE[\cE(\hat{h})] \gtrsim (\gamma n)^{-\frac{2\beta}{2\beta+p}}.
\end{equation}
\end{theorem}
The proof is given in Appendix~\ref{section:apdx proof lower bound}.
For the Sobolev target space, the polynomial spectral regime has \(\ell=\beta/p\). Thus, for bounded nuisance norms, Corollary~\ref{corollary: model selection} matches this lower bound up to logarithmic factors, establishing minimax optimality in both \(n\) and \(\gamma\).

\section{Results for Point Evaluation and Regret}
\label{section:point evaluation}

So far, we have studied the case where the reference distribution is continuous.
We now turn to point evaluation, where the reference distribution is a point mass.
We then use this analysis to derive an \(L^\infty\) error bound and, as a corollary, an optimal regret bound for dose selection.

\subsection{Point Evaluation Error Bound}
\label{subsection: point evaluation error}

We first study point evaluation of the TEF at a fixed treatment value \(a_0\). We measure performance through the pointwise squared error
\begin{equation*}
\cE_{a_0}(\hat h)
:=
\bigl(\hat h(a_0)-h^\star(a_0)\bigr)^2.
\end{equation*}
This corresponds to the singular reference distribution
\(\cP_{\operatorname{ref}}=\delta(a_0)\).
Unlike continuous MISE, point evaluation calls for a sampling distribution that is localized around the target treatment value.
We therefore sample query treatments uniformly from a local ball \(B(a_0,r)\) for some small \(r>0\) and then apply the general theory developed above.
A key advantage of this construction is that optimal estimation requires overlap only in a neighborhood of the target point \(a_0\), rather than throughout the entire treatment region.

For point evaluation, we assume \(\cH=H^\beta(\cA)\) with \(\beta>p/2\), where \(\cA\subset\RR^p\) is a bounded domain with Lipschitz boundary. The Sobolev RKHS can be specified through an appropriate kernel \citep{doumeche2025fast}.
Without such regularity, pointwise guarantees are not generally available in an arbitrary RKHS.
As before, the nuisance space \(\cF\) may be any RKHS with a bounded kernel; we impose no further restrictions on it.

\paragraph{Local Overlap and Sampling.}
For a radius \(0<r\le\operatorname{diam}(\cA)\), let
\begin{equation*}
B(a_0,r)
:=
\left\{
a\in\cA:
\|a-a_0\|_2\le r
\right\}.
\end{equation*}
Algorithm~\ref{algorithm: main} samples query treatments from
\(\cP_{\operatorname{samp}}=\cU_r:=\cU(B(a_0,r))\), while the target reference distribution is
\(\cP_{\operatorname{ref}}=\delta(a_0)\).
We impose \emph{local overlap}: the conditional treatment density exists on this ball and, for some \(\kappa>0\),
\begin{equation*}
p_{A\mid X=x}(a)\ge \kappa
\qquad
\text{for all }x\in\cX
\text{ and }a\in B(a_0,r).
\end{equation*}
Under these conditions, the treatment coverage ratio and sampling coverage function satisfy
\begin{equation*}
R\lesssim\frac{1}{\kappa r^p},
\qquad
B(\lambda)\lesssim
1\vee\left(\frac{r^{2\beta-p}}{\lambda}\right)^{\frac{p}{2\beta}}.
\end{equation*}
The derivations are given in Appendix~\ref{section:apdx proof corollary PE}, using Lemma~\ref{lemma: Sobolev 1} for \(B(\lambda)\).
Substituting these bounds into Theorem~\ref{theorem: general theory} yields the following corollary.

\begin{corollary}[Upper Bound for Point Evaluation]
\label{corollary: point evaluation}
Fix \(a_0\in\cA\) and a radius \(0<r\le\operatorname{diam}(\cA)\).
Suppose \(\cH=H^\beta(\cA)\) with \(\beta>p/2\), and suppose \(\cA\) is a bounded domain with Lipschitz boundary.
Suppose there exists
\(\kappa\in(0,1]\) such that
\(p_{A\mid X=x}(a)\ge \kappa\) for all \(x\in\cX\) and
\(a\in B(a_0,r)\).
Under Assumptions~\ref{ass:id}, \ref{ass:noise}, and~\ref{ass:bounded}, suppose we run Algorithm~\ref{algorithm: main} with
\(\lambda_0\asymp n^{-1}\log n\), \(\lambda\gtrsim n^{-1}\log n\), and
\(\cP_{\operatorname{samp}}=\cU_r\).
Assume additionally that \(\|f^\star\|_\cF\) and
\(\|h^\star\|_\cH\) are bounded by constants.
Then, with probability at least \(1-n^{-10}\), we have
\begin{equation*}
\cE_{a_0}(\hat h_\lambda)
\lesssim
\left[
1\vee
\left(
\frac{r^{2\beta-p}}{\lambda}
\right)^{\frac{p}{2\beta}}
\right]
\left(
\lambda+\frac{1}{\kappa n r^p}
\right).
\end{equation*}
In particular, if
\(r\gtrsim(\kappa n)^{-\frac{1}{2\beta}}\) and we choose
\(\lambda\asymp\frac{\log n}{\kappa n r^p}\), then
\begin{equation*}
\cE_{a_0}(\hat h_\lambda)
\lesssim
(\kappa n)^{-\frac{2\beta-p}{2\beta}}.
\end{equation*}
Here, \(\lesssim\) hides constants and logarithmic factors.
\end{corollary}

We defer the proof of the corollary to
Appendix~\ref{section:apdx proof corollary PE}.
Corollary~\ref{corollary: point evaluation} is the point evaluation specialization of Theorem~\ref{theorem: general theory}.
The resulting exponent
\(\frac{2\beta-p}{2\beta}\) is the optimal pointwise rate for regression over a Sobolev class \citep{tuo2024asymptotic}.
The dependence on \(\kappa\) is also optimal, as established by the lower-bound result in Theorem~\ref{theorem: lower bound point}.
For instance, when \(p=1\), as in the standard dose-response setting, and \(\beta=2\), the rate becomes
\((\kappa n)^{-3/4}\), provided that
\(r\gtrsim(\kappa n)^{-1/4}\).
An important feature of this result is that the radius \(r\) may shrink to zero while still attaining the optimal rate, provided that it remains above the stated threshold.
The condition is \(r\gtrsim(\kappa n)^{-1/(2\beta)}\), rather than equality in order: with \(\lambda\asymp\log n/(\kappa n r^p)\), the powers of \(r\) in the optimized bound cancel, as shown in Appendix~\ref{section:apdx proof corollary PE}. The radius remains bounded by \(\operatorname{diam}(\cA)\), and local overlap must hold on the entire chosen ball.
Thus, unlike continuous MISE, overlap needs to hold only on the small local ball \(B(a_0,r)\), rather than over the entire treatment domain.
Moreover, this result imposes no learning-theoretic assumptions, such as eigenvalue decay, on the nuisance kernel \(K_\cF\).

\paragraph{Comparison with Related Work.}
\citet{mou2023kernel} establish a related plug-in learning bound in a confounding-free
setting; our result applies to the more challenging observational setting with confounded treatment
assignment. Importantly, this extension continues to require no eigenvalue decay assumption on
the nuisance kernel \(K_\cF\). The discussion below examines the relationship with DML and other nonparametric
approaches in greater detail.

As discussed in Section~\ref{section: continuous model selection}, DML methods offer robustness to nuisance-model misspecification, but their guarantees depend on nuisance-estimation rates.
Here, the nuisance errors must be controlled locally in the treatment variable. \citet{colangelo2026double} use \(L^2\) errors over covariates evaluated at the target treatment, while \citet{kennedy2017non} require analogous control uniformly over a shrinking treatment neighborhood. Such pointwise-in-treatment requirements do not follow from a global \(L^2\) learning bound and can be substantially more demanding.
The nuisance-error product is again amplified by inverse-density weighting: if the fitted conditional density is bounded below at order \(\kappa\) on the relevant neighborhood, a bound on the squared remainder can carry a \(1/\kappa^2\) factor.
Under its stated target-smoothness and local-overlap assumptions, Corollary~\ref{corollary: point evaluation} instead requires only correct specification and a bounded RKHS norm for the nuisance outcome model, with no separate pointwise nuisance-rate condition or conditional-density estimate.

Another closely related RKHS-based approach is \citet{singh2024kernel}. Their method is essentially plug-in: it first estimates a response function in a tensor-product RKHS over \((X,A)\), and then obtains the target curve by marginalizing over covariates. They additionally impose a source condition and provide an \(L^\infty\) error bound under that condition. Compared with their work, our approach has the following advantages.
\begin{enumerate}
\item Due to the tensor-product structure, the resulting marginalized curve may indeed be represented as an element of the treatment RKHS \(\cH\). However, their error bounds remain governed by the eigendecay and source conditions of the larger nuisance space \(\cF\), and therefore do not adapt to the potentially simpler structure of \(h^\star\) in \(\cH\). By contrast, our analysis allows broad RKHS choices beyond tensor-product kernels and yields rates whose leading term is controlled by the complexity of \(\cH\) itself.

\item Their \(L^\infty\) bound requires an additional source condition and does not establish a minimax point evaluation guarantee or sharp dependence on local overlap. In contrast, our point evaluation result is minimax optimal under the stated Sobolev target assumptions, with sharp dependence on \(\kappa\).
\end{enumerate}

\paragraph{Density-Free Approaches.}
There is also a line of work showing that one can avoid explicit conditional-density estimation and still achieve fast rates, but only under more specialized conditions.
\citet{robins2008higher} and \citet{bonvini2022fast} use higher-order influence functions (HOIFs) to obtain fast rates for binary and continuous treatments, respectively. In the continuous-treatment case, \citet{bonvini2022fast} avoid conditional-density estimation by replacing it with joint density estimation for \((X,A)\), and their theory requires this joint density to be sufficiently regular. Their guarantees are also limited to H\"older classes \(\cF=C^\beta(\cZ)\) with \(\beta>d/2\), and the HOIF construction is computationally demanding, involving projections onto finite-dimensional subspaces and the evaluation of U-statistics.
The independence-weighting approach of \citet{HGC24} is closer in spirit, but its guarantees are tied to a product-type smoothness class, roughly \(\cF \approx H^{d/2}(\cX)\otimes H^1(\cA)\) together with \(\cH \approx C^2(\cA)\), which is a special case of our framework. Their method also requires solving a large-scale non-convex quadratic program, which is NP-hard in general. Their bound is not optimal in its dependence on the degree of overlap \(\kappa\).
By contrast, our method is computationally practical and provides optimal dependence on both \(n\) and \(\kappa\) for broad RKHS choices of \(\cF\) and \(\cH\), beyond classical Sobolev and H\"older classes. For example, our framework includes \(\cF = H^{\frac{1}{2}+\varepsilon}_{\mathrm{mix}}(\cZ)\) for any \(\varepsilon>0\), which is broader than the class treated in \citet{HGC24}. We also allow substantially greater flexibility in the choice of \(\cH\).

\subsection{Regret for Dose Optimization}
\label{subsection: regret dose optimization}

The point evaluation bound also yields guarantees for treatment-selection rules.
In many applications, the goal is not only to estimate \(h^\star\) at a fixed treatment value, but also to select a treatment that maximizes the TEF.
For this purpose, we impose a global overlap condition over the treatment domain and sample query treatments uniformly,
\(\cP_{\operatorname{samp}}=\cU(\cA)\).
Here, global overlap means that the conditional treatment density is uniformly bounded below over all \(a\in\cA\) and all covariate values \(x\in\cX\).
This allows us to convert the point evaluation analysis into a uniform
\(L^\infty\) bound for the estimator and, consequently, into a regret bound for the greedy treatment rule.

Given the estimator \(\hat h_\lambda\) returned by
Algorithm~\ref{algorithm: main}, consider the greedy treatment-selection rule
\begin{equation*}
\hat a_\lambda
\in
\argmax_{a\in\cA}\hat h_\lambda(a),
\end{equation*}
and measure its simple regret by
\begin{equation*}
\operatorname{Reg}(\hat a_\lambda)
:=
\sup_{a\in\cA}h^\star(a)-h^\star(\hat a_\lambda).
\end{equation*}
The following corollary records the resulting uniform-error and simple-regret guarantees.

\begin{corollary}[\(L^\infty\) Rate and Regret Bounds]
\label{corollary: point regret}
Suppose \(\cA\subset\RR^p\) is compact with Lipschitz boundary and \(p=\cO(1)\),
\(\cH=H^\beta(\cA)\) with \(\beta>p/2\), and there exists
\(\kappa_{\operatorname{glob}}>0\) such that
\(p_{A\mid X=x}(a)\ge\kappa_{\operatorname{glob}}\) for all
\(x\in\cX\) and \(a\in\cA\).
Under Assumptions~\ref{ass:id}, \ref{ass:noise}, and~\ref{ass:bounded}, assume additionally that \(\|f^\star\|_\cF\) and \(\|h^\star\|_\cH\) are bounded by constants.
Suppose we run Algorithm~\ref{algorithm: main} with nuisance regularizer \(\lambda_0\asymp n^{-1}\log n\),
\(\cP_{\operatorname{samp}}=\cU(\cA)\) and choose
\(\lambda\) at the point evaluation rate,
\(\lambda\asymp\log n\cdot (n\kappa_{\operatorname{glob}})^{-1}\).
Then, with probability at least \(1-n^{-9}\),
\begin{equation*}
\sup_{a\in\cA}\bigl|\hat h_\lambda(a)-h^\star(a)\bigr|^2
\lesssim
\left(n\kappa_{\operatorname{glob}}\right)^{-\frac{2\beta-p}{2\beta}}.
\end{equation*}
Consequently, for
\(\hat a_\lambda\in\argmax_{a\in\cA}\hat h_\lambda(a)\),
\begin{equation*}
\operatorname{Reg}(\hat a_\lambda)
\lesssim
\left(n\kappa_{\operatorname{glob}}\right)^{-\frac{2\beta-p}{4\beta}}.
\end{equation*}
The notation \(\lesssim\) hides constants and logarithmic factors.
\end{corollary}

The proof is given in
Appendix~\ref{section:apdx proof point regret}.
The regret bound matches the absolute pointwise-error rate implied by Corollary~\ref{corollary: point evaluation}.

\paragraph{Comparison with Related Work.}
The decision problem in Corollary~\ref{corollary: point regret} is population-level dose optimization: the data-driven rule selects a single treatment value \(\hat a_\lambda\) intended to maximize the marginal TEF \(h^\star(a)=\EE[Y(a)]\).
This is close in spirit to early dose-response work such as \citet{flores2007dose}, which estimates the average dose-response function and then uses the estimated curve to locate an optimal dose.
Such procedures are plug-in in nature and, without additional structure, do not provide rates that adapt to the lower complexity of \(h^\star\) relative to the full covariate-treatment response surface.
Our contribution is complementary but different: we analyze the plug-in maximizer nonasymptotically, show that its simple regret is controlled by \(L^\infty\) error, and obtain minimax-optimal rates that adapt to the lower-dimensional structure of the TEF rather than to the complexity of the full nuisance surface.
One could also estimate an optimal dose by applying DML-based dose-response estimators \citep{kennedy2017non,bonvini2022fast} over many treatment values and then maximizing the resulting curve.
This route is practically more burdensome because it requires repeated local estimation across \(\cA\), whereas our method obtains the full estimated TEF through a single second-stage regression.

This result should also be distinguished from continuous-action policy learning. Works such as \citet{kallus2018policy}, \citet{demirer2019semiparametric}, and \citet{ai2024datadriven} study policies that map covariates to treatment values, \(\pi:\cX\to\cA\), and evaluate welfare or regret relative to a class of individualized policies. In contrast, our rule does not optimize over a policy class and does not recommend a covariate-dependent treatment. Instead, it optimizes the marginal causal curve \(h^\star\). 

\subsection{Minimax Lower Bounds}
The following lower bounds show that the point evaluation and regret rates are sharp even under global overlap, which is stronger than the local condition used for point evaluation.

\begin{theorem}[Minimax Lower Bounds for Point Evaluation and Regret]\label{theorem: lower bound point}
Let \(\cA = [0,1]^p\), let \(a_0\) be an interior point of \(\cA\), and suppose \(\beta>p/2\).
Consider Sobolev classes \(\cH=H^\beta(\cA)\) and \(\cF=H^\beta(\cX\times\cA)\).
Let \(\kappa\in(0,1]\).
Then the following lower bounds hold.
\begin{enumerate}[label=(\roman*), leftmargin=*]
\item Let \(\cP_{\operatorname{glob}}(\kappa)\) denote the set of distributions satisfying
\begin{equation*}
\de \cP_{A \mid X=x} \ge \kappa\,da
\end{equation*}
for all \(x \in \cX\).
Then the minimax pointwise error satisfies
\begin{equation}
\inf_{\hat{h}} \sup_{\cP \in \cP_{\operatorname{glob}}(\kappa)} \EE[\cE_{a_0}(\hat{h})] \gtrsim \left(n\kappa\right)^{-\frac{2\beta-p}{2\beta}}.
\end{equation}

\item For any treatment-selection rule \(\hat a\), define \(\operatorname{Reg}(\hat a)=\sup_{a\in\cA}h^\star(a)-h^\star(\hat a)\). Then
\begin{equation}
\inf_{\hat a} \sup_{\cP \in \cP_{\operatorname{glob}}(\kappa)} \EE[\operatorname{Reg}(\hat a)] \gtrsim \left(n\kappa\right)^{-\frac{2\beta-p}{4\beta}}.
\end{equation}
\end{enumerate}
\end{theorem}
The proof is given in Appendix~\ref{section:apdx proof lower bound}.
These bounds match Corollaries~\ref{corollary: point evaluation} and~\ref{corollary: point regret} up to logarithmic factors, with \(\kappa_{\operatorname{glob}}=\kappa\) for regret. Since the global-overlap class is contained in the local-overlap class, the pointwise lower bound also applies under the weaker local condition.

\section{Experiments}\label{section:experiments}
We evaluate the proposed estimator on both synthetic and semi-real benchmarks.

\subsection{Synthetic Data}
We first study a confounded synthetic benchmark with a known ground-truth TEF, which allows us to isolate the contribution of the second-stage smoothing step relative to standard plug-in estimation.

\paragraph{Data-Generating Process.}
Let \(x=(x^1,\ldots,x^q)\in\mathbb{R}^q\) with \(q=10\), and draw \(x_i \stackrel{\mathrm{i.i.d.}}{\sim}\mathcal{U}([-1,1]^q)\).
Define \(s(x) := \frac{1}{q} \sum_{i=1}^q \sin(x^i)\) and set the true conditional mean to
\[
f^\star(x,a)=\sin(a) + 4 s(x)(\sin(a)+1).
\]
This yields the target TEF \(h^\star(a)=\sin(a)\).
To induce confounding, let \(r(x)=\sigma(2\sum_{i=1}^q x^i)\), where \(\sigma(\cdot)\) denotes the sigmoid function.
Conditional on \(x\), we sample \(u\sim\mathrm{Beta}(20r(x),20(1-r(x)))\) and map it to the treatment domain \(\cA=[-\pi,\pi]\) via \(a=\pi(2u-1)\).
Finally, outcomes are generated as \(y_i=f^\star(x_i,a_i)+\varepsilon_i\) with \(\varepsilon_i\sim\mathcal{N}(0,1)\).
We report MISE for \(n\in\{500,1000\}\).

\paragraph{Proposed Method and Kernel Choices.}
For \textit{our proposed method} (Algorithm~\ref{algorithm: model selection}), we partition the dataset into a training set $\cD_{\text{train}}$ and a validation set $\cD_{\text{val}}$, with sizes $n_1 = |\cD_{\text{train}}|$ and $n_2 = |\cD_{\text{val}}|$, respectively.
The regularization grid is defined as $\Lambda = \{ c \frac{2^{k-1}}{n_1} : k = 1,\dots,9 \}$ with $c = 0.1$, and the proxy regularizer is fixed at $\tilde{\lambda} = c/n_2$.
For pseudo-outcome construction, we set \(\cP_{\operatorname{samp}}\) to the uniform distribution on the treatment domain.
Regarding kernel specification, we employ a Mat\'ern kernel with smoothness parameter $\nu=1.5$ for the first-stage nuisance estimation ($f^\star$), and a smoother Mat\'ern kernel with $\nu=2.5$ for the second-stage target estimation ($h^\star$).

\paragraph{Baselines.}
We compare with two natural baselines: \textit{Plug-in KRR}, which estimates \(f^\star\) and directly marginalizes it without a second smoothing stage, and \textit{Direct Regression}, which regresses \(y\) on \(a\) alone and therefore ignores confounding.
For the \textit{Plug-in KRR} baseline, we adopt the same data-splitting strategy and regularization grid $\Lambda$.
Candidate models $\{\hat f_\lambda\}_{\lambda\in\Lambda}$ are trained on $\cD_{\text{train}}$ using KRR with regularizer \(\lambda\), and the final regularizer is selected by minimizing the Mean Squared Error (MSE) on $\cD_{\text{val}}$.
We use a Mat\'ern kernel with smoothness parameter $\nu=1.5$.
The \textit{Direct Regression} baseline is similarly implemented using KRR with a Mat\'ern kernel ($\nu=1.5$).
We also apply the same hold-out validation strategy: candidate estimators are trained on $\cD_{\text{train}}$ using the regularization grid $\Lambda$, and the final regularizer is selected by MSE on $\cD_{\text{val}}$.

\paragraph{Length-Scale Selection.}
For all methods, we construct a length-scale grid by matching the kernel correlation at the median pairwise distance and cross-validating over this grid, following the widely used median heuristic for kernel bandwidth selection \citep{garreau2017large,gretton2012kernel}.
Specifically, let $r_{\mathrm{med}}$ denote the median pairwise distance.
For each correlation level $\rho$ in a chosen set within $[0.15, 0.85]$, we define the candidate length-scale $\ell_\rho$ by solving $K(r_{\mathrm{med}};\ell_\rho)/K(0;\ell_\rho)=\rho$.
We then cross-validate over the resulting grid $\{\ell_\rho\}$ to select the final length-scale $\ell$.
In our proposed framework, the length-scale of the second-stage kernel $K_\cH$ is set equal to that of $K_\cF$.
In this experiment, the selected length-scale was $\ell = 3$ for our method and Plug-in KRR, and $\ell = 2$ for Direct Regression.

\paragraph{Results.}
Table~\ref{table:MISE_synthetic} presents the MISE with respect to the uniform distribution on \(\cA=[-\pi,\pi]\), averaged over $100$ independent runs. In implementation, the integral is approximated by averaging squared errors over a dense uniform grid on the treatment domain.
Our proposed method consistently outperforms both the plug-in baseline and direct regression across both sample sizes, confirming the efficacy of the second-stage smoothing procedure.

\begin{table}[h]
\centering
\caption{Comparison of MISE over 100 runs (standard error in parentheses) on synthetic data. Reported values are scaled by $100$.}
\label{table:MISE_synthetic}
\vspace{0.2cm}
\begin{tabular}{@{}lcc@{}}
\toprule
Method & $n=1000$ & $n=500$ \\
\midrule
\textbf{Ours}            & 6.22 (0.32)  & 8.40 (0.35)  \\
Plug-in                  & 9.87 (0.31)  & 11.81 (0.28) \\
Direct Regression        & 14.52 (0.31) & 14.84 (0.36) \\
\bottomrule
\end{tabular}
\end{table}

\subsection{Semi-Real Data}
We further evaluate our methodology on a semi-real benchmark derived from the Job Corps dataset, following the same empirical benchmark specification as \citet{colangelo2026double}.
The Job Corps dataset is a well-established benchmark in the continuous treatment effect literature \citep{singh2024kernel,colangelo2026double,huber2020direct,hsu2026testing,lee2018partial,flores2012estimating}.
Funded by the U.S. Department of Labor, the Job Corps program is the largest publicly funded job training program for disadvantaged youth in the United States.
The data come from the National Job Corps Study, conducted during the program's randomized evaluation period (November 1994--February 1996).
While access to the program was randomized among eligible applicants, the intensity of participation was self-selected.
This makes it a canonical benchmark for continuous treatment effect estimation under a selection-on-observables assumption.
The relevant variables are defined as follows:
\begin{itemize}
\item \textbf{Treatment ($a$):} Total hours spent in academic or vocational training classes during the first year after randomization.
\item \textbf{Outcome ($y$):} The proportion of weeks employed during the second year after randomization.
\item \textbf{Covariates ($x$):} A 40-dimensional vector ($x \in \mathbb{R}^{40}$) of baseline characteristics, including age, gender, ethnicity, language ability, education, marital status, household size, household income, prior receipt of public assistance, family background, and health-related behaviors.
\end{itemize}

\paragraph{Semi-Synthetic Construction.}
To obtain a benchmark with known ground truth while preserving the empirical covariate-treatment structure, we adopt a semi-synthetic construction. Let the original Job Corps sample be \(\{x_{\mathrm{job},i}, a_{\mathrm{job},i}, y_{\mathrm{job},i}\}_{i=1}^n\) with \(n=4024\). We first fit a Generalized Random Forest (GRF) regressor \citep{athey2019generalized}, denoted by \(\hat f_{\mathrm{semi}}\), to the original covariates, treatments, and outcomes, and then compute residuals
\begin{equation*}
g_i := y_{\mathrm{job},i} - \hat f_{\mathrm{semi}}(x_{\mathrm{job},i}, a_{\mathrm{job},i}).
\end{equation*}
The semi-real dataset \(\tilde{\cD} = \{x_{\mathrm{job},i}, a_{\mathrm{job},i}, \tilde y_i\}_{i=1}^n\) is obtained by injecting randomized noise into the fitted response,
\begin{equation*}
\tilde y_i = \hat f_{\mathrm{semi}}(x_{\mathrm{job},i}, a_{\mathrm{job},i}) + \tilde e_i g_i,
\end{equation*}
where \(\{\tilde e_i\}_{i=1}^n\) are i.i.d.\ Rademacher random variables. We follow \citet{colangelo2026double} and evaluate MISE on the treatment grid \(\{160, 200, \dots, 2000\}\). The ground-truth marginal dose-response curve is defined by averaging the fitted response surface over the empirical covariate distribution:
\begin{equation*}
h^\star(a) := \frac{1}{n} \sum_{i=1}^n \hat{f}_{\mathrm{semi}}(x_{\mathrm{job},i}, a),
\end{equation*}
for \(a \in \cA\).
The reported MISE therefore measures the discrepancy between the estimated treatment effects and this ground-truth curve on the evaluation grid.

\paragraph{Baselines and Implementation.}
We compare against Plug-in KRR and Direct Regression, both tuned by leave-one-out cross-validation (LOOCV), as well as the DML estimators of \citet{colangelo2026double} with neural-network (NN), kernel-neural-network (KNN), GRF, and least absolute shrinkage and selection operator (LASSO) nuisance learners.
For DML, we use the tuning parameters and bandwidths reported in \citet{colangelo2026double}.
We employ a tensor-product Laplace kernel (Mat\'ern kernel with smoothness $\nu=0.5$) for nuisance estimation (for both our method and Plug-in KRR) and a Mat\'ern kernel with smoothness $\nu=1.5$ for the second-stage KRR.
For Direct Regression, we use a Laplace kernel.
The length-scale parameters for the nuisance kernel \(K_\cF\) were selected via cross-validation on a grid constructed around the median heuristic (similar to the synthetic experiments), resulting in \(\ell_x = 13\) (covariates) and \(\ell_a = 6000\) (treatment).
These length-scale values were used for both our method and Plug-in KRR.
For Direct Regression, we selected the length-scale using the same procedure, which yielded \(\ell = 3000\).
For our method, we set the second-stage kernel length-scale of \(K_\cH\) equal to the treatment length-scale of \(K_\cF\) (i.e., \(\ell = 6000\)).
We denote the sizes of the training and validation sets as $n_1 = |\cD_{\text{train}}|$ and $n_2 = |\cD_{\text{val}}|$, respectively.
Consistent with the synthetic experiments, we define the regularization grid as $\Lambda = \{ c \frac{2^{k-1}}{n_1} : k = 1,\dots,9 \}$ with $c = 0.1$, and set the proxy regularizer to $\tilde{\lambda} = c/n_2$.
The second-stage regression requires a sampling distribution \(\cP_{\operatorname{samp}}\); since approximately \(94\%\) of the observed treatments in the Job Corps data lie within \([0, 3000]\), we sample \(a\) uniformly from this range.
For Plug-in KRR (LOOCV), to ensure computational efficiency during LOOCV, we employ the Nystr\"om approximation with \(m = 700\) landmark points.
For our method and the Plug-in baseline, covariates are normalized to \([0,1]^d\).
For DML baselines, to be consistent with \citet{colangelo2026double}, categorical covariates are one-hot encoded, and the DML baselines use the same standardization procedures as in the original study.

\paragraph{Results.}
Table~\ref{table:MISE_semi_real} summarizes the results over \(100\) independent runs. Our method achieves the lowest MISE among all baselines. Notably, although the ground-truth curve is constructed from the GRF-fitted response surface \(\hat f_{\mathrm{semi}}\), our kernel-based approach still outperforms DML (GRF) (1.24 vs. 2.42). This highlights the ability of our two-stage smoothing procedure to recover the structure of the treatment effect function even under model misspecification.

\begin{table}[h]
\centering
\caption{MISE (standard errors, SEs, in parentheses) for the semi-real benchmark, averaged over \(100\) independent runs. Our proposed method yields the lowest error among all compared baselines.}
\label{table:MISE_semi_real}
\vspace{0.2cm}
\begin{tabular}{lc}
\toprule
Method & Mean MISE (SE) \\
\midrule
\textbf{Ours}       & \textbf{1.2466 (0.1209)} \\
Plug-in LOOCV       & 1.6197 (0.1146) \\
Direct Regression   & 1.6970 (0.1264) \\
DML (GRF)           & 2.4230 (0.1837) \\
DML (NN)            & 2.1065 (0.1454) \\
DML (LASSO)         & 2.8732 (0.2391) \\
DML (KNN)           & 2.9742 (0.2165) \\
\bottomrule
\end{tabular}
\end{table}

\section{Conclusion}\label{section:conclusion}

In this work, we have introduced a two-stage kernel ridge regression framework for estimating continuous treatment effects under confounding. Our primary theoretical contribution is that the statistical complexity of estimating the TEF is governed by the complexity of the target function space $\cH$, rather than the potentially high-dimensional nuisance space $\cF$.
By constructing pseudo-outcomes followed by a second stage of smoothing, our estimator achieves minimax-optimal convergence rates that adapt to both the intrinsic regularity of the TEF and the degree of overlap.
These optimality guarantees cover continuous MISE, point evaluation at a fixed treatment value, and simple regret for optimal dose selection.
Furthermore, we developed a fully data-driven model selection procedure that achieves these optimal rates without requiring prior knowledge of the structural parameters.

Several promising directions remain for future research. First, our current approach is predicated on the assumption of a well-specified outcome model. To improve robustness against model misspecification, future work could investigate incorporating inverse propensity weighting or developing doubly robust kernel estimators that guarantee consistency if either the outcome or treatment density model is correct.
Second, while our analysis leverages the tractability of RKHS theory, extending this ``decoupling'' strategy to general function classes, such as neural networks beyond the NTK regime, represents an important step toward broader practical applicability.

\section*{Acknowledgement}
We thank the reviewers for their thoughtful and constructive feedback, which helped improve the presentation of this work. 
The research is supported by NSF grant DMS-2515679.

{
\bibliographystyle{ims}
\bibliography{bib}
}

\clearpage
\newpage
\appendix
\begin{center}
\textbf{{\Huge Supplementary Materials}}
\end{center}
\vspace{0.5cm}
\etocdepthtag.toc{mtappendix}
\etocsettagdepth{mtchapter}{none}
\etocsettagdepth{mtappendix}{subsection}
{\tableofcontents}

\section{Proof of Theorem~\ref{theorem: general theory}}\label{section:apdx proof general theory}
\subsection{Groundwork for Proof}\label{section:apdx groundwork}

\paragraph{Proof Roadmap.}
We first construct high-probability good events that control empirical covariance operators and Monte Carlo fluctuations.
On this good event, we derive an error decomposition for \(\hat h_\lambda-h^\star\) and bound the pseudo-outcome bias term \(\Bcr\), the second-stage ridge bias term \(\Ccr\), and the propagated first-stage error term \(\Pcr\) separately.

\paragraph{Linear Model in RKHS.}
For the function classes \(\cF\) and \(\cH\), we denote the corresponding Hilbert spaces by \(\mathbb{F}\) and \(\HH\), respectively.
For \(\mathbb{F}\), we define the feature mapping as \(\psi(\cdot): \cZ \to \mathbb{F}\).
For \(\HH\), we define the feature mapping as \(\phi(\cdot): \cA \to \HH\).
Recall that \(\cZ = \cX \times \cA\).
For any \(f_\theta \in \cF\), we write \(\theta\) for its Hilbertian element; similarly, for \(h_\eta \in \cH\), we write \(\eta\) for its Hilbertian element.

We rewrite our response model as a linear model using Hilbertian language.
We have a dataset \(\cD = \{(x_i, a_i, y_i)\}_{i=1}^n\), where \(x_i \in \RR^d, a_i \in \cA, y_i \in \RR\), such that
\begin{equation*}
y_i = \psi(x_i, a_i)^\top \theta^\star + \varepsilon_i
\end{equation*}
where
\begin{equation*}
f^\star(x_i,a_i):=  \psi(x_i, a_i)^\top \theta^\star
\end{equation*}
and \(\theta^\star\) is the corresponding Hilbertian element of \(f^\star\).
For the TEF, for any \(a \in \cA\), we write
\begin{equation*}
h^\star(a)= \EE_{x \sim \cP_x} [f^\star(x,a)] = \phi(a)^\top \eta^\star
\end{equation*}
where \(\eta^\star\) is the corresponding Hilbertian element of \(h^\star\) in \(\HH\).

\paragraph{Design Operators.}
Consider \(N\) generic elements \(\{v_1, \dots v_N\}\) for some \(N>1\) in a Hilbert space \({\XX}\).  
We define the \emph{design operator} of \(\{v_1, \dots, v_N \}\) as $\Vb: \XX \to \RR^N$, which, for all \(\theta \in \XX\), satisfies:
\begin{equation*}
\Vb \theta = (v_1^\top \theta, v_2^\top \theta, \dots, v_N^\top \theta)^\top.
\end{equation*}
Similarly, we define the adjoint of \(\Vb\), denoted as \(\Vb^\top: \RR^N \to \XX \), as the operator such that for all \(\ab =(a_1, \dots, a_N) \in \RR^N\),
\begin{equation*}
\Vb^\top \ab= \sum_{i=1}^N a_i v_i \in \XX.
\end{equation*}

Recall that
\begin{equation*}
\bSigma_{\operatorname{ref}} := \EE_{a \sim \cP_{\operatorname{ref}}}[\phi(a)\otimes \phi(a)],
\qquad
\bSigma_{\operatorname{samp}} := \EE_{a \sim \cP_{\operatorname{samp}}}[\phi(a)\otimes \phi(a)].
\end{equation*}
These are the second-moment operators appearing in the general-theory statement.

\paragraph{Summary of Notations.}
We define the design operators of \(\{\phi(a_j^\prime)\}_{j=1}^n\) and \(\{\psi(x_i,a_i)\}_{i=1}^n\) as \(\Ab\) and \(\Zb\), respectively.
Also, we define
\begin{equation*}
\bar \psi(a):= \EE_{x \sim \cP_x}[\psi(x,a)]
\end{equation*}
and the empirical feature mean as
\begin{equation*}
\hat{\bar \psi}(a):=\frac{1}{n}\sum_{i=1}^n \psi(x_i,a).
\end{equation*}
We define the design operator of \(\{\hat{\bar \psi}({a}^\prime_j)\}_{j=1}^n\) as \(\Wb\).
We summarize the key notation below in Table~\ref{table; appendix notation general theory}.

\begin{table}[H]
    \centering
    \caption{Summary of Notation Used in the Proof of Theorem~\ref{theorem: general theory}}
    \label{table; appendix notation general theory}
    \renewcommand{\arraystretch}{1.4}
    \begin{tabular}{c l}
        \toprule
        \textbf{Notation} & \textbf{Description} \\
        \midrule
        \(\Zb\) & Design operator of \(\{\psi(x_i,a_i)\}_{i=1}^n\) \\
        \(\Ab\) & Design operator of \(\{\phi({a}^\prime_j)\}_{j=1}^n\) \\
        \(\bar \psi(a)\) & Feature mean, defined as \(\EE_{x \sim \cP_x}[\psi(x,a)]\) \\
        \(\hat{\bar \psi}(a)\) & Empirical feature mean, defined as \(\frac{1}{n}\sum_{i=1}^n\psi(x_i,a)\) \\
        \(\Wb\) & Design operator of \(\{\hat{\bar \psi}({a}^\prime_j)\}_{j=1}^n\) \\
        \(\Yb\) & Vector of responses \(y_1, \dots, y_n\) (\(\Yb \in \RR^n\)) \\
        \(\theta^\star\) & Hilbertian element of \(f^\star\) in \(\mathbb{F}\) \\
        \(\eta^\star\) & Hilbertian element of \(h^\star\) in \(\HH\) \\
        \(\bSigma_{\operatorname{ref}}\) & \(\EE_{a \sim \cP_{\operatorname{ref}}}[\phi(a)\otimes \phi(a)]\) \\
        \(\bSigma_{\operatorname{samp}}\) & \(\EE_{a \sim \cP_{\operatorname{samp}}}[\phi(a)\otimes \phi(a)]\) \\
        \(\Sb_{\lambda}\) & \((\bSigma_{\operatorname{samp}}+\lambda \Ib)^{-1/2}\bSigma_{\operatorname{ref}}(\bSigma_{\operatorname{samp}}+\lambda \Ib)^{-1/2}\) \\
        \(\hat{\theta}\) & Corresponding Hilbertian element of \(\hat f\) in Algorithm~\ref{algorithm: main} \\
        \(\hat{\eta}_{\lambda}\) & Corresponding Hilbertian element of \(\hat h_{\lambda}\) in Algorithm~\ref{algorithm: main} \\
        \bottomrule
    \end{tabular}
\end{table}

\paragraph{First-Stage Regression: Hilbertian Formulation.}
Recall that we define the design operator of \(\{\psi(x_i,a_i)\}_{i=1}^n\) as \(\Zb\).
We run kernel ridge regression to obtain the estimator \(\hat{\theta}\) as
\begin{equation*}
\hat\theta = (\Zb^\top \Zb + n\lambda_0 \Ib)^{-1 } \Zb^\top \Yb.
\end{equation*}
Then \(\hat \theta\) is the corresponding Hilbertian element of \(\hat f \in \cF\) in Algorithm~\ref{algorithm: main}.

\paragraph{Second-Stage Regression: Hilbertian Formulation.}
In the second stage regression of Algorithm~\ref{algorithm: main}, we perform regression on the pseudo-outcomes \(m(a_j')\).
We express this process in a closed form using the language of elements and operators in a Hilbert space.

\begin{enumerate}
\item Sample \(\{a_j^\prime\}_{j=1}^n\) from \(\cP_{\operatorname{samp}}\) and create pseudo-outcomes \(\{m(a_j^\prime)\}_{j=1}^n\).

Explicitly, the pseudo-outcome for \(a_j'\) is:
\begin{equation*}
m(a_j') := \frac{1}{n}\sum_{i=1}^n \psi(x_i, a_j')^\top \hat\theta = \hat{\bar \psi}(a_j')^\top \hat \theta.
\end{equation*}

\item
Define the design operator of \(\{\phi({a}^\prime_j)\}_{j=1}^n\) as \(\Ab\).
Obtain the final estimator
\begin{equation*}
\begin{aligned}
\hat{\eta}_\lambda &= (\Ab^\top \Ab + n\lambda \Ib)^{-1} \sum_{j=1}^n \phi(a_j^\prime) m(a_j^\prime) \\
&= (\Ab^\top \Ab + n\lambda \Ib)^{-1} \Ab^\top \Wb \hat{\theta}
\end{aligned}
\end{equation*}
for a regularizer \(\lambda >0\).
\end{enumerate}
For a generic \(a \in \cA\), we also write
\begin{equation*}
m(a):=\frac{1}{n}\sum_{i=1}^n \hat f(x_i,a)=\hat{\bar \psi}(a)^\top \hat\theta .
\end{equation*}

\subsection{Good Events for Proof}\label{subsection:apdx good event general theory}
Now, we define the high-probability good event for the proof. 
In order, we define \(\Ecr_1, \Ecr_2, \Ecr_3\), and finally, we define \(\Ecr_{\operatorname{good}} := \Ecr_1 \cap \Ecr_2 \cap \Ecr_3\).

\paragraph{Good Event 1: \(\Ecr_1\)}
We first define the good event \(\Ecr_1\) as follows:
we first sample \(a_1', \dots, a_n'\) from \(\cP_{\operatorname{samp}}\), which is independent of \(\{x_1, \dots x_n\}\).
Then, with probability at least \(1-\frac{1}{n^{11}}\), by Hoeffding's inequality (and an implicit union bound), it holds for all \(j \in [n]\) that:
\begin{equation}\label{equation; good event 1}
\big|\frac{1}{n}\sum_{i=1}^n f^\star(x_i,a_j') - \EE_{x \sim \cP_x}[f^\star(x, a_j')]\big| \lesssim \frac{\xi\|\theta^\star\|_{\mathbb{F}} \sqrt{\log n}}{\sqrt{n}}.
\end{equation}
We can apply Hoeffding's inequality since for all \(x \in \cX, a \in \cA\)
\begin{equation*}
|f^\star(x,a)| = | \psi(x, a)^\top \theta^\star| \le \| \psi(x, a) \|_\mathbb{F} \|\theta^\star\|_{\mathbb{F}} \leq \xi\|\theta^\star\|_{\mathbb{F}}.
\end{equation*}

\begin{definition}[Good Event $\Ecr_1$]
We define the good event \(\Ecr_1\) as the event where \cref{equation; good event 1} holds for all \(j \in [n]\).
From the above observation, we have \(\PP[\Ecr_1] \geq 1-n^{-11}\).
\end{definition}

\paragraph{Good Event 2: \(\Ecr_2\)}
Next, we define the second good event \(\Ecr_2\).
For a fixed \(a\), we define
\begin{equation*}
\widehat \Qb(a):=\frac{1}{n}\sum_{i=1}^n \psi(x_i, a)\psi(x_i, a)^\top.
\end{equation*}
Also, we define
\begin{equation*}
\Qb( a) = \EE_{x \sim \cP_x}[\psi(x, a)\psi(x, a)^\top].
\end{equation*}
Using Lemma~\ref{lemma: trace class concentration bounded}, with probability at least \(1-\frac{1}{n^{11}}\), we have
\begin{equation}\label{equation; good event 2-1}
\frac{1}{c}(\widehat\Qb({a}^\prime_j) + \frac{\log n}{n} \Ib) \preceq    \Qb({a}^\prime_j) + \frac{\log n}{n} \Ib \preceq c(\widehat\Qb({a}^\prime_j) + \frac{\log n}{n} \Ib)
\end{equation}
for all \(j \in [n]\) for some absolute constant \(c>1\).
We set
\begin{equation*}
\Qb := \EE_{a_j' \sim \cP_{\operatorname{samp}}}[\Qb(a_j')] =\EE_{\cP_x \otimes \cP_{\operatorname{samp}}}[\psi(x,a)\psi(x,a)^\top]
\end{equation*}
and
\begin{equation*}
\bar \Qb := \frac{1}{n} \sum_{j=1}^n \Qb(a_j').
\end{equation*}
\begin{lemma}\label{lemma: sampled Qbar bound}
There exists an absolute constant \(c>1\) such that, with probability at least \(1-n^{-11}\),
\begin{equation}\label{equation; good event 2-2}
\bar \Qb \preceq c(\Qb + \frac{\log n}{n} \Ib)
\end{equation}
holds.
\end{lemma}

\noindent\textbf{Proof of Lemma~\ref{lemma: sampled Qbar bound}. }
For \(\lambda'>0\), define
\[
\Tb_j :=(\Qb +\lambda' \Ib)^{-\frac{1}{2}} (\Qb(a_j^\prime)-\Qb) (\Qb +\lambda' \Ib)^{-\frac{1}{2}}, \qquad j\in[n].
\]
Then \(\EE[\Tb_j]=0\). Since the features are bounded, \(\|\Qb(a)\|_{\operatorname{op}}\le \xi\), so in Lemma D.3 of \citet{wang2026pseudo} we can take \(U \asymp \frac{\xi}{\lambda'}\) because
\begin{equation*}
\|\Tb_j \|_{\operatorname{op}} \lesssim \frac{\xi}{\lambda'}.
\end{equation*}
For the variance proxy, write
\[
X_j:=(\Qb +\lambda' \Ib)^{-\frac{1}{2}} \Qb(a_j^\prime) (\Qb +\lambda' \Ib)^{-\frac{1}{2}},
\qquad
M:=\EE[X_j]=(\Qb +\lambda' \Ib)^{-\frac{1}{2}} \Qb (\Qb +\lambda' \Ib)^{-\frac{1}{2}}.
\]
Then \(\Tb_j=X_j-M\), and
\[
\EE[\Tb_j^2]=\EE[(X_j-M)^2]=\EE[X_j^2]-M^2\preceq \EE[X_j^2]
\]
since \(M^2\succeq 0\). By monotonicity of \(\|\cdot\|_{\operatorname{op}}\) on positive semidefinite operators, we get
\begin{equation*}
\begin{aligned}
\|\EE[\Tb_j^2] \|_{\operatorname{op}}&\le  \| \EE[(\Qb +\lambda' \Ib)^{-\frac{1}{2}} \Qb(a_j^\prime) (\Qb +\lambda' \Ib)^{-\frac{1}{2}}(\Qb +\lambda' \Ib)^{-\frac{1}{2}} \Qb(a_j^\prime) (\Qb +\lambda' \Ib)^{-\frac{1}{2}}] \|_{\operatorname{op}}\\
&= \| \EE[(\Qb +\lambda' \Ib)^{-\frac{1}{2}} \Qb(a_j^\prime) (\Qb +\lambda' \Ib)^{-1} \Qb(a_j^\prime) (\Qb +\lambda' \Ib)^{-\frac{1}{2}} ] \|_{\operatorname{op}} \\
&\stackrel{(i)}{\lesssim }\frac{\xi}{\lambda'} \|\EE[(\Qb +\lambda' \Ib)^{-\frac{1}{2}} \Qb(a_j^\prime)^{\frac{1}{2}} \Qb(a_j^\prime)^{\frac{1}{2}} (\Qb +\lambda' \Ib)^{-\frac{1}{2}} ] \|_{\operatorname{op}}  \\
&\lesssim \frac{\xi}{\lambda'}  \| \EE[(\Qb +\lambda' \Ib)^{-\frac{1}{2}} \Qb(a_j^\prime)  (\Qb +\lambda' \Ib)^{-\frac{1}{2}} ] \|_{\operatorname{op}} \\
&\lesssim \frac{\xi}{\lambda'} \|(\Qb +\lambda' \Ib)^{-\frac{1}{2}} \Qb  (\Qb +\lambda' \Ib)^{-\frac{1}{2}} \|_{\operatorname{op}} \\
&\lesssim \frac{\xi}{\lambda'}.
\end{aligned}
\end{equation*}
Hence \(\|\sum_{j=1}^n \EE[\Tb_j^2]\|_{\operatorname{op}}\lesssim n\xi/\lambda'\).
In step (i), we use
\begin{equation*}
\| \Qb(a_j^\prime)^{\frac{1}{2}} (\Qb +\lambda' \Ib)^{-1} \Qb(a_j^\prime)^{\frac{1}{2}}\|_{\operatorname{op}} \leq \frac{1}{\lambda'} \| \Qb(a_j^\prime)\|_{\operatorname{op}} \leq \frac{\xi}{\lambda'} .
\end{equation*}

Let
\[
\Jb := (\Qb+\lambda'\Ib)^{-1}.
\]
By cyclicity of the trace, for each \(j\),
\begin{equation}
\begin{aligned}
\Tr(\Tb_j^2)
&= \Tr\!\Big( (\Qb+\lambda'\Ib)^{-1/2}\Delta_j(\Qb+\lambda'\Ib)^{-1}\Delta_j(\Qb+\lambda'\Ib)^{-1/2}\Big) \\
&= \Tr\!\big(\Delta_j\,\Jb\,\Delta_j\,\Jb\big),
\end{aligned}
\label{eq:trace_Aj2}
\end{equation}
where \(\Delta_j := \Qb(a'_j)-\Qb\).
Since \(\{a'_j\}_{j=1}^n\) are i.i.d., we have
\begin{equation}
\Tr\!\Big(\sum_{j=1}^n \EE[\Tb_j^2]\Big)
= n\,\EE\Tr(\Tb_1^2)
= n\,\EE\Tr\!\big(\Delta\,\Jb\,\Delta\,\Jb\big), \label{eq:trace_sum_A}
\end{equation}
with generic \(\Delta := \Qb(a')-\Qb\) for \(a'\sim \cP_{\mathrm{samp}}\).

Using Hilbert--Schmidt (HS) norms,
\[
\Tr(\Delta \Jb \Delta \Jb)=\|\Jb^{1/2}\Delta \Jb^{1/2}\|_{\mathrm{HS}}^2,
\]
and the inequality \(\|U-V\|_{\mathrm{HS}}^2 \le 2\|U\|_{\mathrm{HS}}^2+2\|V\|_{\mathrm{HS}}^2\) with
\(U=\Jb^{1/2}\Qb(a')\Jb^{1/2}\) and \(V=\Jb^{1/2}\Qb\Jb^{1/2}\), we obtain
\begin{equation}
\Tr(\Delta \Jb \Delta \Jb)
\le 2\,\Tr\!\big(\Qb(a')\,\Jb\,\Qb(a')\,\Jb\big)
+2\,\Tr\!\big(\Qb\,\Jb\,\Qb\,\Jb\big). \label{eq:split_A}
\end{equation}
We now bound \(\Tr(A\Jb A\Jb)\) for a generic positive semidefinite operator \(A\succeq 0\):
\begin{equation}
\begin{aligned}
\Tr(A\Jb A\Jb)
&= \Tr\!\Big( (\Jb^{1/2} A \Jb^{1/2})^2 \Big) \\
&\le \|\Jb^{1/2} A \Jb^{1/2}\|_{\operatorname{op}}\; \Tr(\Jb^{1/2} A \Jb^{1/2}) \\
&= \|\Jb^{1/2} A \Jb^{1/2}\|_{\operatorname{op}}\; \Tr(A \Jb) \\
&\le \|A\|_{\operatorname{op}}\,\|\Jb\|_{\operatorname{op}}\,\Tr(A \Jb) \\
&\le \frac{\|A\|_{\operatorname{op}}}{\lambda'}\,\Tr\!\big(A(\Qb+\lambda'\Ib)^{-1}\big),
\end{aligned}
\label{eq:AMA_bound_A}
\end{equation}
since \(\|\Jb\|_{\operatorname{op}}=\|(\Qb+\lambda'\Ib)^{-1}\|_{\operatorname{op}}\le 1/\lambda'\).

Under the bounded-feature assumption \(\|\psi(x,a)\|_{\mathbb{F}}^2\le \xi\), we have
\(\|\Qb(a)\|_{\operatorname{op}}\le \xi\) and \(\|\Qb\|_{\operatorname{op}}\le \xi\).
Applying \eqref{eq:AMA_bound_A} with \(A=\Qb(a')\) and \(A=\Qb\), and then taking expectation over \(a'\), yields
\begin{equation}
\begin{aligned}
\EE\Tr\!\big(\Qb(a')\Jb\Qb(a')\Jb\big)
&\le \frac{\xi}{\lambda'}\,\EE\Tr\!\big(\Qb(a')(\Qb+\lambda'\Ib)^{-1}\big)
= \frac{\xi}{\lambda'}\,\Tr\!\big(\Qb(\Qb+\lambda'\Ib)^{-1}\big),
\end{aligned}
\label{eq:term1_A}
\end{equation}
\begin{equation}
\begin{aligned}
\Tr\!\big(\Qb\Jb\Qb\Jb\big)
&\le \frac{\xi}{\lambda'}\,\Tr\!\big(\Qb(\Qb+\lambda'\Ib)^{-1}\big).
\end{aligned}
\label{eq:term2_A}
\end{equation}

Combining \eqref{eq:trace_sum_A}, \eqref{eq:split_A}, \eqref{eq:term1_A}, and \eqref{eq:term2_A}, we conclude that
\begin{equation}
\Tr\!\Big(\sum_{j=1}^n \EE[\Tb_j^2]\Big)
\le \frac{4n\xi}{\lambda'}\,\Tr\!\big(\Qb(\Qb+\lambda'\Ib)^{-1}\big)
\label{eq:trace_final_deff_A}
\end{equation}
Moreover, since \(\Tr(\Qb)=\EE\|\psi\|_{\mathbb{F}}^2 \le \xi\), we have
\(\Tr(\Qb)\|(\Qb+\lambda'\Ib)^{-1}\|_{\operatorname{op}}\le \xi/\lambda'\), and thus the cruder bound
\begin{equation}
\Tr\!\Big(\sum_{j=1}^n \EE[\Tb_j^2]\Big)
\le \frac{4n\xi^2}{\lambda'^2}. \label{eq:trace_final_crude_A}
\end{equation}

By Lemma D.3 of \citet{wang2026pseudo}, for any \(\lambda' \geq c_1 \xi \frac{\log n}{n}\) for some constant \(c_1>1\), we have with probability at least \(1-n^{-11}\),
\begin{equation*}
\| (\Qb +\lambda' \Ib)^{-\frac{1}{2}}(\bar{\Qb} -\Qb) (\Qb +\lambda' \Ib)^{-\frac{1}{2}}\|_{\operatorname{op}} \lesssim \frac{1}{2}.
\end{equation*}
Hence, \(\bar \Qb \lesssim 2(\Qb + \lambda' \Ib)\), and therefore
\begin{equation}
\bar \Qb \preceq c(\Qb + \frac{\log n}{n} \Ib)
\end{equation}
holds for some absolute constant \(c>1\) with probability at least \(1-n^{-11}\).
\hfill\BlackBox

\begin{definition}[Good Event $\Ecr_2$]
We define the good event \(\Ecr_2\) as the event where \cref{equation; good event 2-1} holds for all \(j \in [n]\) and \cref{equation; good event 2-2} holds.
By the above argument, we have \(\PP[\Ecr_2] \geq 1-2n^{-11}\).
\end{definition}

\paragraph{Good Event 3: \(\Ecr_3\).}

Using Lemma~\ref{lemma: trace class concentration bounded}, with probability at least \(1-\frac{1}{n^{11}}\), for some absolute constant \(c>1\), we have
\begin{equation}\label{equation; good event 3-1}
\frac{1}{c}(\Ab^\top\Ab + \log n \cdot \Ib) \preceq n \bSigma_{\operatorname{samp}}+ \log n \cdot \Ib \preceq  c(\Ab^\top\Ab  + \log n \cdot \Ib).
\end{equation}
Similarly, using Lemma~\ref{lemma: trace class concentration bounded}, with probability at least \(1-\frac{1}{n^{11}}\), we have
\begin{equation}\label{equation; good event 3-2}
\frac{1}{c}(\Zb^\top \Zb + \log n \cdot \Ib )\preceq n \EE_{(x_i, a_i) \sim \cP_{x,a}}[\psi(x_i, a_i)\psi(x_i, a_i)^\top]+ \log n \cdot \Ib \preceq c ( \Zb^\top \Zb+ \log n \cdot \Ib)
\end{equation}
for some absolute constant \(c>1\).

\begin{definition}[Good Event $\Ecr_3$]
We define the good event \(\Ecr_3\) as the event where \cref{equation; good event 3-1} and \cref{equation; good event 3-2} both hold.
By the previous argument, we have \(\PP[\Ecr_3] \ge 1-2n^{-11}\).
\end{definition}

\paragraph{Final Good Event: \(\Ecr_{\operatorname{good}}\).}
Lastly, we define our final good event as the intersection of the constructed good events \(\Ecr_1, \Ecr_2\) and \(\Ecr_3\).

\begin{definition}[Final Good Event]
We define \(\Ecr_{\operatorname{good}} := \Ecr_1 \cap \Ecr_2 \cap \Ecr_3\).
We emphasize that this event depends only on the randomness of \((x_i, a_i)\) for \(i \in [n]\) and \(a_j'\) for \(j \in [n]\), not on the randomness of the noise variables.
\end{definition}

\begin{lemma}[Good Event]\label{lemma: good event}
We have
\begin{equation*}
\PP[\Ecr_{\operatorname{good}}] \geq 1- \frac{5}{n^{11}}.
\end{equation*}  
\end{lemma}

\noindent\textbf{Proof of Lemma~\ref{lemma: good event}. }
This follows directly from the previous arguments and by taking a union bound. 
\hfill\BlackBox

\paragraph{Key Inequality under Event \(\Ecr_{\operatorname{good}}\).}
Next, we present a key property of the design operator \(\Wb\) under the final good event \(\Ecr_{\operatorname{good}}\).

\begin{lemma}[Upper Bound of \(\Wb^\top \Wb\)]\label{lemma: key inequality 1}
Under the event \(\Ecr_{\operatorname{good}}\), we have
\begin{equation*}
\frac{1}{n} \Wb^\top{\Wb} \preceq c R\times \frac{1}{n}(\Zb^\top \Zb + \log n \Ib)
\end{equation*}
for some absolute constant \(c>1\).
\end{lemma}

\noindent\textbf{Proof of Lemma~\ref{lemma: key inequality 1}. }
We first observe that by Jensen's inequality (equivalently, Cauchy--Schwarz),
\begin{equation*}
\frac{1}{n}\Wb^\top \Wb = \frac{1}{n}\sum_{j=1}^n \hat{\bar \psi}(a_j^\prime)\hat{\bar \psi}(a_j^\prime)^\top \preceq \frac{1}{n^2} \sum_{i \in [n], j \in [n]} \psi(x_i, {a}^\prime_j)\psi(x_i, {a}^\prime_j)^\top.
\end{equation*}
Indeed, for any \(u \in  \mathbb{F}\),
\begin{equation*}
\begin{aligned}
u^\top \Wb^\top \Wb u
&= \sum_{j=1}^n\left(\frac{1}{n}\sum_{i=1}^n \psi(x_i, a_j^\prime)^\top u\right)^2 \\
&\leq \sum_{j=1}^n \frac{1}{n}\sum_{i=1}^n (\psi(x_i,  a_j^\prime)^\top u)^2.
\end{aligned}
\end{equation*}
Using the definition of \(\Ecr_{\operatorname{good}}\), there exist two absolute constants \(c_1, c_2>1\) such that 
\begin{equation*}
\begin{aligned}
\frac{1}{n^2}\sum_{i \in [n],j \in [n]} \psi(x_i, a_j^\prime)\psi(x_i, a_j^\prime)^\top
&=  \frac{1}{n} \sum_{j=1}^n \frac{1}{n}  \sum_{i \in [n]} \psi(x_i, a_j^\prime)\psi(x_i, a_j^\prime)^\top  \\
&=\frac{1}{n}\sum_{j=1}^n \widehat{\Qb}(a_j^\prime)  \\
&\preceq c_1\frac{1}{n}\sum_{j=1}^n \big({\Qb}(a_j^\prime)  +\frac{\log n}{n}\Ib \big) \quad (\text{by definition of } \Ecr_{\operatorname{good}}) \\
&\preceq c_1\frac{\log n}{n}\Ib + c_1 \frac{1}{n}\sum_{j=1}^n {\Qb}(a_j^\prime) \\
&=c_1\frac{\log n}{n}\Ib  + c_1\bar{\Qb} \\
&\preceq c_2 \frac{\log n}{n}\Ib  + c_2\Qb  \quad( \text{by definition of } \Ecr_{\operatorname{good}})\\
&= c_2\frac{\log n}{n}\Ib +c_2\EE_{ \cP_x \otimes \cP_{\operatorname{samp}} }[\psi(x,a)\psi(x,a)^\top] \\
&\preceq c_2 \frac{\log n}{n}\Ib +c_2R\EE_{(x,a)\sim \cP_{x,a}}[\psi(x,a)\psi(x,a)^\top] \quad \text{(Assumption~\ref{assumption: first overlap})}.
\end{aligned}
\end{equation*}
On the other hand, under the good event \(\Ecr_{\operatorname{good}}\),
\begin{equation*}
c\Big(\frac{1}{n}\Zb^\top \Zb +\frac{\log n}{n}\Ib \Big) \succeq \EE_{(x, a) \sim \cP_{x,a}}[\psi(x, a)\psi(x,a)^\top ]+\frac{\log n}{n}\Ib .
\end{equation*}
By combining these, we obtain the desired result.
\hfill\BlackBox

\begin{remark}[A Useful Observation]\label{remark; useful observation}
For any two design operators \(\Tb_1, \Tb_2\), suppose that
\begin{equation*}
\Tb_1^\top \Tb_1   \preceq c_1(\Tb_2^\top \Tb_2 + c_2\log n \Ib)
\end{equation*}
for some absolute constants \(c_1>1, c_2 >0\). 
Then, for any \(c_3 >0\), for \(c = c_1\max(1, c_2 / c_3)\), we have
\begin{equation*}
\Tb_1^\top \Tb_1 \preceq c(\Tb_2^\top \Tb_2 + c_3 \log n \Ib).
\end{equation*}
\end{remark}

\paragraph{How We Use Remark~\ref{remark; useful observation}.}
The concentration inequalities above naturally produce additive shifts of the form \(+\log n\,\Ib\).
Later we need to work instead with \(+n\lambda\,\Ib\) or \(+n\lambda_0\,\Ib\).
We therefore invoke Remark~\ref{remark; useful observation} with \(c_3 = n\lambda/\log n\) or \(c_3 = n\lambda_0/\log n\).
Because the theorem assumes \(\lambda \gtrsim \log n / n\) and we choose \(\lambda_0 \asymp \log n / n\), these ratios are bounded below by absolute constants, so the resulting multiplicative losses remain absolute.

\subsection{Error Decomposition of Pseudo-Outcomes}\label{section:apdx decomposition}
For any \(a \in \cA\), we aim to decompose the estimation error of our pseudo-outcome \(m(a)\):
\begin{equation*}
m(a)- \phi(a)^\top \eta^\star =m(a)- h^\star(a) =m(a) - \EE_{x \sim \cP_x}[f^\star(x,a)] = m(a)-\EE_{x \sim \cP_x}[{\psi}(x, a)]^\top \theta^\star.
\end{equation*}
Also, we see that 
\begin{equation*}
m(a) = \frac{1}{n} \sum_{i=1}^n \hat{f}(x_i, a) = \hat{\bar \psi}(a)^\top \hat \theta.
\end{equation*}
Then, we get
\begin{equation}\label{equation; error decomposition}
\begin{aligned}
m(a) - \phi(a)^\top \eta^\star &= \frac{1}{n}\sum_{i=1}^n \psi(x_i, a)^\top \hat\theta-  \frac{1}{n}\sum_{i=1}^n \psi(x_i, a)^\top \theta^\star + \frac{1}{n} \sum_{i=1}^n \psi(x_i, a)^\top \theta^\star - \EE_{x \sim \cP_x}[{\psi}(x, a)]^\top \theta^\star \\
&= {\hat{\bar \psi }(a)^\top \hat\theta-  \hat{\bar \psi }(a)^\top \theta^\star } + \underbrace{\hat{\bar \psi }(a)^\top \theta^\star - \bar \psi(a)^\top \theta^\star }_{:=b(a)}\\
&= {\hat{\bar \psi }(a)^\top (\hat\theta- \theta^\star) }  +b(a)
\end{aligned}
\end{equation}
We define \(\bb\) as the vector \((b(a_1'), \dots, b(a_n'))\) where
\begin{equation*}
\begin{aligned}
b(a_j') &= \hat{\bar \psi }(a_j')^\top \theta^\star - \bar \psi(a_j')^\top \theta^\star \\
&=\frac{1}{n}\sum_{i=1}^n f^\star(x_i,a_j') -\EE_{x \sim \cP_x}[f^\star(x,a_j')].
\end{aligned}
\end{equation*}
By \cref{equation; good event 1}, under the good event \(\Ecr_{\operatorname{good}}\), for all \(a_j' \in \{{a}^\prime_j\}_{j=1}^n\) we have
\begin{equation*}
|b(a_j')| \lesssim \|\theta^\star\|_{\mathbb{F}}\frac{\sqrt{\log n}}{\sqrt{n}}.
\end{equation*}
\begin{lemma}[Error Decomposition]\label{lemma: error decomposition}
Set \(\hat{\eta}_\lambda\) as the corresponding Hilbertian element of \(\hat{h}_\lambda\) in \(\HH\) (explicitly defined in Appendix~\ref{section:apdx groundwork}).
Then we have the following decomposition:
\begin{equation*}
\begin{aligned}
\hat\eta_\lambda - \eta^\star  &= (\Ab^\top \Ab + n\lambda \Ib)^{-1} \Ab^\top \bb +(\Ab^\top \Ab + n\lambda \Ib)^{-1} \Ab^\top \Wb (\hat\theta -\theta^\star) -n\lambda(\Ab^\top \Ab + n\lambda\Ib)^{-1} \eta^\star  \\
&:= \Bcr + \Pcr + \Ccr.
\end{aligned}
\end{equation*}
\end{lemma}

\noindent\textbf{Proof of Lemma~\ref{lemma: error decomposition}. }
By the established error decomposition of
\begin{equation*}
m(a) - \phi(a)^\top \eta^\star =\hat{\bar \psi }(a)^\top (\hat\theta-    \theta^\star )  +b(a),
\end{equation*}
we get
\begin{equation*}
\begin{aligned}
\hat\eta_\lambda - \eta^\star &= (\Ab^\top \Ab + n\lambda \Ib)^{-1} \left(\sum_{j=1}^n \phi({a}^\prime_j) m({a}^\prime_j)- \Ab^\top \Ab \eta^\star \;- n\lambda \eta^\star \right)\\
&=(\Ab^\top \Ab +n\lambda \Ib)^{-1} \left(\sum_{j=1}^n \phi({a}^\prime_j) (m(a_j') - \phi({a}^\prime_j)^\top \eta^\star)\; \,-n\lambda \eta^\star \right)\\
&=(\Ab^\top \Ab +n\lambda \Ib)^{-1} \left(\sum_{j=1}^n \phi({a}^\prime_j) (m(a_j') - h^\star(a_j'))\; \,-n\lambda \eta^\star \right)\\
&\stackrel{(i)}{=}(\Ab^\top \Ab + n\lambda\Ib)^{-1} \sum_{j=1}^n \phi({a}^\prime_j) \big(b({a}^\prime_j) + \hat{\bar \psi }({a}^\prime_j)^\top (\hat{\theta} -\theta^\star)\big)-n\lambda(\Ab^\top \Ab +n\lambda \Ib)^{-1} \eta^\star \\
&=\underbrace{(\Ab^\top \Ab + n\lambda \Ib)^{-1} \Ab^\top \bb }_{:=\Bcr}+\underbrace{(\Ab^\top \Ab + n\lambda\Ib)^{-1} \Ab^\top \Wb (\hat\theta -\theta^\star) }_{\Pcr} + \underbrace{(-n\lambda(\Ab^\top \Ab + n\lambda \Ib)^{-1} \eta^\star)}_{\Ccr} \\
&:= \Bcr + \Pcr + \Ccr.
\end{aligned}
\end{equation*}
Step (i) holds by the established decomposition \eqref{equation; error decomposition}.
\hfill\BlackBox

Recall that we defined 
\begin{equation*}
\begin{aligned}
\Bcr &: =(\Ab^\top \Ab + n\lambda \Ib)^{-1} \Ab^\top \bb \\
\Pcr &:= (\Ab^\top \Ab + n\lambda \Ib)^{-1} \Ab^\top \Wb (\hat\theta -\theta^\star) \\
\Ccr &:= -n\lambda(\Ab^\top \Ab + n\lambda\Ib)^{-1} \eta^\star  
\end{aligned}
\end{equation*}
Here, \(\Pcr\), the propagated error term, is our main challenge.

\subsection{Bounding MISE with Respect to \texorpdfstring{\(\mathbf{\Sigma}_{\operatorname{ref}}\)}{Sigma-ref}}
Recall that \(\hat\eta_{\lambda}\) is the Hilbertian element of \(\hat h_\lambda\) and we define 
\begin{equation*}
\cE_{\operatorname{ref}}(\hat{h}_\lambda) = \|\hat \eta_\lambda - \eta^\star\|_{\bSigma_{\operatorname{ref}}}^2 := \cE_{\HH}(\hat \eta_\lambda) .
\end{equation*}
Using the established decomposition, we get
\begin{equation*}
\cE_{\HH}(\hat \eta_\lambda)
\lesssim \|\bSigma_{\operatorname{ref}}^{\frac{1}{2}} (\Bcr+ \Ccr +\Pcr) \|_{\HH}^2 \lesssim \|\bSigma_{\operatorname{ref}}^{\frac{1}{2}} \Bcr\|_{\HH}^2+\|\bSigma_{\operatorname{ref}}^{\frac{1}{2}} \Ccr \|_{\HH}^2 +\|\bSigma_{\operatorname{ref}}^{\frac{1}{2}} \Pcr\|_{\HH}^2.
\end{equation*}

We define the following important second-moment operator in \(\HH\):
\begin{equation*}
\Sb_{\lambda} := (\bSigma_{\operatorname{samp}} + \lambda \Ib)^{-\frac{1}{2}} \bSigma_{\operatorname{ref}} (\bSigma_{\operatorname{samp}} + \lambda \Ib)^{-\frac{1}{2}}.
\end{equation*}

\paragraph{Some Useful Inequalities.}
Recall the definition of the good event \(\Ecr_{\operatorname{good}}\).
By the theorem assumption, \(n\lambda \gtrsim \log n\), and by construction \(\lambda_0 \asymp \log n / n\).
Therefore Remark~\ref{remark; useful observation} allows us to replace the shifts \(+\log n\,\Ib\) from the concentration events by \(+n\lambda\,\Ib\) and \(+n\lambda_0\,\Ib\), up to absolute constants.
Under the event \(\Ecr_{\operatorname{good}}\), we thus have for some \(c>1\),
\begin{equation*}
\frac{1}{c}(n\bSigma_{\operatorname{samp}} + n \lambda \Ib) \preceq \Ab^\top \Ab + n \lambda \Ib \preceq c(n\bSigma_{\operatorname{samp}} + n \lambda \Ib),
\end{equation*}
and similarly,
\begin{equation*}
\frac{1}{c}\big(n\EE_{(x, a) \sim \cP_{x,a}}[\psi(x, a)\psi(x,a)^\top ] + n\lambda_0 \Ib\big)
\preceq \Zb^\top \Zb + n\lambda_0 \Ib
\preceq c\big(n\EE_{(x, a) \sim \cP_{x,a}}[\psi(x, a)\psi(x,a)^\top ] + n\lambda_0 \Ib\big).
\end{equation*}
We also observe that
\begin{equation*}
\begin{aligned}
\| (\tfrac{1}{n}\Ab^\top \Ab + \lambda \Ib)^{-\frac{1}{2}} \bSigma_{\operatorname{ref}}  (\tfrac{1}{n}\Ab^\top \Ab + \lambda \Ib)^{-\frac{1}{2}}\|_{\operatorname{op}}
&=\| \bSigma_{\operatorname{ref}}^{\frac{1}{2}}  (\tfrac{1}{n}\Ab^\top \Ab + \lambda \Ib)^{-1}\bSigma_{\operatorname{ref}}^{\frac{1}{2}}\|_{\operatorname{op}} \\
&\lesssim \| \bSigma_{\operatorname{ref}}^{\frac{1}{2}}  (\bSigma_{\operatorname{samp}} + \lambda \Ib)^{-1}\bSigma_{\operatorname{ref}}^{\frac{1}{2}}\|_{\operatorname{op}} \\
&= \|\Sb_{\lambda} \|_{\operatorname{op}}.
\end{aligned}
\end{equation*}
Hence,
\begin{equation}\label{equation; key 2}
(\tfrac{1}{n}\Ab^\top \Ab + \lambda \Ib)^{-\frac{1}{2}} \bSigma_{\operatorname{ref}}  (\tfrac{1}{n}\Ab^\top \Ab + \lambda \Ib)^{-\frac{1}{2}} \preceq c\|\Sb_{\lambda} \|_{\operatorname{op}} \Ib
\end{equation}
for some absolute constant \(c>1\).

\paragraph{Bounding Term \(\Bcr\).}
We first bound the pseudo-outcome bias term \(\Bcr\).
Observe that
\begin{equation*}
\begin{aligned}
\| \bSigma_{\operatorname{ref}}^{\frac{1}{2}}\Bcr \|_{\mathbb{H}}^2
&\lesssim \bb^\top \Ab^\top (\Ab^\top \Ab + n\lambda \Ib)^{-1} \bSigma_{\operatorname{ref}} (\Ab^\top \Ab + n\lambda \Ib)^{-1} \Ab \bb  \\
&\lesssim \frac{1}{n} \|\Sb_{\lambda}\|_{\operatorname{op}} \cdot \bb^\top \Ab^\top (\Ab^\top \Ab + n\lambda \Ib)^{-1}  \Ab \bb \quad \text{(by \cref{equation; key 2})} \\
&\lesssim \frac{1}{n} \|\Sb_{\lambda}\|_{\operatorname{op}} \|\bb\|_2^2 \\
&\lesssim \frac{\log n}{n} \|\Sb_{\lambda}\|_{\operatorname{op}}\|\theta^\star\|_{\mathbb{F}}^2,
\end{aligned}
\end{equation*}
where the last inequality follows from \cref{equation; good event 1}, which implies \(\|\bb\|_2^2 \lesssim \log n \|\theta^\star\|_{\mathbb{F}}^2\).

\paragraph{Bounding Term \(\Ccr\).}
Observe that
\begin{equation*}
\begin{aligned}
\|\bSigma_{\operatorname{ref}}^{\frac{1}{2} } \Ccr \|_\HH^2
&=(n\lambda)^2\|(\Ab^\top \Ab + n\lambda \Ib)^{-1} \eta^\star \|_{\bSigma_{\operatorname{ref}}}^2 \\
&= \lambda^2 (\eta^\star)^\top (\tfrac{1}{n}\Ab^\top \Ab + \lambda \Ib)^{-1}\bSigma_{\operatorname{ref}}(\tfrac{1}{n}\Ab^\top \Ab + \lambda \Ib)^{-1} \eta^\star\\
&\lesssim \lambda^2 \|\Sb_{\lambda}\|_{\operatorname{op}} (\eta^\star)^\top (\tfrac{1}{n}\Ab^\top \Ab + \lambda \Ib)^{-1} \eta^\star \quad \text{(by \cref{equation; key 2})}\\
&\lesssim \lambda \|\eta^\star\|_{\HH}^2 \|\Sb_\lambda\|_{\operatorname{op}}.
\end{aligned}
\end{equation*}

Recall that we defined 
\begin{equation*}
\Pcr := (\Ab^\top \Ab + n\lambda\Ib)^{-1} \Ab^\top \Wb (\hat\theta -\theta^\star).
\end{equation*}
We decompose it into 
\begin{equation*}
\begin{aligned}
\Pcr &= (\Ab^\top \Ab + n\lambda\Ib)^{-1} \Ab^\top \Wb (\Zb^\top \Zb+ n\lambda_0 \Ib)^{-1}\Zb^\top \bm \varepsilon - n\lambda_0(\Ab^\top \Ab + n\lambda\Ib)^{-1} \Ab^\top \Wb (\Zb^\top \Zb+ n\lambda_0 \Ib)^{-1} \theta^\star \\
&:= \Pcr_v + \Pcr_b
\end{aligned}
\end{equation*}
for 
\begin{equation*}
\begin{aligned}
\Pcr_v &:=  (\Ab^\top \Ab + n\lambda\Ib)^{-1} \Ab^\top \Wb (\Zb^\top \Zb+ n\lambda_0 \Ib)^{-1}\Zb^\top \bm \varepsilon, \\
\Pcr_b &:=  -n\lambda_0(\Ab^\top \Ab + n\lambda\Ib)^{-1} \Ab^\top \Wb (\Zb^\top \Zb+ n\lambda_0 \Ib)^{-1} \theta^\star,
\end{aligned}
\end{equation*}
and we bound each term in turn.

\paragraph{Bounding Term \(\Pcr_v\) (Propagated Variance).}
Under the event \(\Ecr_{\operatorname{good}}\), with probability at least \(1-n^{-11}\) over the randomness of \(\bm{\varepsilon}\), Lemma~\ref{lemma: quadratic martingale} gives
{\small \begin{equation*}
\begin{aligned}
&\| \bSigma_{\operatorname{ref}}^{\frac{1}{2}}  \Pcr_v \|_{\HH}^2  \\
&= \|\bSigma_{\operatorname{ref}}^{\frac{1}{2}} (\Ab^\top \Ab + n\lambda \Ib)^{-1} \Ab^\top  \Wb(\Zb^\top \Zb+ n\lambda_0 \Ib)^{-1}\Zb^\top \bm \varepsilon \|^2_{\mathbb{H}} \\
&\lesssim \Tr\left(\bSigma_{\operatorname{ref}}^{\frac{1}{2}}(\Ab^\top \Ab + n\lambda \Ib)^{-1} \Ab^\top  \Wb (\Zb^\top \Zb+ n\lambda_0 \Ib)^{-1}\Zb^\top \Zb (\Zb^\top \Zb+ n\lambda_0 \Ib)^{-1} \Wb^\top \Ab (\Ab^\top \Ab + n\lambda \Ib)^{-1} \bSigma_{\operatorname{ref}}^{\frac{1}{2}}\right) \log n\\
&\lesssim \Tr\left(\bSigma_{\operatorname{ref}}^{\frac{1}{2}} (\Ab^\top \Ab + n\lambda \Ib)^{-1} \Ab^\top  \Wb (\Zb^\top \Zb+ n\lambda_0 \Ib)^{-1} \Wb^\top \Ab(\Ab^\top \Ab + n\lambda \Ib)^{-1} \bSigma_{\operatorname{ref}}^{\frac{1}{2}}\right) \log n \\
&\stackrel{\text{(i)}}{\lesssim }R\Tr\left(\bSigma_{\operatorname{ref}}^{\frac{1}{2}} (\Ab^\top \Ab + n\lambda \Ib)^{-1} \Ab\Ab^\top (\Ab^\top \Ab + n\lambda \Ib)^{-1}\bSigma_{\operatorname{ref}}^{\frac{1}{2}} \right) \log n \\
&\lesssim R\Tr\left(\bSigma_{\operatorname{ref}} (\Ab^\top \Ab + n\lambda \Ib)^{-1} \right) \log n\\
&\stackrel{\text{(ii)}}{\lesssim} \frac{R\log n}{n} \Tr(\Sb_{\lambda}).
\end{aligned}
\end{equation*}}
Step (i) uses Lemma~\ref{lemma: key inequality 1} together with Remark~\ref{remark; useful observation}, which gives
\[
\Wb(\Zb^\top \Zb+ n\lambda_0 \Ib)^{-1}\Wb^\top \preceq cR \Ib
\]
because \(n\lambda_0 \asymp \log n\).
Step (ii) uses the comparison between \(\Ab^\top\Ab + n\lambda \Ib\) and \(n(\bSigma_{\operatorname{samp}}+\lambda \Ib)\) established above.

\paragraph{Bounding Term \(\Pcr_b\) (Propagated Bias).}
Next, observe that
\begin{equation*}
\begin{aligned}
&\| \bSigma_{\operatorname{ref}}^{\frac{1}{2}}  \Pcr_b \|_{\HH} \\
&\leq
n\lambda_0\|\bSigma_{\operatorname{ref}}^{\frac{1}{2}} (\Ab^\top \Ab + n\lambda \Ib)^{-1} \Ab^\top  \Wb(\Zb^\top \Zb+ n\lambda_0 \Ib)^{-1} \theta^\star \|_{\mathbb{H}} \\
&\lesssim
n\lambda_0  \|\bSigma_{\operatorname{ref}}^{\frac{1}{2}} (\Ab^\top \Ab + n\lambda \Ib)^{-\frac{1}{2}} \|_{\operatorname{op}}\cdot \|(\Ab^\top \Ab + n\lambda \Ib)^{-\frac{1}{2}} \Ab^\top  \|_{\operatorname{op}} \| \Wb(\Zb^\top \Zb+ n\lambda_0 \Ib)^{-1} \theta^\star \|_{2} \\
&\lesssim
n\lambda_0 \cdot  \frac{1}{\sqrt{n}}\|\Sb_{\lambda}\|_{\operatorname{op}}^{\frac{1}{2}}\cdot \| \Wb(\Zb^\top \Zb+ n\lambda_0 \Ib)^{-1} \theta^\star \|_{2} \quad \text{(by \cref{equation; key 2})} \\
&\lesssim
n\lambda_0 \cdot  \frac{1}{\sqrt{n}}\|\Sb_{\lambda}\|_{\operatorname{op}}^{\frac{1}{2}}\cdot  \| \Wb(\Zb^\top \Zb+ n\lambda_0 \Ib)^{-1/2} \|_{\operatorname{op}} \| (\Zb^\top \Zb+ n\lambda_0 \Ib)^{-1/2} \theta^\star \|_{\FF} \\
&\lesssim
\sqrt{n\lambda_0 }\cdot  \frac{1}{\sqrt{n}}\|\Sb_{\lambda}\|_{\operatorname{op}}^{\frac{1}{2}}\cdot  \| \Wb(\Zb^\top \Zb+ n\lambda_0 \Ib)^{-1/2} \|_{\operatorname{op}} \|\theta^\star \|_{\mathbb{F}} .
\end{aligned}
\end{equation*}
By Lemma~\ref{lemma: key inequality 1} and Remark~\ref{remark; useful observation}, since \(\lambda_0 \asymp \frac{\log n}{n}\), we have
\begin{equation*}
\| (\Zb^\top \Zb+ n\lambda_0 \Ib)^{-1/2}\Wb^\top \Wb(\Zb^\top \Zb+ n\lambda_0 \Ib)^{-1/2} \|_{\operatorname{op}}
\lesssim  R.
\end{equation*}
Thus,
\begin{equation*}
\begin{aligned}
\| \bSigma_{\operatorname{ref}}^{\frac{1}{2}}  \Pcr_b \|^2_{\HH}
&\lesssim
n\lambda_0 \cdot  \frac{1}{n}\|\Sb_{\lambda}\|_{\operatorname{op}}\cdot R \|\theta^\star \|_{\mathbb{F}}^2 \\
&\lesssim \frac{R\log n}{n}\|\Sb_{\lambda}\|_{\operatorname{op}}\|\theta^\star \|_{\mathbb{F}}^2 \qquad( \text{since }\lambda_0 \asymp \tfrac{\log n}{n}).
\end{aligned}
\end{equation*}

Since replacing \(R\) by \(R \vee 1\) preserves Assumption~\ref{assumption: first overlap}, we may assume \(R\ge 1\).
Therefore, gathering all the above results, we obtain
\begin{equation}\label{equation; general theory MISE bound}
\begin{aligned}
\cE_{\operatorname{ref}}(\hat h_\lambda)
&\lesssim \lambda \|\Sb_{\lambda} \|_{\operatorname{op}} \| \eta^\star\|^2_{\HH}  +\frac{R\log n}{n} \Tr(\Sb_{\lambda}) +\frac{R\log n}{n} \|\Sb_{\lambda}\|_{\operatorname{op}} \|\theta^\star \|_{\mathbb{F}}^2 \\
& \lesssim  \lambda \|\Sb_{\lambda} \|_{\operatorname{op}} \| h^\star\|^2_{\cH}  +\frac{R\log n}{n} \Tr(\Sb_{\lambda}) +\frac{R\log n}{n} \|\Sb_{\lambda}\|_{\operatorname{op}} \|f^\star \|_{\cF}^2.
\end{aligned}
\end{equation}
Under Assumption~\ref{assumption: second overlap}, we have
\begin{equation*}
\Sb_\lambda \preceq B(\lambda)\Ib,
\qquad
\|\Sb_\lambda \|_{\operatorname{op}} \le B(\lambda).
\end{equation*}
Also,
\begin{equation*}
\Tr(\Sb_\lambda) = \Tr( (\bSigma_{\operatorname{samp}} + \lambda \Ib)^{-1} \bSigma_{\operatorname{ref}} ).
\end{equation*}
By Assumption~\ref{assumption: second overlap},
\begin{equation*}
\bSigma_{\operatorname{samp}} + \lambda \Ib \succeq \frac{1}{2B(\lambda)}( \bSigma_{\operatorname{ref}} + B(\lambda)  \lambda \Ib),
\end{equation*}
because \(\bSigma_{\operatorname{samp}} + \lambda \Ib \succeq B(\lambda)^{-1}\bSigma_{\operatorname{ref}}\) and \(\bSigma_{\operatorname{samp}} + \lambda \Ib \succeq \lambda \Ib\).
Thus
\begin{equation*}
\Tr(\Sb_\lambda) = \Tr( (\bSigma_{\operatorname{samp}} + \lambda\Ib)^{-1}\bSigma_{\operatorname{ref}} ) \lesssim  B(\lambda) \Tr((\bSigma_{\operatorname{ref}} + B(\lambda) \lambda\Ib)^{-1} \bSigma_{\operatorname{ref}}).
\end{equation*}
By definition of \(\Gamma(\cdot)\),
\begin{equation*}
\Tr((\bSigma_{\operatorname{ref}} + B(\lambda) \lambda\Ib)^{-1}\bSigma_{\operatorname{ref}}) \le \Gamma(B(\lambda)\lambda).
\end{equation*}
Hence,
\begin{equation*}
\Tr(\Sb_\lambda) \lesssim B(\lambda) \cdot \Gamma(B(\lambda)\lambda ).
\end{equation*}
Combining this with \cref{equation; general theory MISE bound} yields the desired upper bound.
Finally, intersecting \(\Ecr_{\operatorname{good}}\) with the event from Lemma~\ref{lemma: quadratic martingale} used for \(\Pcr_v\) and taking a union bound gives the stated probability \(1-n^{-10}\).

\section{Proof of Theorem~\ref{theorem: general theory model selection}}\label{section:apdx proof model selection}

\subsection{Two Key Propositions for Model Selection}
Recall that when we run Algorithm~\ref{algorithm: model selection}, the proxy regularizer \(\tilde\lambda\) is fixed in advance and satisfies \(\tilde\lambda\gtrsim \log n/n\).
Under Assumption~\ref{assumption: second overlap}, we define the term $\Ocr$ as follows:
\begin{equation*}
\Ocr :=
B(\tilde\lambda)
\left\{
\tilde\lambda\|\eta^\star\|_{\HH}^2
+\frac{\xi\log n}{n}R(1+\|\theta^\star\|_{\FF}^2)
\right\}.
\end{equation*}
This term, $\Ocr$, will serve as the price of model selection in our analysis, up to logarithmic factors and constants depending on the fixed noise and target-radius parameters.

For any estimator \(\hat h\) and its corresponding Hilbertian element \(\hat{\eta}\), we define the \textbf{in-sample MISE} as the MSE among \(\tilde{a}_j, j \in [n]\), explicitly
\begin{equation*}
\cE^{\operatorname{in}}_{\HH}(\hat\eta) := \frac{1}{n} \sum_{j=1}^n |\phi(\tilde a_j)^\top (\hat \eta -\eta^\star)|^2.
\end{equation*}
We present two oracle inequalities.
The first is an oracle inequality for the in-sample MISE defined above.
\begin{proposition}[Oracle Inequality for In-Sample MISE]\label{proposition; in-sample oracle}
With probability at least \(1-2n^{-11}\),
\begin{equation*}
\cE_\HH^{\operatorname{in}}(\hat \eta_{\operatorname{ms}} ) \lesssim \min_{\lambda \in \Lambda}\cE_{\HH}^{\operatorname{in}}(\hat \eta_{\lambda} ) + \Ocr.
\end{equation*}
\end{proposition}
This proposition provides a performance guarantee for our model selection procedure with respect to the in-sample MISE.
Next, we provide a key inequality for connecting the in-sample MISE and the population MISE.
\begin{proposition}[Connecting In-Sample MISE and Population MISE]\label{proposition; in sample to population MSE}
With probability at least \(1-2n^{-11}\), the following holds uniformly for all \(\lambda\in\Lambda\):
\begin{equation*}
\begin{aligned}
\cE^{\operatorname{in}}_{\HH}(\hat\eta_\lambda)&\lesssim \cE_{\HH}(\hat\eta_\lambda)+ \Ocr, \\
\cE_{\HH}(\hat\eta_\lambda)&\lesssim \cE^{\operatorname{in}}_{\HH}(\hat\eta_\lambda)+ \Ocr .
\end{aligned}
\end{equation*}
\end{proposition}
By combining the two propositions above, we obtain our desired Theorem~\ref{theorem: general theory model selection}.

\subsection{Notations and Hilbertian Formulation for Algorithm~\ref{algorithm: model selection}}
We define some notations for the proof. 
\paragraph{Summary of Notations.}
\begin{itemize}
\item Set \(\cD_1 = \{x_{1i}, a_{1i}, y_{1i}\}_{i=1}^{n_1}\) and \(\cD_2 = \{x_{2i}, a_{2i}, y_{2i}\}_{i=1}^{n_2}\).
\item We define the design operator of \(\{\psi(x_{1i},a_{1i})\}_{i=1}^{n_1}\) in \(\cD_1\) as \(\Zb_1\).
\item We define the design operator of \(\{\psi(x_{2i},a_{2i})\}_{i=1}^{n_2}\) in \(\cD_2\) as \(\Zb_2\).
\item Set \(\{a_{1j}'\}_{j=1}^{n_1}\) as the sampled treatments from the distribution \(\cP_{\operatorname{samp}}\) for estimating the \(\hat h_\lambda\)'s.

\item Set \(\{a_{2j}'\}_{j=1}^{n_2}\) as the sampled treatments from the distribution \(\cP_{\operatorname{samp}}\) for estimating \(\tilde h\).
\item We define the design operator of \(\{\phi(a_{1j}')\}_{j=1}^{n_1}\)  as \(\Ab_1\).
\item We define the design operator of \(\{\phi(a_{2j}')\}_{j=1}^{n_2}\)  as \(\Ab_2\).
\item  We define the vector version of \(\{y_{1i}\}_{i=1}^{n_1}\) in \(\cD_1\) as \( \Yb_1\).
\item We define the vector version of \(\{y_{2i}\}_{i=1}^{n_2}\) in \(\cD_2\) as \( \Yb_2\).
\item We define the feature mean as \(\bar \psi(a) := \EE_
{x \sim \cP_x}[\psi(x,a)]\).
\item We define \(\hat{\bar{\psi}}_1(a) := \frac{1}{n_1}\sum_{i=1}^{n_1} \psi(x_{1i},a)\) and the design operator of \(\{\hat{\bar \psi}_1({a}_{1j}')\}_{j=1}^{n_1}\) in \(\cD_1\) as \(\Wb_1\).
\item We define \(\hat{\bar \psi}_2(a) := \frac{1}{n_2}\sum_{i=1}^{n_2} \psi(x_{2i},a)\) and the design operator of \(\{\hat{\bar \psi}_2({a}_{2j}')\}_{j=1}^{n_2}\) in \(\cD_2\) as \(\Wb_2\).
\item We define the design operator of \(\{\phi(\tilde{a}_j)\}_{j=1}^n\) as \(\Ub\).
\end{itemize}

\paragraph{Hilbertian Formulation for \(\tilde h\).}
Also, we rewrite the regression that obtains \(\tilde{h}\) using Hilbertian elements.
To get \(\tilde{h}\), we run Algorithm~\ref{algorithm: main} with the proxy regularizer \(\tilde{\lambda}\), using the dataset \(\cD_2\).
Through the first nuisance estimation, we obtain \(\tilde{f}\), and denote its Hilbertian element as \(\tilde{\theta}\).
Explicitly, this is
\begin{equation*}
\tilde{\theta} := (\Zb_2^\top \Zb_2 + n_2 
\lambda_0\Ib)^{-1} \Zb_2^\top \Yb_2.
\end{equation*}
Now, using this, \(\tilde{\eta}\) is obtained as follows:
\begin{equation*}
\tilde{\eta} := (\Ab_2^\top \Ab_2 + n_2\tilde \lambda \Ib)^{-1}\Ab_2^\top \Wb_2 \tilde\theta.
\end{equation*}

\paragraph{Hilbertian Formulation for \(\hat h_\lambda\). }
Similar to the analysis in Appendix~\ref{section:apdx groundwork}, with a slight abuse of notation, we get 
\begin{equation*}
\hat{\eta}_\lambda 
= (\Ab_1^\top \Ab_1 + n_1\lambda \Ib)^{-1} \Ab_1^\top \Wb_1 \hat{\theta}
\end{equation*}
for 
\begin{equation*}
\hat\theta = (\Zb_1^\top \Zb_1 + n_1\lambda_0 \Ib)^{-1 } \Zb_1^\top \Yb_1.
\end{equation*}

\subsection{Good Events for Model Selection}

We previously defined a good event for the dataset \(\cD\) in Appendix~\ref{subsection:apdx good event general theory}.
By applying the same argument used there, we aim to define a good event for the model selection proof.
Recall that \(n_1 = n_2 = \frac{n}{2}\).
By the same logic as in the preceding Appendix~\ref{subsection:apdx good event general theory}, the following holds for each case with a probability of at least \(1-n^{-11}\) for all \(\{a_{1j}'\}_{j=1}^{n_1}, \{a_{2j}'\}_{j=1}^{n_2}\):
\begin{equation}\label{equation; good event model selection 1}
\begin{aligned}
&\big|\frac{1}{n_1}\sum_{i=1}^{n_1} f^\star(x_{1i},a_{1j}') - \EE_{x \sim \cP_x}[f^\star(x, a_{1j}')]\big| \lesssim \frac{\xi\|\theta^\star\|_{\mathbb{F}}\sqrt{\log n}}{\sqrt{n}} \\ 
&\big|\frac{1}{n_2}\sum_{i=1}^{n_2} f^\star(x_{2i},a_{2j}') - \EE_{x \sim \cP_x}[f^\star(x, a_{2j}')]\big| \lesssim \frac{\xi\|\theta^\star\|_{\mathbb{F}}\sqrt{\log n}}{\sqrt{n}} 
\end{aligned}    
\end{equation}

Again, through the same logic, by applying Lemma~\ref{lemma: trace class concentration bounded} for \(\delta = n^{-11}\), the following four inequalities also hold for some absolute constant \(c>1\), each with a probability of at least \(1-n^{-11}\):
\begin{equation}\label{equation; good event model selection 2}
\begin{aligned}
& \frac{1}{c}(\Ab_1^\top \Ab_1 + \log n \Ib) \preceq n_1 \bSigma_{\operatorname{samp}} + \log n \Ib  \preceq c(\Ab_1^\top \Ab_1 + \log n \Ib)  \\
& \frac{1}{c}(\Ab_2^\top \Ab_2 + \log n \Ib ) \preceq n_2 \bSigma_{\operatorname{samp}} + \log n \Ib  \preceq c(\Ab_2^\top \Ab_2 + \log n \Ib ) \\
& \frac{1}{c}(\Zb_1^\top \Zb_1 + \log n \Ib ) \preceq n_1 \EE_{(x_i, a_i) \sim \cP_{x,a}}[\psi(x_i, a_i)\psi(x_i, a_i)^\top] + \log n \Ib  \preceq c(\Zb_1^\top \Zb_1 + \log n \Ib ) \\
& \frac{1}{c}(\Zb_2^\top \Zb_2 + \log n \Ib ) \preceq n_2 \EE_{(x_i, a_i) \sim \cP_{x,a}}[\psi(x_i, a_i)\psi(x_i, a_i)^\top] + \log n \Ib  \preceq c(\Zb_2^\top \Zb_2 + \log n \Ib ).
\end{aligned}    
\end{equation}

In Appendix~\ref{section:apdx proof general theory}, we defined \(\Qb(a)\) as
\begin{equation*}
\Qb(a) = \EE_{x \sim \cP_x}[\psi(x, a)\psi(x, a)^\top].
\end{equation*}
Also, we similarly define \(\bar \Qb_1\) and \(\bar \Qb_2\) as follows: 
\begin{equation*}
\bar \Qb_1 := \frac{1}{n_1} \sum_{j=1}^{n_1} \Qb(a_{1j}'), \quad
\bar \Qb_2 := \frac{1}{n_2} \sum_{j=1}^{n_2} \Qb(a_{2j}')
\end{equation*}
Then, by the same argument as in the proof of Lemma~\ref{lemma: sampled Qbar bound}, for some constant \(c>1\), the following holds, each with a probability of at least \(1-n^{-11}\):
\begin{equation*}
\bar \Qb_1 \preceq c( \Qb + \frac{\log n }{n_1}\Ib), \quad
\bar \Qb_2 \preceq c(\Qb + \frac{\log n}{n_2} \Ib).
\end{equation*}
For the event where the two inequalities above hold, by applying the same argument as in Lemma~\ref{lemma: key inequality 1} under that event, we get
\begin{equation}\label{equation; good event model selection 3}
\begin{aligned}
&\frac{1}{n_1} \Wb_1^\top{\Wb_1} \preceq cR\times \frac{1}{n_1}(\Zb_1^\top \Zb_1 + \log n \Ib) \\
&\frac{1}{n_2} \Wb_2^\top{\Wb_2} \preceq cR\times \frac{1}{n_2}(\Zb_2^\top \Zb_2 + \log n \Ib)
\end{aligned}    
\end{equation}

Lastly, we present the concentration for \(\Ub^\top \Ub\).
Since \(\bSigma_{\operatorname{ref}} = \EE_{a \sim \cP_{\operatorname{ref}}}[\phi(a)\phi(a)^\top]\), by applying Lemma~\ref{lemma: trace class concentration bounded}, 
with probability at least \(1-n^{-11}\), we have 
\begin{equation}\label{equation; good event model selection 4}
\frac{1}{c}( \Ub^\top \Ub  + \log n \Ib) \preceq n \bSigma_{\operatorname{ref}} + \log n \Ib \preceq  c(\Ub^\top \Ub  + \log n \Ib).
\end{equation}

\begin{definition}[Good Event]
We define the good event \(\Fcr_{\operatorname{good}}\) as the event where \cref{equation; good event model selection 1,equation; good event model selection 2,equation; good event model selection 3,equation; good event model selection 4} hold.
\end{definition}
Then, using the above observation, we get the following lemma directly.
\begin{lemma}\label{lemma: model selection good event}
We have
\begin{equation*}
\PP[\Fcr_{\operatorname{good}}] \ge  1-\frac{1}{n^{10}}.
\end{equation*}
\end{lemma}

\noindent\textbf{Proof of Lemma~\ref{lemma: model selection good event}. }
The proof is straightforward from the previous observations by taking a union bound.
\hfill\BlackBox

\subsection{Proof of Proposition~\ref{proposition; in-sample oracle}}

We analyze the quality of the pseudo-test-outcomes defined as 
\begin{equation*}
\tilde \mb :=  \Ub \tilde{\eta}= \Ub (\Ab_2^\top \Ab_2 + n_2\tilde \lambda \Ib)^{-1}\Ab_2^\top \Wb_2 \tilde\theta.
\end{equation*}
Recall that we set
\begin{equation*}
\tilde{\theta} :=(\Zb_2^\top \Zb_2 + n_2\lambda_0 \Ib)^{-1} \Zb_2^\top \Yb_2.
\end{equation*}
We also define 
\begin{equation*}
\mb^\star := \Ub \eta^\star.
\end{equation*}
By the same argument as in Appendix~\ref{section:apdx decomposition}, we can see that the estimation error of \(\tilde\eta\) can be decomposed as
\begin{equation*}
\begin{aligned}
\tilde\eta - \eta^\star  &=
(\Ab_2^\top \Ab_2 + n_2\tilde \lambda \Ib)^{-1}\Ab_2^\top \bb_2 +  (\Ab_2^\top \Ab_2 + n_2\tilde \lambda \Ib)^{-1} \Ab_2^\top \Wb_2 (\tilde \theta -\theta^\star) - n_2\tilde\lambda \cdot  (\Ab_2^\top \Ab_2 + n_2\tilde \lambda\Ib)^{-1} \eta^\star \\
&:= \tilde \Bcr + \tilde \Pcr + \tilde \Ccr
\end{aligned}
\end{equation*}
and this leads to
\begin{equation}\label{equation; decomposition model selection}
\begin{aligned}
\tilde\mb - \mb^\star = \Ub(\tilde\eta - \eta^\star ) = 
\Ub \tilde\Bcr + \Ub \tilde\Pcr + \Ub \tilde\Ccr.
\end{aligned}
\end{equation}

\begin{lemma}\label{lemma: model selection pseudo test subgaussian}
Under the event \(\Fcr_{\operatorname{good}}\), 
\begin{equation*}
\| \tilde \mb - \EE_{\bm \varepsilon_2}[{\tilde \mb}] \|_{\psi_2}^2 \lesssim \sigma^2 B(\tilde\lambda) \cdot R.
\end{equation*}
In \(\EE_{\bm{\varepsilon}_2}[\tilde\mb]\), the expectation is taken with respect to the noise variables.
\end{lemma}

\noindent\textbf{Proof of Lemma~\ref{lemma: model selection pseudo test subgaussian}. }
We see that 
\begin{equation*}
\begin{aligned}
\|\tilde \mb - \EE_{\bm \varepsilon_2}[{ \tilde \mb}] \|_{\psi_2}  &\lesssim \| \Ub (\Ab_2^\top \Ab_2 + n_2\tilde \lambda \Ib)^{-1} \Ab_2^\top \Wb_2 (\tilde \theta -\EE_{\bm \varepsilon_2}[\tilde \theta]) \|_{\psi_2}\\
&\lesssim  \|\Ub (\Ab_2^\top \Ab_2 + n_2\tilde \lambda \Ib)^{-1} \Ab_2^\top \Wb_2 (\Zb_2^\top \Zb_2 + n_2\lambda_0 \Ib)^{-1} \Zb_2^\top \bm\varepsilon_2\|_{\psi_2} \\
&\lesssim \sigma  \|\Ub (\Ab_2^\top \Ab_2 + n_2 \tilde \lambda \Ib)^{-1} \Ab_2^\top \Wb_2 (\Zb_2^\top \Zb_2 + n_2\lambda_0 \Ib)^{-1} \Zb_2^\top \|_{\operatorname{op}} \\
&\lesssim \sigma  \|\Ub (\Ab_2^\top \Ab_2 + n_2\tilde \lambda \Ib)^{-1} \Ab_2^\top \Wb_2 (\Zb_2^\top \Zb_2 + n_2\lambda_0 \Ib)^{-\frac{1}{2}}\|_{\operatorname{op}} \\
&\lesssim \sigma  \|\Ub (\Ab_2^\top \Ab_2 + n_2\tilde \lambda \Ib)^{-1} \Ab_2^\top \|_{\operatorname{op}} \cdot  \|\Wb_2 (\Zb_2^\top \Zb_2 + n_2\lambda_0 \Ib)^{-\frac{1}{2}}\|_{\operatorname{op}} \\
&\stackrel{(i)}{\lesssim} \sigma \sqrt{R} \|\Ub (\Ab_2^\top \Ab_2 + n_2\tilde \lambda \Ib)^{-1} \Ab_2^\top \|_{\operatorname{op}}\\
&\lesssim \sigma \sqrt{R} \|\Ub (\Ab_2^\top \Ab_2 + n_2\tilde \lambda \Ib)^{-\frac{1}{2}} \|_{\operatorname{op}} \\
&\stackrel{(ii)}{\lesssim} \sigma \sqrt{R  B(\tilde \lambda)}
\end{aligned}
\end{equation*}
where step \((i)\) holds by \cref{equation; good event model selection 3} and step \((ii)\) holds by combining Assumption~\ref{assumption: second overlap} and the relation
\begin{equation*}
\frac{1}{c}(n_2\bSigma_{\operatorname{samp}} + n_2 \tilde \lambda \Ib) \preceq \Ab_2^\top \Ab_2 + n_2 \tilde \lambda \Ib 
\end{equation*}
and 
\begin{equation*}
\frac{1}{n}\Ub^\top \Ub \preceq c (\bSigma_{\operatorname{ref}} + \tilde \lambda\Ib)
\end{equation*}
for some absolute constant \(c>1\) under the good event \(\Fcr_{\operatorname{good}}\). 
We also repeatedly apply Remark~\ref{remark; useful observation} at each step.
\hfill\BlackBox

\begin{lemma}\label{lemma: model selection pseudo test bias}
Under the good event \(\Fcr_{\operatorname{good}}\), we get
\begin{equation*}
\| \EE_{\bm \varepsilon_2}[{\tilde \mb}] - \mb^\star \|_{2}^2
\lesssim
B(\tilde\lambda)
\left\{
R\|\theta^\star\|_{\mathbb{F}}^2\log n
+n\tilde\lambda\|\eta^\star\|_\HH^2
\right\}.
\end{equation*}
\end{lemma}

\noindent\textbf{Proof of Lemma~\ref{lemma: model selection pseudo test bias}. }
Recall that \(\tilde \mb = \Ub \tilde \eta\) and we defined 
\begin{equation*}
\tilde \Pcr :=(\Ab_2^\top \Ab_2 + n_2\tilde \lambda \Ib)^{-1} \Ab_2^\top \Wb_2 (\tilde \theta -\theta^\star)
\end{equation*}
in the previous decomposition \cref{equation; decomposition model selection}.
We decompose it again as
\begin{equation*}
\begin{aligned}
\tilde \Pcr &= (\Ab_2^\top \Ab_2 + n_2\tilde \lambda \Ib)^{-1} \Ab_2^\top \Wb_2 (\Zb_2^\top \Zb_2+ n_2\lambda_0 \Ib)^{-1} \Zb_2^\top \bm \varepsilon_2  \\
\quad &+(\Ab_2^\top \Ab_2 + n_2\tilde \lambda \Ib)^{-1} \Ab_2^\top \Wb_2  (-n_2 \lambda_0) (\Zb_2^\top \Zb_2+ n_2\lambda_0 \Ib)^{-1} \theta^\star \\
&:= \tilde{\Pcr}_v + \tilde \Pcr_b,
\end{aligned}
\end{equation*}
where 
\begin{equation*}
\begin{aligned}
\tilde{\Pcr}_v &:= (\Ab_2^\top \Ab_2 + n_2\tilde \lambda \Ib)^{-1} \Ab_2^\top \Wb_2 (\Zb_2^\top \Zb_2+ n_2\lambda_0 \Ib)^{-1} \Zb_2^\top \bm \varepsilon_2\\
\tilde{\Pcr}_b &:= -n_2\lambda_0(\Ab_2^\top \Ab_2 + n_2\tilde \lambda \Ib)^{-1} \Ab_2^\top \Wb_2 (\Zb_2^\top \Zb_2+ n_2\lambda_0 \Ib)^{-1} \theta^\star.
\end{aligned}
\end{equation*}
Using the previous decomposition \cref{equation; decomposition model selection}, we get
\begin{equation*}
\| \EE_{\bm \varepsilon_2}[{\tilde \mb}] - \mb^\star \|_{2} \le \| \Ub \tilde \Ccr \|_{2} +  \| \Ub \tilde \Pcr_b \|_{2}+ \| \Ub \tilde \Bcr \|_{2}.
\end{equation*}
Firstly, we see that 
\begin{equation*}
\begin{aligned}
\| \Ub \tilde\Bcr \|_{2}^2 &\lesssim     \bb_2^\top \Ab_2^\top (\Ab_2^\top \Ab_2 + n_2 \tilde \lambda \Ib)^{-1} \Ub^\top \Ub (\Ab_2^\top \Ab_2 + n_2\tilde \lambda \Ib)^{-1} \Ab_2 \bb_2 \\
&\stackrel{(i)}{\lesssim }    \bb_2^\top \Ab_2^\top (\Ab_2^\top \Ab_2 + n_2\tilde \lambda \Ib)^{-1} (n\bSigma_{\operatorname{ref}} + \log n \Ib)(\Ab_2^\top \Ab_2 + n_2\tilde \lambda \Ib)^{-1} \Ab_2 \bb_2 \\
&\stackrel{(ii)}{\lesssim} B(\tilde \lambda)\cdot \bb_2^\top \Ab_2^\top (\Ab_2^\top \Ab_2 + n_2\tilde \lambda \Ib)^{-1}  \Ab_2 \bb_2 \quad (\text{By Assumption~\ref{assumption: second overlap} and \cref{equation; good event model selection 2}})\\
&\lesssim  B(\tilde \lambda)\|\bb_2\|_2^2 \\
&\stackrel{(iii)}{\lesssim} B(\tilde \lambda)  \|\theta^\star \|_\FF^2 \log n,
\end{aligned}
\end{equation*}
where step (i) holds by the definition of the good event \(\Fcr_{\operatorname{good}}\), especially \cref{equation; good event model selection 4} and step (ii) holds by \cref{equation; good event model selection 2} and Assumption~\ref{assumption: second overlap}.
Also, step (iii) holds by \cref{equation; good event model selection 1}.
We also repeatedly apply Remark~\ref{remark; useful observation} at each step.

Secondly, we aim to bound 
\begin{equation*}
\begin{aligned}
\|\Ub \tilde\Ccr\|_2^2 &\lesssim  (n_2 \tilde{\lambda})^2 (\eta^\star)^\top (\Ab_2^\top \Ab_2 + n_2 \tilde\lambda \Ib)^{-1}\Ub^\top \Ub (\Ab_2^\top \Ab_2 + n_2 \tilde\lambda \Ib)^{-1} \eta^\star \\
&\stackrel{(i)}{\lesssim} (n_2 \tilde{\lambda})^2 (\eta^\star)^\top (\Ab_2^\top \Ab_2 + n_2 \tilde\lambda \Ib)^{-1} (n\bSigma_{\operatorname{ref}} + \log n \Ib) (\Ab_2^\top \Ab_2 + n_2 \tilde\lambda \Ib)^{-1} \eta^\star \\
&\stackrel{(ii)}{\lesssim} (n_2 \tilde{\lambda})^2  B(\tilde \lambda) (\eta^\star)^\top (\Ab_2^\top \Ab_2 + n_2 \tilde\lambda \Ib)^{-1}  \eta^\star \\
&\lesssim (n_2 \tilde{\lambda})  B(\tilde \lambda)\| \eta^\star\|_\HH^2 \\
&\lesssim n\tilde\lambda B(\tilde \lambda)\| \eta^\star\|_\HH^2, 
\end{aligned}
\end{equation*}
where step (i) holds by \cref{equation; good event model selection 4}  and step (ii) holds by \cref{equation; good event model selection 2} and Assumption~\ref{assumption: second overlap}.

Lastly, we bound
\begin{equation*}
\begin{aligned}
&\|  \Ub \tilde \Pcr_b \|_2^2 \\
&\leq 
(n_2 \lambda_0)^2\|\Ub (\Ab_2^\top \Ab_2 + n_2\tilde \lambda \Ib)^{-1} \Ab_2^\top  \Wb_2(\Zb_2^\top \Zb_2+ n_2\lambda_0 \Ib)^{-1} \theta^\star \|_{\HH}^2 \\
&\lesssim 
(n_2\lambda_0)^2 \|\Ub (\Ab_2^\top \Ab_2 + n_2\tilde \lambda \Ib)^{-\frac{1}{2}} \|_{\operatorname{op}}^2\cdot \|(\Ab_2^\top \Ab_2 + n_2\tilde \lambda \Ib)^{-\frac{1}{2}} \Ab_2^\top  \|_{\operatorname{op}}^2 \| \Wb_2(\Zb_2^\top \Zb_2+ n_2\lambda_0\Ib)^{-1} \theta^\star \|_{\FF}^2 \\
&\lesssim 
(n_2\lambda_0)^2  \|\Ub (\Ab_2^\top \Ab_2 +  n_2\tilde \lambda \Ib)^{-\frac{1}{2}} \|_{\operatorname{op}}^2\cdot  \| \Wb_2(\Zb_2^\top \Zb_2+ n_2\lambda_0 \Ib)^{-1} \theta^\star \|_{\FF}^2\\
&\lesssim 
(n_2 \lambda_0)^2  \|\Ub (\Ab_2^\top \Ab_2  + n_2\tilde \lambda \Ib)^{-\frac{1}{2}} \|_{\operatorname{op}}^2\cdot  \| \Wb_2(\Zb_2^\top \Zb_2+ n_2\lambda_0 \Ib)^{-\frac{1}{2}} \|_{\operatorname{op}}^2\cdot  \frac{1}{n_2 \lambda_0}\| \theta^\star \|_{\FF}^2 \\
&\stackrel{(i)}{\lesssim }
n_2 \lambda_0 \|\Ub (\Ab_2^\top \Ab_2  + n_2 \tilde \lambda \Ib)^{-\frac{1}{2}} \|_{\operatorname{op}}^2\cdot  R\cdot  \| \theta^\star \|_{\FF}^2 \\
&\stackrel{(ii)}{\lesssim }
B(\tilde \lambda) \cdot R  n_2\lambda_0 \| \theta^\star \|_{\FF}^2 \\
&\lesssim B(\tilde \lambda) R \log n \| \theta^\star \|_{\FF}^2
\end{aligned}
\end{equation*}
where step (i) holds by \cref{equation; good event model selection 3} and step (ii) holds by \cref{equation; good event model selection 4} and Assumption~\ref{assumption: second overlap}.
\hfill\BlackBox

\paragraph{Final Step of the Proof.}
We now aim to apply Lemma~\ref{lemma: loss model selection}.
We can see that 
\begin{equation*}
\begin{aligned}
\cE_{\HH}^{\operatorname{in}}(\hat 
\eta_{\operatorname{ms}} )& \lesssim \min_{\lambda \in \Lambda}\cE^{\operatorname{in}}_{\HH}(\hat \eta_{\lambda} ) +  \frac{1}{n} \|\Ub (\EE[\tilde \eta]-\eta^\star) \|_2^2 +  \frac{\log(Ln)}{n} \|\Ub (\tilde\eta - \EE[\tilde \eta]) \|^2_{\psi_2} \\
&\lesssim \min_{\lambda \in \Lambda}\cE_\HH^{\operatorname{in}}(\hat \eta_{\lambda} ) + B(\tilde\lambda)
\left\{
\tilde\lambda\|\eta^\star\|_\HH^2
+\frac{\log n}{n}R(1+\|\theta^\star\|_\FF^2)
\right\}
\end{aligned}
\end{equation*}
holds with probability at least \(1-n^{-11}\) under the good event \(\Fcr_{\operatorname{good}}\).
Since \(\PP[\Fcr_{\operatorname{good}}] \ge 1-n^{-11}\), this completes the proof.

\begin{lemma}[Theorem 5.2 from \citealt{wang2026pseudo}]\label{lemma: loss model selection}
Let $\left\{\boldsymbol{z}_i\right\}_{i=1}^n$ be deterministic elements in a set $\mathcal{Z} ; g^{\star}$ and $\left\{g_j\right\}_{j=1}^m$ be deterministic functions on $\mathcal{Z} ; \widetilde{g}$ be a random function on $\mathcal{Z}$. Define
\begin{equation*}
\mathcal{L}(g)=\frac{1}{n} \sum_{i=1}^n\left|g\left(\boldsymbol{z}_i\right)-g^{\star}\left(\boldsymbol{z}_i\right)\right|^2
\end{equation*}
for any function $g$ on $\mathcal{Z}$. Assume that the random vector $\widetilde{\boldsymbol{y}}=\left(\widetilde{g}\left(\boldsymbol{z}_1\right), \widetilde{g}\left(\boldsymbol{z}_2\right), \cdots, \widetilde{g}\left(\boldsymbol{z}_n\right)\right)^{\top}$ satisfies $\|\widetilde{\boldsymbol{y}}-\mathbb{E} \widetilde{\boldsymbol{y}}\|_{\psi_2} \leq V<\infty$. Choose any
\begin{equation*}
\widehat{j} \in \underset{j \in[m]}{\operatorname{argmin}}\left\{\frac{1}{n} \sum_{i=1}^n\left|g_j\left(\boldsymbol{z}_i\right)-\widetilde{g}\left(\boldsymbol{z}_i\right)\right|^2\right\} .
\end{equation*}

There exists a universal constant $C$ such that for any $\delta \in(0,1]$, with probability at least $1-\delta$ we have
\begin{equation*}
\mathcal{L}\left(g_{\widehat{j}}\right) \leq \inf _{\gamma>0}\left\{(1+\gamma) \min _{j \in[m]} \mathcal{L}\left(g_j\right)+C\left(1+\gamma^{-1}\right)\left(\mathcal{L}(\mathbb{E} \widetilde{g})+\frac{V^2 \log (m / \delta)}{n}\right)\right\} .
\end{equation*}

Consequently,
\begin{equation*}
\mathbb{E} \mathcal{L}\left(g_{\widehat{j}}\right) \leq \inf _{\gamma>0}\left\{(1+\gamma) \min _{j \in[m]} \mathcal{L}\left(g_j\right)+C\left(1+\gamma^{-1}\right)\left(\mathcal{L}(\mathbb{E} \widetilde{g})+\frac{V^2(1+\log m)}{n}\right)\right\}.
\end{equation*}    
\end{lemma}

\subsection{Proof of Proposition~\ref{proposition; in sample to population MSE}}
We aim to analyze the behavior of \(\hat h_\lambda\), which is obtained by running Algorithm~\ref{algorithm: main} using the dataset \(\cD_1\).
We use the decomposition established in Appendix~\ref{section:apdx decomposition}:
\begin{equation*}
\begin{aligned}
\hat\eta_\lambda - \eta^\star  &= (\Ab_1^\top \Ab_1 + n_1\lambda \Ib)^{-1} \Ab_1^\top \bb_1 +(\Ab_1^\top \Ab_1 + n_1\lambda \Ib)^{-1} \Ab_1^\top \Wb_1 (\hat\theta -\theta^\star) -n_1\lambda(\Ab_1^\top \Ab_1 + n_1\lambda \Ib)^{-1} \eta^\star  \\
&:= \Bcr_1 + \Pcr_1 + \Ccr_1.
\end{aligned}
\end{equation*}

\begin{lemma}\label{lemma: psi norm model selection}
Under the event \(\Fcr_{\operatorname{good}}\),
\(\|\hat \eta_\lambda-\EE_{\bm{\varepsilon}_1}[\hat \eta_\lambda] \|_{\psi_2}^2  \lesssim R\sigma^2\).
\end{lemma}

\noindent\textbf{Proof of Lemma~\ref{lemma: psi norm model selection}. }
Using the established error decomposition, we get
\begin{equation*}
\begin{aligned}
\|\hat \eta_\lambda-\EE_{\bm{\varepsilon}_1}[\hat \eta_\lambda] \|_{\psi_2} &\lesssim \|(\Ab_1^\top \Ab_1 + n_1\lambda \Ib)^{-1} \Ab_1^\top \Wb_1 (\hat\theta -\EE_{\bm{\varepsilon}_1}[\hat\theta]) \|_{\psi_2} \\
&\lesssim \|(\Ab_1^\top \Ab_1 + n_1\lambda \Ib)^{-1} \Ab_1^\top \Wb_1  (\Zb_1^\top  \Zb_1 + n_1 \lambda_0 \Ib)^{-1} \Zb_1^\top \bm \varepsilon_1 \|_{\psi_2} \\
&\le \sigma \|(\Ab_1^\top \Ab_1 + n_1\lambda \Ib)^{-1} \Ab_1^\top \Wb_1  (\Zb_1^\top \Zb_1+  n_1 \lambda_0 \Ib)^{-1} \Zb_1^\top \|_{\operatorname{op}} \\
&\le \sigma \|(\Ab_1^\top \Ab_1 + n_1\lambda \Ib)^{-1} \Ab_1^\top \Wb_1  ( \Zb_1^\top  \Zb_1 + n_1 \lambda_0 \Ib)^{-\frac{1}{2}}  \|_{\operatorname{op}} \\
&\stackrel{(i)}{\lesssim} \sigma \sqrt{R} \|(\Ab_1^\top \Ab_1 + n_1\lambda \Ib)^{-1} \Ab_1^\top \|_{\operatorname{op}} \\
&\lesssim \frac{\sigma \sqrt{R}}{\sqrt{n_1\lambda}} \\
&\lesssim \sigma \sqrt{R}
\end{aligned}
\end{equation*}
where (i) holds by \cref{equation; good event model selection 3} and Remark~\ref{remark; useful observation}.
\hfill\BlackBox

\begin{lemma}\label{lemma: Hilbert norm model selection}
Under the event \(\Fcr_{\operatorname{good}}\), we have
\begin{equation*}
\|\EE_{\bm{\varepsilon}_1}[\hat \eta_\lambda] - \eta^\star \|_{\HH}^2 \lesssim \| \eta^\star\|_\HH^2 + R\|\theta^\star\|_\FF^2.
\end{equation*}
\end{lemma}

\noindent\textbf{Proof of Lemma~\ref{lemma: Hilbert norm model selection}. }
Using the same argument as in Appendix~\ref{section:apdx decomposition}, we get
\begin{equation*}
\begin{aligned}
&\|\EE_{\bm{\varepsilon}_1}[\hat \eta_\lambda] - \eta^\star \|_{\HH} \\ &= \| (\Ab_1^\top \Ab_1 + n_1\lambda \Ib)^{-1} \Ab_1^\top \bb_1 +(\Ab_1^\top \Ab_1 + n_1\lambda \Ib)^{-1} \Ab_1^\top \Wb_1 (\EE_{\bm{\varepsilon}_1}[\hat\theta] -\theta^\star) -n_1\lambda(\Ab_1^\top \Ab_1 + n_1\lambda \Ib)^{-1} \eta^\star \|_{\HH} \\
&\le \|(\Ab_1^\top \Ab_1 + n_1\lambda \Ib)^{-1} \Ab_1^\top \bb_1 \|_\HH + \|(\Ab_1^\top \Ab_1 + n_1\lambda \Ib)^{-1} \Ab_1^\top \Wb_1 (\EE_{\bm{\varepsilon}_1}[\hat\theta] -\theta^\star)  \|_\HH \\
&\quad  + \| n_1\lambda(\Ab_1^\top \Ab_1 + n_1\lambda \Ib)^{-1} \eta^\star\|_\HH
\end{aligned}
\end{equation*}
We bound each term.
First, we see that 
\begin{equation*}
\begin{aligned}
\|(\Ab_1^\top \Ab_1 + n_1\lambda \Ib)^{-1} \Ab_1^\top \bb_1 \|^2_\HH &\lesssim  \frac{1}{n_1\lambda} \|(\Ab_1^\top \Ab_1 + n_1\lambda \Ib)^{-1/2} \Ab_1^\top \bb_1 \|_{\HH}^2 \\
&\lesssim  \frac{1}{n_1\lambda} \|(\Ab_1^\top \Ab_1 + n_1\lambda \Ib)^{-1/2} \Ab_1^\top \|_{\operatorname{op}}^2 \| \bb_1 \|_2^2 \\
&\lesssim \frac{1}{n_1\lambda} \log n  \|\theta^\star \|_\FF^2 \quad (\text{by definition of \(\Fcr_{\operatorname{good}}\)})  \\
&\lesssim  \|\theta^\star \|_\FF^2.
\end{aligned}
\end{equation*}
Next, we see that
\begin{equation*}
\begin{aligned}
&\|(\Ab_1^\top \Ab_1 + n_1\lambda \Ib)^{-1} \Ab_1^\top \Wb_1 (\EE_{\bm{\varepsilon}_1}[\hat\theta] -\theta^\star)  \|_\HH \\
&\lesssim    \|(\Ab_1^\top \Ab_1 + n_1\lambda \Ib)^{-1} \Ab_1^\top \Wb_1 (\Zb_1^\top  \Zb_1 + n_1\lambda_0 \Ib)^{-1} n_1\lambda_0 \theta^\star\|_\HH \\
&\lesssim  \frac{1}{\sqrt{n_1 \lambda}}  \|(\Ab_1^\top \Ab_1 + n_1\lambda \Ib)^{-\frac{1}{2}} \Ab_1^\top \Wb_1 (\Zb_1^\top  \Zb_1 + n_1\lambda_0 \Ib)^{-1} n_1\lambda_0 \theta^\star\|_\HH \\
&\lesssim    \frac{1}{\sqrt{n_1 \lambda}}\|\Wb_1 (\Zb_1^\top  \Zb_1 +  n_1\lambda_0 \Ib)^{-1} n_1\lambda_0 \theta^\star\|_\FF \\ 
&\lesssim    \frac{1}{\sqrt{n_1 \lambda}}\|\Wb_1 (\Zb_1^\top  \Zb_1 +  n_1\lambda_0 \Ib)^{-\frac{1}{2}} \|_{\operatorname{op}} \|(\Zb_1^\top  \Zb_1 +  n_1\lambda_0 \Ib)^{-\frac{1}{2}}  n_1\lambda_0 \theta^\star\|_\FF \\ 
&\lesssim \frac{\sqrt{n_1\lambda_0}}{\sqrt{n_1\lambda}} \|\theta^\star \|_\FF \|\Wb_1 (\Zb_1^\top  \Zb_1 +  n_1\lambda_0 \Ib)^{-\frac{1}{2}} \|_{\operatorname{op}} \\
&\lesssim \sqrt{R} \|\theta^\star \|_\FF \quad \text{(by \(\lambda \gtrsim \lambda_0\) and \cref{equation; good event model selection 3}}).
\end{aligned}
\end{equation*}
Lastly, 
\begin{equation*}
\| n_1\lambda (\Ab_1^\top \Ab_1 + n_1\lambda \Ib)^{-1} \eta^\star\|_\HH^2 \lesssim  \| \eta^\star\|_\HH^2.
\end{equation*}
Combining the obtained bounds, we finally prove the lemma.
\hfill\BlackBox

\paragraph{Final Step: Connecting Empirical and Population MSE.}
Fix \(\lambda\in\Lambda\) and define
\[
f_\lambda(a):=\langle\phi(a),\hat\eta_\lambda-\eta^\star\rangle_\HH,
\qquad a\in\cA.
\]
Then
\[
\|f_\lambda\|_n^2=\cE_{\HH}^{\operatorname{in}}(\hat\eta_\lambda),
\qquad
\|f_\lambda\|_{L^2(\cP_{\operatorname{ref}})}^2=\cE_{\HH}(\hat\eta_\lambda),
\]
where \(\|g\|_n^2=n^{-1}\sum_{i=1}^n g^2(\tilde a_i)\).
The estimates below use only the training split and the training query points.
Thus, after conditioning on the corresponding training-side event, \(f_\lambda\) remains independent of the pseudo-test points \(\{\tilde a_i\}_{i=1}^n\).

First, by kernel boundedness and Lemma~\ref{lemma: psi norm model selection}, for every fixed \(a\in\cA\),
\[
\|\langle\phi(a),\hat\eta_\lambda-\EE_{\bm\varepsilon_1}[\hat\eta_\lambda]\rangle_\HH\|_{\psi_2}
\lesssim
\sqrt{\xi}\,\|\hat\eta_\lambda-\EE_{\bm\varepsilon_1}[\hat\eta_\lambda]\|_{\psi_2}
\lesssim
\sqrt{\xi R}\,\sigma.
\]
Also, by kernel boundedness and Lemma~\ref{lemma: Hilbert norm model selection},
\[
|\langle\phi(a),\EE_{\bm\varepsilon_1}[\hat\eta_\lambda]-\eta^\star\rangle_\HH|
\lesssim
\sqrt{\xi}\big(\|\eta^\star\|_\HH+\sqrt{R}\|\theta^\star\|_\FF\big).
\]
Therefore, for sufficiently large constants \(c_1,c_2>1\), define
\[
r_n
:=
c_1\sqrt{\xi}\big(
\sqrt{R}\sigma\sqrt{\log n}
+\|\eta^\star\|_\HH
+\sqrt{R}\|\theta^\star\|_\FF
\big),
\]
and
\[
U_n
:=
c_2\sqrt{\xi}\big(
\sqrt{R}\sigma
+\|\eta^\star\|_\HH
+\sqrt{R}\|\theta^\star\|_\FF
\big).
\]
Then, for \(\tilde a_1\sim\cP_{\operatorname{ref}}\),
\[
\PP\big(|f_\lambda(\tilde a_1)|>r_n\big)\le n^{-16},
\qquad
\EE |f_\lambda(\tilde a_1)|^4\le U_n^4.
\]
The fourth-moment bound follows from the preceding uniform sub-Gaussian control, the boundedness of the deterministic mean term, and the standard inequality
\(\EE|X|^4\lesssim\|X\|_{\psi_2}^4\) for sub-Gaussian \(X\).

We use the following elementary clipped comparison to pass between the empirical pseudo-test norm and the population \(L^2(\cP_{\operatorname{ref}})\) norm.
The statement is included for completeness because \(f\) is random but independent of the pseudo-test points after conditioning on the training split.
\begin{lemma}[Empirical-to-Population Comparison]\label{lemma: empirical population scalar comparison}
Let \(z_1,\dots,z_n\) be i.i.d. samples from a distribution \(P\), and let \(f\) be a random real-valued function independent of \(z_1,\dots,z_n\).
Define
\[
\|f\|_n^2:=\frac{1}{n}\sum_{i=1}^n f^2(z_i),
\qquad
\|f\|_{L^2(P)}^2:=\EE[f^2(z_1)].
\]
Suppose that, for some deterministic \(r,U\ge 0\) and \(\varepsilon\in(0,1)\),
\[
\PP(|f(z_1)|>r)\le \varepsilon,
\qquad
\EE|f(z_1)|^4\le U^4.
\]
Then, for any \(\delta\in(0,1)\), with probability at least \(1-\delta-n\varepsilon\),
\[
\|f\|_n^2
\lesssim
\|f\|_{L^2(P)}^2
+r^2\frac{\log(1/\delta)}{n}
+U^2\varepsilon^{1/2},
\]
and
\[
\|f\|_{L^2(P)}^2
\lesssim
\|f\|_n^2
+r^2\frac{\log(1/\delta)}{n}
+U^2\varepsilon^{1/2}.
\]
\end{lemma}

\noindent\textbf{Proof of Lemma~\ref{lemma: empirical population scalar comparison}. }
Let \(\bar f(z):=f(z)\mathbf{1}\{|f(z)|\le r\}\).
With probability at least \(1-n\varepsilon\), all sampled points satisfy
\(|f(z_i)|\le r\), and hence \(\|f\|_n^2=\|\bar f\|_n^2\).
For the clipped variable \(Y=\bar f^2(z_1)\), we have \(0\le Y\le r^2\) and
\(\EE Y^2\le r^2\EE Y\).
Bernstein's inequality for bounded variables gives, with probability at least \(1-\delta\),
\[
\bigg|\frac{1}{n}\sum_{i=1}^n \bar f^2(z_i)-\EE \bar f^2(z_1)\bigg|
\le
\frac{1}{2}\EE \bar f^2(z_1)
+C r^2\frac{\log(1/\delta)}{n}
\]
for a universal constant \(C>0\).
Moreover,
\[
\EE[f^2(z_1)\mathbf{1}\{|f(z_1)|>r\}]
\le
\{\EE f^4(z_1)\}^{1/2}\PP(|f(z_1)|>r)^{1/2}
\le U^2\varepsilon^{1/2}.
\]
Combining these three displays gives both inequalities.
\hfill\BlackBox

Applying Lemma~\ref{lemma: empirical population scalar comparison} with
\[
z_i=\tilde a_i,\qquad
f=f_\lambda,\qquad
\varepsilon=n^{-16},\qquad
\delta=n^{-13},
\]
we obtain, for each fixed \(\lambda\in\Lambda\), with probability at least \(1-n^{-15}-n^{-13}\),
\begin{equation*}
\begin{aligned}
\cE_{\HH}^{\operatorname{in}}(\hat\eta_\lambda)
&\lesssim
\cE_{\HH}(\hat\eta_\lambda)
+\frac{r_n^2\log n}{n}
+U_n^2n^{-8},\\
\cE_{\HH}(\hat\eta_\lambda)
&\lesssim
\cE_{\HH}^{\operatorname{in}}(\hat\eta_\lambda)
+\frac{r_n^2\log n}{n}
+U_n^2n^{-8}.
\end{aligned}
\end{equation*}
The final two terms are bounded by \(\Ocr\), up to logarithmic factors, since \(B(\tilde\lambda)\ge 1\), \(R\ge 1\), and the fixed noise and target-radius parameters are absorbed into the constants.
The bounds are uniform in \(\lambda\) over the grid because \(r_n\) and \(U_n\) do not depend on \(\lambda\) as long as \(\lambda\gtrsim\lambda_0\).
Since \(|\Lambda|\lesssim n\), a union bound over \(\lambda\in\Lambda\), together with the training-side good event, proves Proposition~\ref{proposition; in sample to population MSE}.

\subsection{The Final Step}\label{section:apdx proof oracle inequality}
This can be proved directly by combining the results of Propositions~\ref{proposition; in-sample oracle} and~\ref{proposition; in sample to population MSE}.
Use the fact that \(|\Lambda| \lesssim n\).

\section{Proofs of Specializations of the General Theory}\label{section:apdx proof specializations}

\subsection{Continuous MISE Bound for a Fixed Regularizer}\label{section:apdx proof theorem MISE}
The following specialization of Theorem~\ref{theorem: general theory} is used to prove the adaptive guarantee in Corollary~\ref{corollary: model selection}. The spectral regimes are those in Definition~\ref{def:spectral}.

\begin{corollary}[Continuous MISE: Upper Bound]\label{corollary: MISE}
Suppose we run Algorithm~\ref{algorithm: main} with nuisance regularizer $\lambda_0 \asymp n^{-1}\log n$, main regularizer \(\lambda \gtrsim n^{-1}\log n\), and sampling distribution $\cP_{\operatorname{samp}} = \cP_{\operatorname{ref}}$.
Under Assumptions~\ref{ass:id}, \ref{ass:noise}, and~\ref{ass:bounded}, together with Relative Overlap of degree $\gamma$, with probability at least $1 - n^{-10}$, the estimator $\hat{h}_\lambda$ satisfies:
\begin{equation}\label{eq:main_bound}
\cE(\hat{h}_\lambda) \lesssim \underbrace{\lambda \|h^\star\|_\cH^2}_{\text{Bias}} + \underbrace{\frac{\Gamma(\lambda)}{\gamma n}}_{\text{Variance}} + \underbrace{\frac{\|f^\star\|_\cF^2}{\gamma n}}_{\text{Nuisance Error}}.
\end{equation}
Consequently, optimizing $\lambda$ yields the following convergence rates in terms of the effective sample size $n_{\operatorname{eff}} = \gamma n$:
\begin{itemize}
\item \textbf{Case (a) (Polynomial):} $\lambda \asymp n_{\operatorname{eff}}^{-\frac{2\ell}{1+2\ell}} \implies  \cE(\hat{h}_\lambda) \lesssim n_{\operatorname{eff}}^{-\frac{2\ell}{1+2\ell}} + \frac{\|f^\star\|_\cF^2}{n_{\operatorname{eff}}}$.
\item \textbf{Case (b) (Exponential):} choosing $\lambda \asymp n_{\operatorname{eff}}^{-1}$ up to logarithmic factors gives \(\cE(\hat{h}_\lambda) \lesssim \frac{1}{n_{\operatorname{eff}}} + \frac{\|f^\star\|_\cF^2}{n_{\operatorname{eff}}}\).
\item \textbf{Case (c) (Finite Rank):} choosing $\lambda \asymp D/n_{\operatorname{eff}}$ up to logarithmic factors gives \(\cE(\hat{h}_\lambda) \lesssim \frac{D}{n_{\operatorname{eff}}} + \frac{\|f^\star\|_\cF^2}{n_{\operatorname{eff}}}\).
\end{itemize}
The notation \(\lesssim\) hides constants and logarithmic factors.
\end{corollary}

\noindent\textbf{Proof of Corollary~\ref{corollary: MISE}.}
We apply Theorem~\ref{theorem: general theory}.
In the continuous-reference setting, \(\cP_{\operatorname{samp}}=\cP_{\operatorname{ref}}\), and hence
\[
\bSigma_{\operatorname{samp}}=\bSigma_{\operatorname{ref}},
\qquad
B(\cdot)\equiv 1.
\]
Relative Overlap of degree \(\gamma\) gives \(R=\gamma^{-1}\), so \(n_{\operatorname{eff}}=\gamma n\).
Therefore Theorem~\ref{theorem: general theory} yields
\[
\cE(\hat h_\lambda)
\lesssim
\lambda\|h^\star\|_\cH^2
+\frac{\Gamma(\lambda)}{\gamma n}
+\frac{\|f^\star\|_\cF^2}{\gamma n}.
\]
The three displayed rates follow from the standard effective-dimension bounds:
\[
\begin{aligned}
\Gamma(\lambda)&\lesssim \lambda^{-1/(2\ell)}
&&\text{under polynomial decay},\\
\Gamma(\lambda)&\lesssim \log(1/\lambda)
&&\text{under exponential decay},\\
\Gamma(\lambda)&\lesssim D
&&\text{under finite rank}.
\end{aligned}
\]
Optimizing \(\lambda\) in each case proves the corollary.

\subsection{Proof of Corollary~\ref{corollary: model selection}}\label{section:apdx proof theorem model selection continuous}
We apply Theorem~\ref{theorem: general theory model selection}.
In the continuous-reference regime, \(B(\cdot)\equiv 1\) and \(R=\gamma^{-1}\).
By construction of the geometric grid \(\Lambda\), there is a grid point \(\lambda^\star\in\Lambda\) within a constant factor of the oracle regularizer for each of the three spectral regimes in Corollary~\ref{corollary: MISE}.
Applying Theorem~\ref{theorem: general theory model selection} at \(\lambda^\star\), and then using Corollary~\ref{corollary: MISE}, gives the stated rates up to logarithmic factors.

\subsection{Proof of Corollary~\ref{corollary: point evaluation}}\label{section:apdx proof corollary PE}
Let \(B(a_0,r)=\{a\in\cA:\|a-a_0\|_2\le r\}\) and \(\cU_r=\cU(B(a_0,r))\).
We may take \(r\le \operatorname{diam}(\cA)\) without loss of generality, since larger radii give the same sampling region after intersection with \(\cA\).
We apply Theorem~\ref{theorem: general theory} with \(\cP_{\operatorname{ref}}=\delta(a_0)\) and \(\cP_{\operatorname{samp}}=\cU_r\).
Then \(\bSigma_{\operatorname{ref}}=\phi(a_0)\otimes\phi(a_0)\), and \(\bSigma_{\operatorname{samp}}=\EE_{a\sim\cU_r}[\phi(a)\otimes\phi(a)]\).
Since \(\bSigma_{\operatorname{ref}}\) has rank one, \(\Gamma(t)\lesssim 1\) for all \(t>0\).
By Lemma~\ref{lemma: Sobolev 1},
\begin{equation*}
B_r(\lambda)
\lesssim
1\vee
\left(\frac{r^{2\beta-p}}{\lambda}\right)^{\frac{p}{2\beta}}.
\end{equation*}
The local density lower bound gives \(R_r\lesssim(\kappa r^p)^{-1}\), since for all \(a\in B(a_0,r)\),
\begin{equation*}
\frac{d\cU_r}{d\cP_{A\mid X=x}}(a)
=\frac{1}{|B(a_0,r)|p_{A\mid X=x}(a)}
\le \frac{1}{\kappa |B(a_0,r)|}
\lesssim \frac{1}{\kappa r^p}.
\end{equation*}
Substitution into Theorem~\ref{theorem: general theory} gives
\begin{equation*}
\cE_{a_0}(\hat h_\lambda)
\lesssim
B_r(\lambda)
\left(
\lambda \|h^\star\|_\cH^2
+\frac{1}{\kappa n r^p}
+\frac{1}{\kappa n r^p}\|f^\star\|_\cF^2
\right).
\end{equation*}
Under the boundedness assumptions on \(\|h^\star\|_\cH\) and \(\|f^\star\|_\cF\), this becomes
\begin{equation*}
\cE_{a_0}(\hat h_\lambda)
\lesssim
\left[
1\vee
\left(\frac{r^{2\beta-p}}{\lambda}\right)^{\frac{p}{2\beta}}
\right]
\left(\lambda+\frac{1}{\kappa n r^p}\right).
\end{equation*}

It remains to choose \(\lambda\) and verify the stated range of radii.
Let \(\alpha=p/(2\beta)\in(0,1)\), \(T=r^{2\beta-p}\), \(V=(\kappa n r^p)^{-1}\), and \(L=\log n\).
The local volume bound \(|B(a_0,r)|\asymp r^p\), together with \(\kappa|B(a_0,r)|\le1\), implies \(V\gtrsim n^{-1}\). Thus \(\lambda\asymp LV\) satisfies the regularization requirement in Theorem~\ref{theorem: general theory}.
For this choice, the risk is bounded, up to constants, by
\begin{equation*}
\{1\vee(T/\lambda)^\alpha\}(\lambda+V)
\lesssim LV+L^{1-\alpha}T^\alpha V^{1-\alpha}.
\end{equation*}
The second term has no remaining power of \(r\), because
\begin{equation*}
T^\alpha V^{1-\alpha}
=r^{\frac{p(2\beta-p)}{2\beta}}
(\kappa n r^p)^{-\frac{2\beta-p}{2\beta}}
=
(\kappa n)^{-\frac{2\beta-p}{2\beta}}.
\end{equation*}
If \(r\gtrsim(\kappa n)^{-1/(2\beta)}\), then \(V\lesssim(\kappa n)^{-(2\beta-p)/(2\beta)}\), so the first term has the same rate up to logarithmic factors. This proves the corollary and explains why \(r\) need only exceed the stated threshold rather than be of exactly that order. The upper restriction \(r\le\operatorname{diam}(\cA)\) keeps the volume comparison valid, and overlap must hold throughout \(B(a_0,r)\). For smaller radii, the local sample size \(n\kappa r^p\) limits the bound.

\subsection{Proof of Corollary~\ref{corollary: point regret}}\label{section:apdx proof point regret}

\begin{lemma}[Crude Sobolev-Norm Control]\label{lemma: crude Sobolev norm control}
Under the assumptions of Corollary~\ref{corollary: point regret}, with probability at least \(1-n^{-10}\),
\[
\|\hat h_\lambda-h^\star\|_{H^\beta(\cA)}
\lesssim
M_n,
\qquad
M_n:=
\|h^\star\|_{\cH}
+(1+\sqrt{\kappa_{\operatorname{glob}}})\|f^\star\|_{\cF}
+\sigma\sqrt n.
\]
\end{lemma}

\noindent\textbf{Proof of Lemma~\ref{lemma: crude Sobolev norm control}. }
Since \(\cH=H^\beta(\cA)\), it is enough to control \(\|\hat\eta_\lambda-\eta^\star\|_{\HH}\).
By Lemma~\ref{lemma: error decomposition},
\[
\hat\eta_\lambda-\eta^\star=\Bcr+\Pcr_v+\Pcr_b+\Ccr,
\]
where
\[
\begin{aligned}
\Bcr&=(\Ab^\top\Ab+n\lambda\Ib)^{-1}\Ab^\top\bb,\\
\Pcr_v&=(\Ab^\top\Ab+n\lambda\Ib)^{-1}\Ab^\top
\Wb(\Zb^\top\Zb+n\lambda_0\Ib)^{-1}\Zb^\top\bm\varepsilon,\\
\Pcr_b&=-n\lambda_0(\Ab^\top\Ab+n\lambda\Ib)^{-1}\Ab^\top
\Wb(\Zb^\top\Zb+n\lambda_0\Ib)^{-1}\theta^\star,\\
\Ccr&=-n\lambda(\Ab^\top\Ab+n\lambda\Ib)^{-1}\eta^\star.
\end{aligned}
\]
We bound these four terms in the \(\HH\)-norm on the good event \(\Ecr_{\operatorname{good}}\).
First,
\[
\|\Bcr\|_{\HH}
\le
\|(\Ab^\top\Ab+n\lambda\Ib)^{-1}\Ab^\top\|_{\operatorname{op}}\|\bb\|_2
\lesssim
\sqrt{\frac{\log n}{n\lambda}}\|\theta^\star\|_{\FF},
\]
where we used \(\|\bb\|_2^2\lesssim\log n\|\theta^\star\|_{\FF}^2\).
Second,
\[
\|\Ccr\|_{\HH}
\le \|\eta^\star\|_{\HH}.
\]
For the propagated bias term, Lemma~\ref{lemma: key inequality 1} and Remark~\ref{remark; useful observation} give
\[
\|\Wb(\Zb^\top\Zb+n\lambda_0\Ib)^{-1/2}\|_{\operatorname{op}}\lesssim \sqrt R,
\]
because \(n\lambda_0\asymp\log n\). Hence
\[
\begin{aligned}
\|\Pcr_b\|_{\HH}
&\lesssim
n\lambda_0(n\lambda)^{-1/2}
\|\Wb(\Zb^\top\Zb+n\lambda_0\Ib)^{-1}\theta^\star\|_2\\
&\le
n\lambda_0(n\lambda)^{-1/2}
\|\Wb(\Zb^\top\Zb+n\lambda_0\Ib)^{-1/2}\|_{\operatorname{op}}
\|(\Zb^\top\Zb+n\lambda_0\Ib)^{-1/2}\theta^\star\|_{\FF}\\
&\lesssim
\sqrt{\frac{R\lambda_0}{\lambda}}\|\theta^\star\|_{\FF}.
\end{aligned}
\]
Finally, conditional on the design, Lemma~\ref{lemma: quadratic martingale} applied to the linear transform defining \(\Pcr_v\) yields, with probability at least \(1-n^{-11}\),
\[
\begin{aligned}
\|\Pcr_v\|_{\HH}^2
&\lesssim
\sigma^2\log n\,
\Tr\!\left[
(\Ab^\top\Ab+n\lambda\Ib)^{-1}\Ab^\top
\Wb(\Zb^\top\Zb+n\lambda_0\Ib)^{-1}\Zb^\top\Zb \right.\\
&\hspace{3.2cm}\left.
\cdot(\Zb^\top\Zb+n\lambda_0\Ib)^{-1}
\Wb^\top\Ab(\Ab^\top\Ab+n\lambda\Ib)^{-1}
\right]\\
&\lesssim
R\sigma^2\log n\,
\Tr\!\left[
(\Ab^\top\Ab+n\lambda\Ib)^{-1}\Ab^\top\Ab
(\Ab^\top\Ab+n\lambda\Ib)^{-1}
\right].
\end{aligned}
\]
The last step again uses
\(\Wb(\Zb^\top\Zb+n\lambda_0\Ib)^{-1}\Wb^\top\preceq cR\Ib\).
If \(s_1,\ldots,s_m\) are the nonzero eigenvalues of \(\Ab^\top\Ab\), then \(m\le n\) and
\[
\Tr\!\left[
(\Ab^\top\Ab+n\lambda\Ib)^{-1}\Ab^\top\Ab
(\Ab^\top\Ab+n\lambda\Ib)^{-1}
\right]
=\sum_{i=1}^m\frac{s_i}{(s_i+n\lambda)^2}
\le \frac{m}{n\lambda}
\le \frac{1}{\lambda}.
\]
Therefore \(\|\Pcr_v\|_{\HH}\lesssim \sigma\sqrt{R\log n/\lambda}\).
Combining the four bounds gives
\[
\|\hat\eta_\lambda-\eta^\star\|_{\HH}
\lesssim
\|\eta^\star\|_{\HH}
+\left(
\sqrt{\frac{\log n}{n\lambda}}
+\sqrt{\frac{R\lambda_0}{\lambda}}
\right)\|\theta^\star\|_{\FF}
+\sigma\sqrt{\frac{R\log n}{\lambda}}.
\]
In the global-overlap setting, \(R\lesssim\kappa_{\operatorname{glob}}^{-1}\), \(\lambda_0\asymp n^{-1}\log n\), and \(\lambda\asymp \log n/(n\kappa_{\operatorname{glob}})\).
Substituting these choices gives the claimed bound.
\hfill\BlackBox

Let
\[
r_n:=\left(n\kappa_{\operatorname{glob}}\right)^{-\frac{2\beta-p}{2\beta}}.
\]
Under the global overlap condition in Corollary~\ref{corollary: point regret}, the proof of Corollary~\ref{corollary: point evaluation} applies with \(\cP_{\operatorname{samp}}=\cU(\cA)\) and \(R\lesssim \kappa_{\operatorname{glob}}^{-1}\) for every fixed \(a\in\cA\).
Thus, with logarithmic factors adjusted, for every fixed \(a\in\cA\),
\[
|\hat h_\lambda(a)-h^\star(a)|^2\lesssim r_n
\]
with high probability.
To make this simultaneous over \(\cA\), let \(\mathcal{N}_\delta\) be a Euclidean \(\delta\)-net of \(\cA\).
Because \(p=\cO(1)\), \(\log |\mathcal{N}_\delta|\lesssim \log(1/\delta)\).
Using a union bound over \(\mathcal{N}_\delta\), we obtain
\[
\max_{a\in\mathcal{N}_\delta}|\hat h_\lambda(a)-h^\star(a)|\lesssim r_n^{1/2}
\]
with probability at least \(1-n^{-10}\), after absorbing the net-size logarithm into \(\widetilde{\cO}\).

It remains to pass from the net to the continuum.
By the Sobolev embedding theorem, with \(\alpha=\min\{\beta-p/2,1\}>0\),
\[
|g(a)-g(a')|\lesssim \|g\|_{H^\beta(\cA)}\|a-a'\|^\alpha
\qquad\text{for every }g\in H^\beta(\cA).
\]
By Lemma~\ref{lemma: crude Sobolev norm control}, we also have the crude bound
\[
\|\hat h_\lambda-h^\star\|_{H^\beta(\cA)}\lesssim M_n
\]
with probability at least \(1-n^{-10}\).
Choose \(\delta\) so that \(M_n\delta^\alpha\lesssim r_n^{1/2}\), for instance \(\delta\asymp M_n^{-1/\alpha}r_n^{1/(2\alpha)}\).
Since \(p=\cO(1)\), the additional covering logarithm is controlled by \(\log M_n+\log(1/r_n)\), and is absorbed into the hidden logarithmic factors.
Consequently,
\[
\|\hat h_\lambda-h^\star\|_{L^\infty(\cA)}
\lesssim r_n^{1/2}
\]
with probability at least \(1-n^{-9}\), proving the uniform bound.

For the regret statement, let \(a^\star\in\argmax_{a\in\cA}h^\star(a)\).
Then
\begin{equation*}
\begin{aligned}
\operatorname{Reg}(\hat a_\lambda)
&=h^\star(a^\star)-h^\star(\hat a_\lambda)\\
&\le h^\star(a^\star)-\hat h_\lambda(a^\star)
+\hat h_\lambda(\hat a_\lambda)-h^\star(\hat a_\lambda)\\
&\le 2\|\hat h_\lambda-h^\star\|_{L^\infty(\cA)}.
\end{aligned}
\end{equation*}
Together with the uniform bound above, this gives the claimed simple-regret rate.

\section{Proofs of the Minimax Lower Bounds}
\label{section:apdx proof lower bound}

We prove Theorems~\ref{theorem: lower bound MISE} and~\ref{theorem: lower bound point} using a common confounded hard instance.
This construction is slightly stronger than a reduction to an unconfounded regression subclass, and it mirrors the main difficulty of the paper: weak overlap under covariate-dependent treatment assignment.
Let \(\cA=[0,1]^p\) and write \(a_\circ=(0,\dots,0)\).
Assume \(\cX=[0,1]^q\) with \(q\ge 1\), so that the construction can encode covariate-dependent treatment assignment.
Write \(x=(x_1,\dots,x_q)\), and define
\[
u(x):=2x_1-1,
\qquad
v(a):=2a_1-1.
\]
Then \(\EE[u(X)]=0\), \(\EE[u(X)^2]=1/3\), and
\(\int_{\cA}v(a)\,da=0\).

\paragraph{A Common Confounded Construction.}
Let \(\omega\in(0,1]\) denote an overlap level.
Set
\[
\pi_\omega:=\min\{1,2\omega\},
\qquad
c_\omega:=1-\frac{\omega}{\pi_\omega}.
\]
Then \(\pi_\omega\in[\omega,2\omega]\) and \(c_\omega\in[0,1/2]\).
For \(h\in H^\beta(\cA)\) supported away from \(a_\circ\), define
\[
f_h^\star(x,a):=c_0h(a)+\tau u(x),
\]
where \(c_0,\tau>0\) are sufficiently small constants so that
\[
\|f_h^\star\|_{H^\beta(\cX\times\cA)}\le 1
\qquad
\text{whenever}
\qquad
\|h\|_{H^\beta(\cA)}\le 1.
\]
This follows from the continuity of the lift \(h(a)\mapsto h(a)\) from \(H^\beta(\cA)\) to \(H^\beta(\cX\times\cA)\), and from choosing \(\tau\) small enough.

For this \(h\), consider the data-generating process
\[
X\sim \cU(\cX),
\]
\[
A\mid X=x\sim (1-\pi_\omega)\delta_{a_\circ}
+\pi_\omega q_x(a)\,da,
\qquad
q_x(a):=1+c_\omega u(x)v(a),
\]
and
\[
Y=c_0h(A)+\tau u(X)+\varepsilon,
\qquad
\varepsilon\sim\cN(0,1),
\]
where \(\varepsilon\) is independent of \((X,A)\).
Since \(c_\omega\le1/2\), \(u(x)\in[-1,1]\), \(v(a)\in[-1,1]\), and \(\int_{\cA}v(a)\,da=0\), the function \(q_x\) is a valid density and satisfies
\[
\frac12\le q_x(a)\le \frac32
\qquad
\text{for all }(x,a).
\]
Moreover,
\[
\pi_\omega q_x(a)
\ge
\pi_\omega(1-c_\omega)
=
\pi_\omega\frac{\omega}{\pi_\omega}
=
\omega
\qquad
\text{for all }(x,a).
\]
Equivalently, for every Borel set \(B\subseteq\cA\),
\[
\cP_{A\mid X=x}(B)
\ge
\omega |B|.
\]
Thus the absolutely continuous component of \(\cP_{A\mid X=x}\) has Lebesgue density at least \(\omega\).
Consequently, setting \(\omega=\gamma\) gives Relative Overlap of degree \(\gamma\) with respect to \(\cU(\cA)\), because
\[
\frac{d\cU(\cA)}{d\cP_{A\mid X=x}}(a)
=
\frac{1}{\pi_\gamma q_x(a)}
\le
\frac{1}{\gamma}
\qquad
\text{Lebesgue-a.e.}
\]
Similarly, setting \(\omega=\kappa\) gives the point evaluation overlap condition defining \(\cP_{\operatorname{glob}}(\kappa)\), namely
\[
\de \cP_{A\mid X=x}\ge \kappa\, da
\qquad
\text{for all }x\in\cX.
\]
If one insists on fully absolutely continuous treatment laws, the atom at \(a_\circ\) can be replaced by a smooth density supported in a sufficiently small neighborhood of \(a_\circ\); because all alternatives below are supported away from \(a_\circ\), the same lower bounds are unchanged.

The TEF under this construction is
\[
h_h^\star(a)
=
\EE_X[f_h^\star(X,a)]
=
c_0h(a),
\]
because \(\EE[u(X)]=0\).
The construction is genuinely confounded whenever \(c_\omega>0\).
Indeed, for \(a\) in the continuous part of the treatment distribution, Bayes' rule gives
\[
\EE[u(X)\mid A=a]
=
\frac{\EE[u(X)\{1+c_\omega u(X)v(a)\}]}
{\EE[1+c_\omega u(X)v(a)]}
=
\frac{c_\omega}{3}v(a),
\]
and therefore
\[
\EE[Y\mid A=a]
=
c_0h(a)+\frac{\tau c_\omega}{3}v(a),
\]
which differs from the TEF except in the degenerate no-confounding case.

Now define the transformed response
\[
Z_i:=Y_i-\tau u(X_i)=c_0h(A_i)+\varepsilon_i.
\]
The map \((X_i,A_i,Y_i)\mapsto(X_i,A_i,Z_i)\) is one-to-one because \(\tau\) and \(u(\cdot)\) are fixed and known in the hard instance.
For one observation, the joint density of \((X,A,Z)\) under \(h\) factors as
\[
p_{X,A}(x,a)\,\varphi(z-c_0h(a)),
\]
where \(p_{X,A}\) does not depend on \(h\) and \(\varphi\) is the standard normal density.
Hence the reduced sample
\[
T^n:=\{(A_i,Z_i)\}_{i=1}^n
\]
is sufficient for \(h\).
Equivalently, conditional on \(T^n\), the remaining variables \(X_1,\dots,X_n\) are generated from a Markov kernel that does not depend on \(h\).
Thus the full causal experiment is a randomization of the reduced experiment by an \(h\)-independent kernel.
By the Blackwell sufficiency theorem, for any loss considered below, every decision rule based on the full causal data has an equivalent randomized rule based only on \(T^n\) with the same risk.
Therefore, lower bounds for the reduced experiment also hold for the original causal experiment.

Marginally,
\[
A\sim (1-\pi_\omega)\delta_{a_\circ}+\pi_\omega\cU(\cA),
\]
because \(\EE[u(X)]=0\).
Let
\[
N_\omega:=\sum_{i=1}^n\Ib\{A_i\ne a_\circ\}
\sim \operatorname{Binomial}(n,\pi_\omega).
\]
Conditional on \(N_\omega=m\), the informative observations are exactly \(m\) i.i.d. samples from the usual random-design regression model
\[
U_j\sim\cU(\cA),
\qquad
Z_j=c_0h(U_j)+\varepsilon_j,
\qquad
j=1,\dots,m.
\]
The remaining \(n-m\) atom observations satisfy \(A_i=a_\circ\) and are pure noise, because all alternatives are supported away from \(a_\circ\) and hence \(h(a_\circ)=0\).

We will repeatedly use the following random-sample-size reduction.
Let
\[
M_\omega:=\lceil 2\pi_\omega n\rceil\vee 1,
\qquad
\mathscr E_\omega:=\{N_\omega\le M_\omega\}.
\]
By Markov's inequality,
\[
\Pr(\mathscr E_\omega)\ge \frac12.
\]
For each lower bound below, we calibrate the fixed-sample-size regression alternatives to sample size \(M_\omega\).
Conditional on any event \(N_\omega=m\le M_\omega\), the corresponding Kullback--Leibler divergences are no larger than in the fixed-\(M_\omega\) regression experiment.
Therefore the same Fano, Assouad, or Le Cam lower bound applies on \(\mathscr E_\omega\).
Multiplying by \(\Pr(\mathscr E_\omega)\ge1/2\) only changes the universal constant.
Since
\[
M_\omega\lesssim \pi_\omega n \lesssim \omega n,
\]
all fixed-\(M_\omega\) lower bounds below translate into lower bounds in terms of the overlap-adjusted sample size \(\omega n\).
As usual, if \(\omega n\) is bounded, the displayed rates are interpreted up to truncation by a positive constant; this does not affect the asymptotic rate statements.

\paragraph{MISE Lower Bound.}
Set \(\omega=\gamma\).
By the construction above, the resulting distributions satisfy Relative Overlap of degree \(\gamma\) with respect to \(\cU(\cA)\).
Consider the standard Assouad or Fano packing for \(p\)-dimensional Sobolev regression over \(H^\beta([0,1]^p)\), with all packing functions supported in a fixed subcube away from \(a_\circ\).
For the fixed-\(M_\gamma\) random-design regression experiment
\[
U_j\sim\cU(\cA),
\qquad
Z_j=c_0h(U_j)+\varepsilon_j,
\qquad
j=1,\dots,M_\gamma,
\]
the classical Sobolev regression lower bound gives
\[
\inf_{\hat h}
\sup_{\|h\|_{H^\beta(\cA)}\le1}
\EE\!\left[
\|\hat h-c_0h\|_{L^2(\cU(\cA))}^2
\right]
\gtrsim
M_\gamma^{-\frac{2\beta}{2\beta+p}}.
\]
By the random-sample-size reduction above, the same lower bound holds for the original causal experiment up to a universal constant.
Using \(\pi_\gamma\le2\gamma\), we have
\[
M_\gamma\lesssim \gamma n,
\]
and hence
\[
M_\gamma^{-\frac{2\beta}{2\beta+p}}
\gtrsim
(\gamma n)^{-\frac{2\beta}{2\beta+p}}.
\]
Therefore,
\[
\inf_{\hat{h}}
\sup_{\cP \in \cP(\gamma)}
\EE[\cE(\hat{h})]
\gtrsim
(\gamma n)^{-\frac{2\beta}{2\beta+p}}.
\]
This proves Theorem~\ref{theorem: lower bound MISE}.

\paragraph{Point Evaluation Lower Bound.}
Set \(\omega=\kappa\).
By the construction above, the resulting distributions belong to \(\cP_{\operatorname{glob}}(\kappa)\), because
\[
\de \cP_{A\mid X=x}\ge \kappa\,da
\qquad
\text{for all }x\in\cX.
\]
We give the fixed-\(M_\kappa\) pointwise lower bound explicitly.
Let \(\psi\in C_c^\infty(\RR^p)\) be supported in the unit ball and satisfy \(\psi(0)=1\).
Since \(a_0\) is an interior point, for all sufficiently small \(w>0\), the ball \(B(a_0,w)\) is contained in \(\cA\) and is away from \(a_\circ\).
Define
\[
\Delta:=c_\psi w^{\beta-p/2},
\qquad
h_+(a):=\Delta\psi\!\left(\frac{a-a_0}{w}\right),
\qquad
h_-(a):=-\Delta\psi\!\left(\frac{a-a_0}{w}\right),
\]
where \(c_\psi>0\) is a sufficiently small constant.
The Sobolev scaling
\[
\left\|\psi\!\left(\frac{\cdot-a_0}{w}\right)\right\|_{H^\beta(\cA)}
\lesssim w^{p/2-\beta}
\]
implies \(\|h_+\|_{H^\beta(\cA)}\vee\|h_-\|_{H^\beta(\cA)}\le1\).
In the fixed-\(M_\kappa\) regression experiment, the Kullback--Leibler (KL) divergence satisfies
\[
\mathrm{KL}(P_{h_+}^{(M_\kappa)},P_{h_-}^{(M_\kappa)})
=
\frac{M_\kappa c_0^2}{2}\|h_+-h_-\|_{L^2(\cA)}^2
\lesssim
M_\kappa\Delta^2w^p
\asymp
M_\kappa w^{2\beta}.
\]
Taking \(w\asymp M_\kappa^{-1/(2\beta)}\), with the implicit constant small enough, makes the KL divergence bounded by a sufficiently small universal constant.
The two point values are separated by
\[
|c_0h_+(a_0)-c_0h_-(a_0)|
=
2c_0\Delta
\asymp
M_\kappa^{-\frac{2\beta-p}{4\beta}}.
\]
Le Cam's two-point lemma therefore gives
\[
\inf_{\hat h}
\sup_{\theta\in\{+,-\}}
\EE_\theta\!\left[
\bigl(\hat h(a_0)-c_0h_\theta(a_0)\bigr)^2
\right]
\gtrsim
M_\kappa^{-\frac{2\beta-p}{2\beta}}.
\]
Again, the random-sample-size reduction transfers this lower bound to the original causal experiment.
Since \(\pi_\kappa\le2\kappa\),
\[
M_\kappa\lesssim \kappa n.
\]
Therefore,
\[
M_\kappa^{-\frac{2\beta-p}{2\beta}}
\gtrsim
(\kappa n)^{-\frac{2\beta-p}{2\beta}}.
\]
It follows that
\[
\inf_{\hat h}
\sup_{\cP\in\cP_{\operatorname{glob}}(\kappa)}
\EE[\cE_{a_0}(\hat h)]
\gtrsim
(\kappa n)^{-\frac{2\beta-p}{2\beta}}.
\]
This proves part (i) of Theorem~\ref{theorem: lower bound point}.

\paragraph{Regret Lower Bound.}
Again set \(\omega=\kappa\).
We use a two-bump testing argument for the reduced regression experiment, calibrated to sample size \(M_\kappa\).
Let \(\psi\) be a nonnegative smooth bump supported in the unit ball, with \(\psi(0)=1\).
Choose two points \(a_1,a_2\in\cA\), both away from \(a_\circ\), such that the balls \(B(a_1,w)\) and \(B(a_2,w)\) are disjoint and contained in the interior of \(\cA\).
For \(j=1,2\), define
\[
h_j(a):=
\Delta\,\psi\!\left(\frac{a-a_j}{w}\right),
\qquad
\Delta\asymp w^{\beta-p/2}.
\]
With the implicit constant sufficiently small, \(h_j\) lies in the unit ball of \(H^\beta(\cA)\).
Moreover,
\[
\|h_1-h_2\|_{L^2(\cA)}^2
\lesssim
\Delta^2w^p
\asymp
w^{2\beta}.
\]
For the fixed-\(M_\kappa\) reduced regression experiment, the Kullback--Leibler divergence between the two alternatives satisfies
\[
\mathrm{KL}(P_{h_1}^{(M_\kappa)},P_{h_2}^{(M_\kappa)})
\lesssim
M_\kappa w^{2\beta}.
\]
Taking
\[
w\asymp M_\kappa^{-1/(2\beta)}
\]
makes this divergence bounded by a sufficiently small universal constant.
With this choice,
\[
\Delta\asymp
w^{\beta-p/2}
\asymp
M_\kappa^{-\frac{2\beta-p}{4\beta}}.
\]

We now translate testing difficulty into regret.
Let
\[
G_j:=\{a\in\cA:h_j(a)\ge \Delta/2\}.
\]
Then \(G_1\) and \(G_2\) are disjoint, and \(G_j\) contains a smaller ball around \(a_j\) on which
\[
h_j(a)\ge \frac{\Delta}{2}.
\]
Under the TEF \(c_0h_j\), the optimal value is \(c_0\Delta\).
If a treatment-selection rule \(\hat a\) does not select inside \(G_j\), then
\[
\sup_{a\in\cA}c_0h_j(a)-c_0h_j(\hat a)
\ge
\frac{c_0\Delta}{2}.
\]
Given any selection rule \(\hat a\), define the test
\[
\varphi=\Ib\{\hat a\in G_1\}.
\]
Under \(h_1\), the event \(\{\varphi=0\}\) implies \(\hat a\notin G_1\), and hence regret at least \(c_0\Delta/2\).
Under \(h_2\), the event \(\{\varphi=1\}\) implies \(\hat a\in G_1\), where \(h_2(\hat a)=0\), and hence regret at least \(c_0\Delta/2\).
Therefore
\[
\sup_{j\in\{1,2\}}
\EE_{h_j}\!\left[
\sup_{a\in\cA}c_0h_j(a)-c_0h_j(\hat a)
\right]
\ge
\frac{c_0\Delta}{4}
\left\{
P_{h_1}^{(M_\kappa)}(\varphi=0)
+P_{h_2}^{(M_\kappa)}(\varphi=1)
\right\}.
\]
By Le Cam's two-point method and the KL bound above, the testing error in braces is bounded below by a positive universal constant.
Hence
\[
\inf_{\hat a}
\sup_{j\in\{1,2\}}
\EE_{h_j}\!\left[
\sup_{a\in\cA}c_0h_j(a)-c_0h_j(\hat a)
\right]
\gtrsim
M_\kappa^{-\frac{2\beta-p}{4\beta}}.
\]
Applying the random-sample-size reduction and using \(M_\kappa\lesssim \kappa n\), we obtain
\[
\inf_{\hat a}
\sup_{\cP\in\cP_{\operatorname{glob}}(\kappa)}
\EE[\operatorname{Reg}(\hat a)]
\gtrsim
(\kappa n)^{-\frac{2\beta-p}{4\beta}}.
\]
This proves part (ii) of Theorem~\ref{theorem: lower bound point}.

\section{Technical Lemmas}

\subsection{Concentration Inequalities}
First, we present the key lemmas of the trace class from \citet*{wang2026pseudo}.
The theorem and lemma numbers cited from that work in this appendix follow its \href{https://arxiv.org/abs/2302.10160v4}{arXiv version v4}.
\begin{lemma}[Lemma D.1 from \citealt{wang2026pseudo}]\label{lemma: quadratic martingale}
Suppose that $\boldsymbol{x} \in \mathbb{R}^d$ is a zero-mean random vector with $\|\boldsymbol{x}\|_{\psi_2} \leq 1$. There exists a universal constant $C>0$ such that for any symmetric and positive semi-definite matrix $\boldsymbol{\Sigma} \in \mathbb{R}^{d \times d}$,
\begin{equation*}
\mathbb{P}\left(\boldsymbol{x}^{\top} \boldsymbol{\bSigma} \boldsymbol{x} \leq C \operatorname{Tr}(\boldsymbol{\bSigma}) t\right) \geq 1-e^{-r(\boldsymbol{\bSigma}) t}, \quad \forall t \geq 1 .
\end{equation*}

Here $r(\boldsymbol{\bSigma})=\operatorname{Tr}(\boldsymbol{\bSigma}) /\|\boldsymbol{\bSigma}\|_2$ is the effective rank of \(\boldsymbol{\bSigma}\).
\end{lemma}
\begin{lemma}[Corollary D.1 from \citealt{wang2026pseudo}]\label{lemma: trace class concentration bounded}
Let $\left\{\boldsymbol{x}_i\right\}_{i=1}^n$ be i.i.d. random elements in a separable Hilbert space $\mathbb{H}$ with $\boldsymbol{\bSigma}=$ $\mathbb{E}\left(\boldsymbol{x}_i \otimes \boldsymbol{x}_i\right)$ being trace class. Define $\hat{\boldsymbol{\bSigma}}=\frac{1}{n} \sum_{i=1}^n \boldsymbol{x}_i \otimes \boldsymbol{x}_i$. Choose any constant $\gamma \in(0,1)$ and define an event $\mathcal{A}=\{(1-\gamma)(\boldsymbol{\bSigma}+\lambda \boldsymbol{I}) \preceq \hat{\boldsymbol{\bSigma}}+\lambda \boldsymbol{I} \preceq(1+\gamma)(\boldsymbol{\bSigma}+\lambda \boldsymbol{I})\}$.
1. If $\left\|\boldsymbol{x}_i\right\|_{\mathbb{H}}^2 \leq \xi$ holds almost surely for some constant $\xi$, then there exists a constant $C \geq 1$ determined by $\gamma$ such that $\mathbb{P}(\mathcal{A}) \geq 1-\delta$ holds so long as $\delta \in(0,1 / 14]$ and $\lambda \geq \frac{C \xi \log (n / \delta)}{n}$.    
Hence with probability at least \(1-\delta\), we have 
\begin{equation*}
\frac{1}{C\xi}(\widehat \bSigma + \frac{\log(n/\delta)}{n}\Ib ) \preceq \bSigma + \frac{\log(n/\delta)}{n} \Ib \preceq C\xi(\widehat \bSigma + \frac{\log(n/\delta)}{n}\Ib ).
\end{equation*}
\end{lemma}

\subsection{Properties of Sobolev Kernel and RKHS}

In this section, we consider the function space \(\cH =H^\beta(\cA)\) and also set \(\HH\) as the corresponding Hilbert space of \(\cH\).
Also, define \(\phi(\cdot)\) as the feature mapping of \(\HH\).

\begin{lemma}\label{lemma: Sobolev 1}
Assume that \(\beta>p/2\) and that \(\cA\) is a bounded Lipschitz domain. For \(0<r\le\operatorname{diam}(\cA)\), let
\[
B(a_0,r)=\{a\in\cA:\|a-a_0\|_2\le r\}.
\]
Define
\[
\bSigma_r:=\EE_{a\sim\cU(B(a_0,r))}[\phi(a)\phi(a)^\top].
\]
For any \(\lambda>0\), we have
\begin{equation*}
\phi(a_0)^\top(\bSigma_r+\lambda\Ib)^{-1}\phi(a_0)
\lesssim
1\vee\left(\frac{r^{2\beta-p}}{\lambda}\right)^{\frac{p}{2\beta}}.
\end{equation*}
Hence
\begin{equation*}
\phi(a_0)\phi(a_0)^\top
\lesssim
\left[
1\vee\left(\frac{r^{2\beta-p}}{\lambda}\right)^{\frac{p}{2\beta}}
\right]
(\bSigma_r+\lambda\Ib).
\end{equation*}
The hidden constants may depend on \(\cA\), \(\beta\), and \(p\), but not on \(\lambda\) or \(r\).
\end{lemma}

\noindent\textbf{Proof of Lemma~\ref{lemma: Sobolev 1}. }
Since \(\cA\) is a bounded Lipschitz domain, it admits a bounded extension operator from \(H^\beta(\cA)\) to \(H^\beta(\RR^p)\).
The uniform cone condition for Lipschitz domains also implies the measure-density property
\[
|B(a_0,t)|\asymp t^p
\]
uniformly over \(a_0\in\cA\) and \(0<t\le \operatorname{diam}(\cA)\), with constants depending only on \(\cA\).
Moreover, after rescaling, the local sets
\[
D_{a_0,t}:=t^{-1}\bigl(B(a_0,t)-a_0\bigr)
\]
belong to a Lipschitz class with uniformly bounded constants.
Consequently, the scaled local Sobolev interpolation inequality below holds with constants depending only on \(\cA,\beta\), and \(p\).
For \(0<\varepsilon\le r\), write \(B_t:=B(a_0,t)\). We claim that for every \(g\in H^\beta(\cA)\),
\begin{equation}\label{eq:scaled-local-sobolev}
|g(a_0)|^2
\lesssim
\left(\frac{r}{\varepsilon}\right)^p
\|g\|_{L^2(\cU(B_r))}^2
+\varepsilon^{2\beta-p}\|g\|_{H^\beta(\cA)}^2 .
\end{equation}
Indeed, let
\[
D_\varepsilon:=\varepsilon^{-1}(B_\varepsilon-a_0),
\qquad
\widetilde g(z):=g(a_0+\varepsilon z),\quad z\in D_\varepsilon .
\]
Because the rescaled domains \(D_\varepsilon\) have uniformly bounded Lipschitz character and diameter at most two, the Sobolev embedding constant can be chosen uniformly over \(a_0\) and \(\varepsilon\).
That is, since \(\beta>p/2\), there is a constant \(C_{\operatorname{emb}}=C_{\operatorname{emb}}(\cA,p,\beta)\) such that, for every \(u\in H^\beta(D_\varepsilon)\),
\[
\|u\|_{C^0(\overline{D_\varepsilon})}^2
\le
C_{\operatorname{emb}}
\left\{
\|u\|_{L^2(D_\varepsilon)}^2
+
[u]_{\dot H^\beta(D_\varepsilon)}^2
\right\}.
\]
Applying this inequality to \(u=\widetilde g\) gives
\begin{equation*}
|\widetilde g(0)|^2
\lesssim
\|\widetilde g\|_{L^2(D_\varepsilon)}^2
+
[\widetilde g]_{\dot H^\beta(D_\varepsilon)}^2 ,
\end{equation*}
where \([\cdot]_{\dot H^\beta}\) denotes the homogeneous Sobolev seminorm. Scaling gives
\begin{equation*}
\|\widetilde g\|_{L^2(D_\varepsilon)}^2
=
\varepsilon^{-p}\|g\|_{L^2(B_\varepsilon)}^2
\le
\varepsilon^{-p}\|g\|_{L^2(B_r)}^2,
\end{equation*}
and
\begin{equation*}
[\widetilde g]_{\dot H^\beta(D_\varepsilon)}^2
\lesssim
\varepsilon^{2\beta-p}
[g]_{\dot H^\beta(B_\varepsilon)}^2
\le
\varepsilon^{2\beta-p}
\|g\|_{H^\beta(\cA)}^2 .
\end{equation*}
Since the Lipschitz-domain measure-density property gives \(|B_r|\asymp r^p\),
\begin{equation*}
\varepsilon^{-p}\|g\|_{L^2(B_r)}^2
=
\varepsilon^{-p}|B_r|
\|g\|_{L^2(\cU(B_r))}^2
\lesssim
\left(\frac r\varepsilon\right)^p
\|g\|_{L^2(\cU(B_r))}^2 .
\end{equation*}
Combining the last three displays proves \eqref{eq:scaled-local-sobolev}.

Now define
\[
Q=\phi(a_0)^\top(\bSigma_r+\lambda\Ib)^{-1}\phi(a_0),
\qquad
\theta=(\bSigma_r+\lambda\Ib)^{-1}\phi(a_0).
\]
By the reproducing property, the corresponding Sobolev function
\[
g_\theta(a)=\inner{\phi(a)}{\theta}_{\HH}
\]
satisfies
\[
g_\theta(a_0)=Q.
\]
Moreover,
\[
(\bSigma_r+\lambda\Ib)\theta=\phi(a_0).
\]
Taking the inner product with \(\theta\) gives
\begin{equation*}
Q
=
\inner{\theta}{\bSigma_r\theta}_{\HH}
+\lambda\|\theta\|_{\HH}^2
=
\|g_\theta\|_{L^2(\cU(B_r))}^2
+\lambda\|\theta\|_{\HH}^2 .
\end{equation*}
Using the equivalence between the RKHS norm and the \(H^\beta(\cA)\) norm, we obtain
\begin{equation}\label{eq:local-sobolev-energy}
\|g_\theta\|_{L^2(\cU(B_r))}^2
+
\lambda\|g_\theta\|_{H^\beta(\cA)}^2
\lesssim Q.
\end{equation}

Applying \eqref{eq:scaled-local-sobolev} to \(g_\theta\), and then using
\eqref{eq:local-sobolev-energy}, yields, for every \(0<\varepsilon\le r\),
\begin{equation*}
\begin{aligned}
Q^2
=
|g_\theta(a_0)|^2
&\lesssim
\left(\frac{r}{\varepsilon}\right)^p
\|g_\theta\|_{L^2(\cU(B_r))}^2
+
\varepsilon^{2\beta-p}
\|g_\theta\|_{H^\beta(\cA)}^2 \\
&\lesssim
\left\{
\left(\frac{r}{\varepsilon}\right)^p
+
\frac{\varepsilon^{2\beta-p}}{\lambda}
\right\}Q .
\end{aligned}
\end{equation*}
If \(Q=0\), the desired bound is trivial. Otherwise, dividing by \(Q\) gives
\begin{equation}\label{eq:Q-epsilon-bound}
Q
\lesssim
\left(\frac{r}{\varepsilon}\right)^p
+
\frac{\varepsilon^{2\beta-p}}{\lambda}.
\end{equation}

If \(\lambda\ge r^{2\beta-p}\), take \(\varepsilon=r\). Then
\[
Q
\lesssim
1+\frac{r^{2\beta-p}}{\lambda}
\lesssim 1.
\]
If \(\lambda<r^{2\beta-p}\), take
\[
\varepsilon=(\lambda r^p)^{1/(2\beta)}.
\]
Then
\[
\frac{\varepsilon}{r}
=
\left(\frac{\lambda}{r^{2\beta-p}}\right)^{1/(2\beta)}
<1,
\]
so this choice is admissible. Moreover,
\[
\left(\frac{r}{\varepsilon}\right)^p
=
\frac{\varepsilon^{2\beta-p}}{\lambda}
=
\left(\frac{r^{2\beta-p}}{\lambda}\right)^{\frac{p}{2\beta}}.
\]
Therefore \eqref{eq:Q-epsilon-bound} gives
\[
Q
\lesssim
\left(\frac{r^{2\beta-p}}{\lambda}\right)^{\frac{p}{2\beta}}.
\]
Combining the two cases proves
\[
\phi(a_0)^\top(\bSigma_r+\lambda\Ib)^{-1}\phi(a_0)
=
Q
\lesssim
1\vee
\left(\frac{r^{2\beta-p}}{\lambda}\right)^{\frac{p}{2\beta}}.
\]

It remains to prove the operator inequality. Let
\[
A:=\bSigma_r+\lambda\Ib .
\]
Since \(A\succ0\), the rank-one positive operator
\[
A^{-1/2}\phi(a_0)\phi(a_0)^\top A^{-1/2}
\]
has operator norm
\[
\begin{aligned}
\left\|
A^{-1/2}\phi(a_0)\phi(a_0)^\top A^{-1/2}
\right\|_{\operatorname{op}}
&=
\left\|A^{-1/2}\phi(a_0)\right\|_{\HH}^2 \\
&=
\phi(a_0)^\top A^{-1}\phi(a_0) \\
&= Q.
\end{aligned}
\]
Hence, in the Loewner order,
\[
A^{-1/2}\phi(a_0)\phi(a_0)^\top A^{-1/2}
\lesssim
\left[
1\vee
\left(\frac{r^{2\beta-p}}{\lambda}\right)^{\frac{p}{2\beta}}
\right]\Ib .
\]
Multiplying by \(A^{1/2}\) on the left and on the right gives
\[
\phi(a_0)\phi(a_0)^\top
\lesssim
\left[
1\vee
\left(\frac{r^{2\beta-p}}{\lambda}\right)^{\frac{p}{2\beta}}
\right]
(\bSigma_r+\lambda\Ib),
\]
as desired.
\hfill\BlackBox

\end{document}